\documentclass{article}
\usepackage[utf8]{inputenc}
\usepackage{amsmath, amssymb}

\usepackage{natbib}
\usepackage{xcolor} 
\usepackage{hyperref}

\usepackage{subcaption}
\usepackage{graphicx}
\graphicspath{{figures/}}
\usepackage{float}
\usepackage{amsthm}

\usepackage{geometry}
\usepackage{centernot}

\usepackage{mathtools}

\usepackage[linesnumbered,ruled,vlined]{algorithm2e}

\SetCommentSty{mycommfont}
\SetKwInput{KwInput}{Input}                
\SetKwInput{KwOutput}{Output}              

\newcommand{\vs}[1]{^{(#1)}}

\newcommand{\blckdiag}{D}

\newcommand{\clustprod}{K\vs{1} \times \dots \times K\vs{V}}
\newcommand{\sumoverclusters}{ \sum_{k\vs{1}=1}^{K\vs{1}} \dots  \sum_{k\vs{V}=1}^{K\vs{V}}}

\newtheorem{definition}{Definition}[section]
\newtheorem{remark}{Remark}[section]
\newtheorem{theorem}{Theorem}[section]
\newtheorem{lemma}{Lemma}[section]
\newtheorem{corollary}{Corollary}[section]
\newtheorem{proposition}{Proposition}[section]
\newtheorem{assumption}{Assumption}[section]

\title{Learning Sparsity and Block Diagonal Structure in Multi-View Mixture Models}
\author{Iain Carmichael\footnote{Department of Statistics, University of Washington, Seattle, WA}}
\date{\today}

\begin{document}

\maketitle

\begin{abstract}
Scientific studies increasingly collect multiple modalities of data to investigate a phenomenon from several perspectives.
In integrative data analysis it is important to understand how information is heterogeneously spread across these different data sources.
To this end, we consider a parametric clustering model for the subjects in a \textit{multi-view} data set (i.e. multiple sources of data from the same set of subjects) where each view marginally follows a mixture model.
In the case of two views, the dependence between them is captured by a \textit{cluster membership matrix} parameter and we aim to learn the structure of this matrix (e.g. the zero pattern).
First, we develop a penalized likelihood approach to estimate the sparsity pattern of the cluster membership matrix.
For the specific case of block diagonal structures, we develop a constrained likelihood formulation where this matrix is constrained to be \textit{block diagonal up to permutations} of the rows and columns.
To enforce block diagonal constraints we propose a novel optimization approach based on the symmetric graph Laplacian.
We demonstrate the performance of these methods through both simulations and applications to data sets from cancer genetics and neuroscience.
Both methods naturally extend to multiple views.
\end{abstract}



\section{Introduction} \label{s:intro}

Scientific studies often investigate a phenomenon from several perspectives by collecting multiple modalities of data.
For example, modern cancer studies collect data from several genomic platforms such as RNA expression, microRNA, DNA methylation and copy number variations \citep{cancer2012comprehensive, hoadley2018cell}.
Neuroscientists investigate neurons using transcriptomic, electrophysiological and morphological measurements \citep{tasic2018shared, gouwens2019classification, gouwens2020toward}.
These integrative studies require methods to analyze \textit{multi-view} data: a fixed set of observations with several disjoint sets of variables (views). 

Classical multi-view methods such as \textit{canonical correlation analysis} for dimensionality reduction estimate joint information shared by all views \citep{hotelling1936relations}.
Similarly, many multi-view clustering methods assume there is one \textit{consensus} clustering (see Figure \ref{fig:toy_data_joint} below) that is present in each data-view \citep{shen2009integrative, kumar2011co, kirk2012bayesian, lock2013bayesian, gabasova2017clusternomics, wang2019integrative}. 
A singular focus on joint signals ignores the possibility that information is heterogeneously spread across the views.
For example, environmental factors might show up in a clinical data-view, but not in a genomic data-view.
Contemporary multi-view methods examine how information is shared (or not shared) by different views.
Recent work in dimensionality reduction looks for \textit{partially shared} latent signals \citep{lock2013joint, klami2014group, zhao2016bayesian, gaynanova2017structural, feng2018angle}. 
Similarly, recent multi-view clustering methods investigate how clustering information is spread across multiple views \citep{hellton2016integrative, gao2019clusterings, gao2019testing}.

A motivating example comes from breast cancer pathology where investigators study tumors using both genomic and  histological\footnote{Meaning a doctor or algorithm visually examines an image of a tumor biopsy.} information \citep{carmichael2019joint}.
Breast cancer tumor subtypes can be defined using either genomic information (e.g. the PAM50 molecular subtypes \citealt{parker2009supervised}) or histological information (e.g. high, medium or low grade \citealt{hoda2020rosen}).
Some cluster information may be jointly shared by both data views e.g. if histological subtype 1 corresponds to exactly genomic subtype 1. 
Other information may be contained in one view but not another view e.g. if histological subtype 2 correspond to genomic subtypes 2 and 3.
See Figure \ref{fig:motiv_ex_pi} below.

We develop an approach to learn how information is spread across views in a multi-view mixture model (MVMM) \citep{bickel2004multi}.
This model, detailed in Section \ref{s:mvmm}, makes two assumptions for a $V \ge 2$ view data set:
\begin{enumerate}
\item Marginally, each view follows a mixture model i.e. there are $V$ sets of \textit{view-specific clusters}.
\item The views are independent given the marginal view cluster memberships.
\end{enumerate}
We further assume there may be some kind of ``interesting relationship" between clusters in different views.
For example, in a two-view data set every observation has two (hidden) cluster labels $(y\vs{1}, y\vs{2}) \in [K\vs{1}] \times [K\vs{2}]$ where $K\vs{v}$ is the number of clusters in the $v$th view and $[K] := \{1, \dots, K\}$.
The joint distribution of the cluster labels is described by the \textit{cluster membership probability matrix} $\pi \in \mathbb{R}^{K\vs{1} \times K\vs{2}}_+$ where
\begin{equation*}
\pi_{k\vs{1}, k\vs{2}} = P(y\vs{1}=k\vs{1}, y\vs{2}=k\vs{2}), \text{ for } k\vs{1} \in [K\vs{1}] \text{ and } k\vs{2} \in [K\vs{2}].
\end{equation*}

The structure of this matrix captures how information is shared between the two views.
Figure \ref{fig:motiv_ex_pi} shows a hypothetical $\pi$ matrix.
Many of the entries are zero, meaning, for instance, an observation cannot be simultaneously in cluster 1 in the first view and cluster 2 in the second view.
In this example, cluster 1 in the first view is exactly the same as cluster 1 in the second view; this information is shared jointly by both views.
On the other hand, cluster 3 in the second view breaks up into clusters 3, 4 and 5 in the first view; here there is information in the first view that is not contained in the second view.
In general $\pi$ may be anywhere from rank 1 (i.e. the views are independent thus share no information) to diagonal (the consensus clustering case where the views contain the same information).
The goal of this paper is to learn the structure of $\pi$ while simultaneously learning the cluster parameters (e.g. cluster means).

\begin{figure}[H]
\centering
\begin{subfigure}[t]{0.49\textwidth}
\centering
\includegraphics[width=.8\linewidth, height=.8\linewidth]{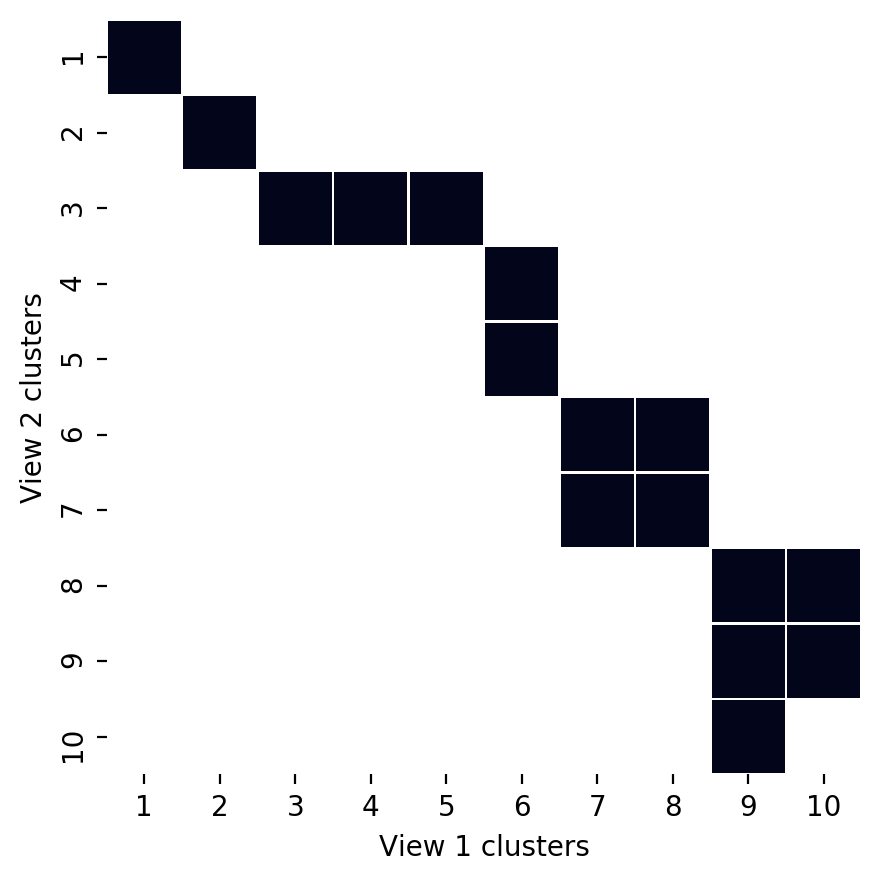}
\caption{
An example block diagonal $\pi$ matrix.
}
\label{fig:motiv_ex_pi}
\end{subfigure}
\begin{subfigure}[t]{0.49\textwidth}
\centering
\includegraphics[width=.48\linewidth, height=.8\linewidth]{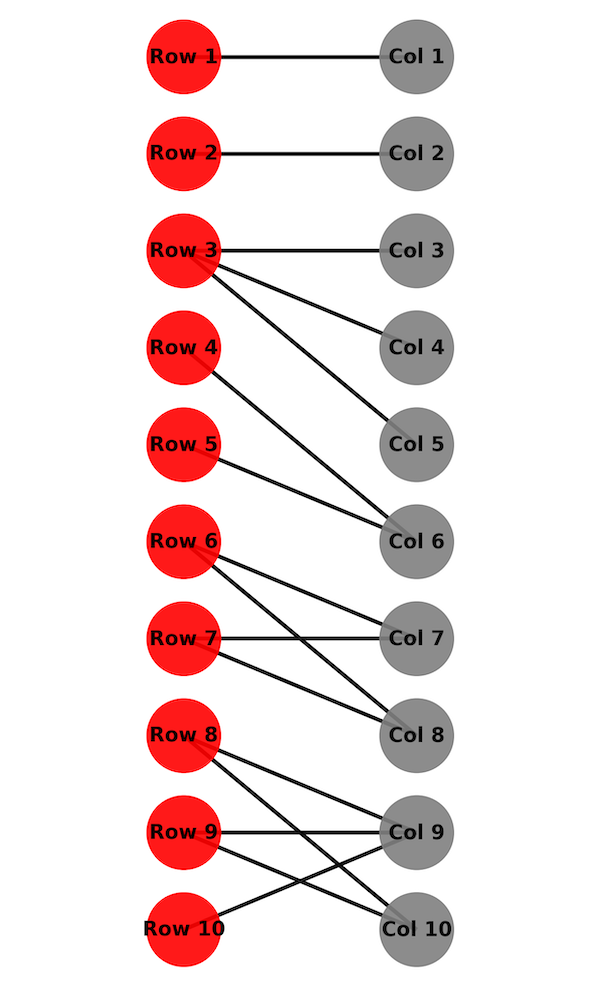}
\caption{
The bipartite graph whose node sets are the rows and columns of $\pi$ and whose edges the entries of $\pi$.
}
\label{fig:motiv_ex_pi_bipt_graph}
\end{subfigure}
\caption{
The matrix $\pi$ (Figure \ref{fig:motiv_ex_pi}) captures the between-view dependence.
$\pi$ can equivalently be thought of as a bipartite graph (Figure \ref{fig:motiv_ex_pi_bipt_graph}).
The connected components of this graph represent ``clusters of clusters" that are related to each other.
Note $\pi$ can be rectangular in general.
}
\label{fig:motiv_ex_pi_mats}
\end{figure}

Section \ref{s:mvmm} formalizes the multi-view mixture model outlined above.
Section \ref{s:sparsity_log_pen} presents a penalized likelihood approach making use of the concave $\log(\cdot + \delta)$  penalty to estimate the zero-pattern of $\pi$.
Section \ref{s:block_diag} considers the case when $\pi$ has block diagonal structure and formulates a block diagonally constrained maximum likelihood version of the MVMM.
This section develops an alternating minimization approach for imposing block diagonal matrix constraints in general via the symmetric Laplacian.
A detailed discussion of this alternating algorithm and convergence results are provided in Section \ref{s:wnn}.
An extension of this approach to block diagonal multi-arrays is sketched in Section \ref{s:multiarray}.
Section  \ref{s:simulations} presents a simulation study of the methods developed in this paper.
Section \ref{s:real_data} applies these methods to the TCGA breast cancer data set and an excitatory mouse neuron data set. 
The main algorithmic ideas are presented in the body of the paper and detailed discussions are provided in the appendix. 
Proofs and additional simulations are also provided in the appendix.

The methods developed in this paper are implemented in a publicly available python package \url{www.github.com/idc9/mvmm}.
Code to reproduce the simulations as well as supplementary data and figures can be found at \url{www.github.com/idc9/mvmm_sim}.
The code makes use of the following python packages: \cite{hunter2007matplotlib, mckinney2010data, walt2011numpy, pedregosa2011scikit, diamond2016cvxpy, waskom2017seaborn, cameron2020lifelines, scipy2020virtanen}.

\subsection{Summary of contributions and related work}

We develop two novel methods that explore how information is shared between views in the parametric multi-view mixture model of \cite{bickel2004multi}.
Both methods impose interpretable structure  --- sparsity (Section \ref{s:sparsity_log_pen}) or block diagonal constraints (Section  \ref{ss:mvmm_block_diag}) --- on the cluster membership matrix.
They also lead to challenging optimization issues.
Our approaches to address these challenges are of interest in applications beyond this paper.

Many existing multi-view clustering methods focus on the consensus clustering case (see reference above).
While the consensus clustering case is a special case of the MVMM when $\pi$ is diagonal, our method allows for more flexible relations among the clusters in each view.
The work of \cite{gao2019clusterings, gao2019testing} takes an important step beyond consensus clustering by developing a test for independence between the views in a two-view MVMM.

The \textit{joint and individual clustering} (JIC) method developed by \cite{hellton2016integrative} is a multi-view clustering algorithm based on dimensionality reduction using JIVE \citep{lock2013joint}.
JIC identifies information that is either shared by all views (joint clusters) or is only contained in one view (individual clusters).
An immediate difference between JIC and our methods is that we work with parametric mixture models while JIC is based on dimensionality reduction.
Moreover, our methods take a different perspective on how information can be shared among views (see Footnote \ref{fn:jic}).

The $\log(\cdot + \delta)$ penalized likelihood approach adopted in Section \ref{s:sparsity_log_pen} was developed in \cite{huang2017model} for (single view) mixture-model model selection.
To fit mixture models with this penalty \cite{huang2017model} suggests an EM algorithm where the M-step is approximated with a soft-thresholding operation.
This soft-thresholding approximation --- based on a heuristic argument --- is used by a number of other papers \citep{yao2018robust,  yu2019new, bugdary2019online} and similar approximations appear elsewhere \citep{gu2019learning}.
We provide rigorous justification for this soft-thresholding approximation and show the algorithm is insensitive to the choice of $\delta$ for small values of $\delta$ (Theorem \ref{thm:approx_soft}).

The task of learning model parameters with block diagonal structure arises in a variety of contexts including: graphical models \citep{marlin2009sparse, tan2015cluster, devijver2018block, kumar2019unified}, co-clustering \citep{han2017bilateral, nie2017learning}, subspace clustering \citep{feng2014robust, lu2018subspace}, principal components analysis \citep{asteris2015sparse}, and community detection \citep{nie2016constrained}.
Learning parameter values and block diagonal structure simultaneously is a combinatorial problem that is generally intractable except in certain special cases \citep{asteris2015sparse}.
Block diagonal constraints are often enforced with continuous optimization approaches using the unnormalized graph Laplacian \citep{nie2016constrained, nie2017learning}.

Sections \ref{s:block_diag} and \ref{s:wnn} develop an approach to impose block diagonal constraints via the symmetric graph Laplacian.
This approach avoids the strong modeling assumptions --- that the row and column sums are known ahead of time --- required by the unnormalized Laplacian \citep{nie2016constrained, nie2017learning}.
By making use of an extremal characterization of generalized eigenvalues we provide an alternating algorithm for the penalized symmetric Laplacian Problem \eqref{prob:bd_eval_pen} that is no more computationally burdensome than the analogous problem with the unnormalized Laplacian (see Section \ref{s:un_lap_bad}).

\subsection{Notation} 

A multi-view random vector $x \in \mathbb{R}^{\sum_{v=1}^V d\vs{v}}$ is a random vector where the variables have been partitioned into $V$ mutually exclusive sets of sizes $d\vs{1}, \dots, d\vs{V}$.
We write $x\vs{v} \in \mathbb{R}^{d\vs{v}}$ for the $v$th view i.e. $x$ is the concatenation of the $x\vs{1}, \dots, x\vs{V}$. 
We use superscript parenthesis, e.g. $x\vs{v}$, to reference quantities related to a particular view.

For a matrix $V \in \mathbb{R}^{R \times C}$ let $V(r, :) \in \mathbb{R}^C$ denote the $r$th row and let $V(:, c) = V_c \in \mathbb{R}^R$ denote the $c$th column.
For $v \in \mathbb{R}^n$, let $\text{diag}(v) \in \mathbb{R}^{n \times n}$ be the diagonal matrix whose diagonal elements are given by $v$.
Let $\mathbf{1}_n \in \mathbb{R}^n$ be the vector of ones.
For a set $A \subseteq [n]$ let $\mathbf{1}_A \in \{0, 1\}^n$ denote the vector with 1s in the entries corresponding to elements of $A$ and 0s elsewhere.
The indicator function, $I: \mathbb{R}^n \to \{0\} \cup \infty$, of a set $\mathcal{C} \subseteq \mathbb{R}^n$ is defined by  $I(x) = 0$ if $x \in \mathcal{C}$ and $I(x) = \infty$ if $x \not\in \mathcal{C}$.
For vectors $a, b$ let $a \odot b$ denote the Haadamard (element-wise) product.

For a symmetric matrix $A \in \mathbb{R}^{n \times n}$ we write $\lambda_1(A) \ge \lambda_2(A) \ge \dots$ for the eigenvalues sorted in decreasing order and  $\lambda_{(1)}(A) \le \lambda_{(2)}(A) \le \dots$ for the eigenvalues sorted in increasing order.
For two symmetric matrices $A, B \in \mathbb{R}^{n \times n}$ we write $\lambda_{1}(A, B) \ge \lambda_{2}(A, B) \ge \dots$ for the \textit{generalized eigenvalues} (i.e. numbers $\lambda$ where there exists a $v \in \mathbb{R}^{n}$ such that $A v =  \lambda B v$ with the normalization $v^T B v = 1$).


\section{Multi-view mixture model specification} \label{s:mvmm} 

This section describes a multi-view mixture model for $V \ge 2$ views \citep{bickel2004multi, gao2019clusterings}.
This model assumes that marginally, each view follows a mixture model and that the views are conditionally independent given the view cluster memberships.

In detail, let $x\vs{v} \in \mathbb{R}^{d\vs{v}}$ denote the random vector for the $v$th view. 
In the $v$th view there are $K\vs{v}$ view-specific clusters and let $y\vs{v} \in [K\vs{v}]$ denote latent, view specific membership assignment for the $v$th view.
Let
\begin{equation*}
f(x\vs{v} | y\vs{v} =k) = \phi\vs{v}(x | \Theta\vs{v}_k) \text{ for }  k \in [K\vs{v}], v \in [V],
\end{equation*}
be the conditional distribution of the $k$th cluster in the $v$th view where $ \phi\vs{v}(\cdot | \theta)$ is a density function with parameter $\theta$ (e.g. cluster means).
Also let
\begin{equation*}
P\left(y = (k\vs{1}, \dots, k\vs{V}) \right) = \pi_{k\vs{1}, \dots, k\vs{V}}  \text{ for } k\vs{v} \in [K\vs{v}], v \in [V],
\end{equation*}
be the joint distribution of the view specific labels where $y = (y\vs{1}, \dots, y\vs{V})\in \mathbb{Z}^V_+$ is the latent cluster membership vector and $\pi \in \mathbb{R}^{\clustprod}$ is the cluster membership probability multi-array (non-negative entries summing to 1). 
Then the probability density function of the joint distribution is
\begin{equation}\label{eq:mvmm_joint_pdf}
f(x, y = (k\vs{1}, \dots, k\vs{V}) | \Theta, \pi ) =   \pi_{k\vs{1}, \dots, k\vs{V}}  \prod_{v=1}^V  \phi\vs{v}(x\vs{v} | \Theta\vs{v}_{k\vs{v}}),
\end{equation}
where $\Theta:= \{ \{\Theta\vs{v}_{k} \}_{k=1}^{K\vs{v}}  \}_{v=1}^V$ is the collection of view specific cluster parameters.
We further assume that the marginal view-specific cluster probabilities are strictly positive, i.e.
\begin{equation}
0 <  \pi\vs{v}_{\textbf{k}} := P(y\vs{v} = \textbf{k}) = \sum_{j\vs{1}=1}^{K\vs{1}} \dots  \sum_{j\vs{v-1}=1}^{K\vs{v-1}} \sum_{j\vs{v+1}=1}^{K\vs{v+1}}  \dots \sum_{j\vs{V}=1}^{K\vs{V}}  \pi_{j\vs{1}, \dots, j\vs{v-1}, \textbf{k}, j\vs{v+1} \dots j\vs{V}}
\end{equation}
for each $k \in [K\vs{v}], \text{ and } v \in [V]$. 

The marginal distribution of the $v$th view, $x\vs{v}$, is a mixture model with $K\vs{v}$ view-specific clusters (Figure \ref{fig:toy_data_view_marginal}).
The joint distribution, $x$, is a mixture model with $|\text{supp}(\pi)| \in [\min_{v \in [V]} (K\vs{v}), \prod_{v=1}^VK\vs{v}]$ \textit{overall clusters} (Figures \ref{fig:toy_data_indep}-\ref{fig:toy_data_joint}).
In other words, looking at the joint distribution there is one set\footnote{\label{fn:jic} In the JIC model the view joint distribution has $V + 1$ sets of clusters for a $V$-view data set --- one set of joint clusters and $V$ sets of view-individual clusters. For details see \citep{hellton2016integrative}.} of $|\text{supp}(\pi)|$ clusters, but the clusters share parameters.

\begin{remark}
This model promotes parameter sharing; if $\pi$ is dense, the number of overall clusters scales multiplicatively (e.g. like $O(K^V)$) in the number of view  marginal clusters while the number of cluster parameters (e.g. cluster means) scaled additively (e.g. like $O(VK)$).
\end{remark}

Figure \ref{fig:toy_data} shows three scenarios for a $V=2$ view data set. 
Both views are one dimensional and marginally follow a \textit{Gaussian mixture model} (GMM) with $K\vs{1} = K\vs{2}=10$ clusters (Figure \ref{fig:toy_data_view_marginal}).
In the first scenario (Figure  \ref{fig:toy_data_indep}) there is no information shared between the two views; $\pi$ is a rank 1 matrix.
In the third scenario (Figure \ref{fig:toy_data_joint}) the two views capture the same information i.e. the clusters in the first view are the same clusters as the clusters in the second view.
Here $\pi$ is a diagonal matrix (after appropriately permuting the cluster labels).
In the second scenario (Figure  \ref{fig:toy_data_partial}) the two views have partially overlapping information.
In this scenario $\pi$ is the block diagonal matrix shown in Figure \ref{fig:motiv_ex_pi} above.

\begin{figure}[H]
 \centering
 \begin{subfigure}[t]{0.24\textwidth}
\includegraphics[width=\linewidth, height=\linewidth]{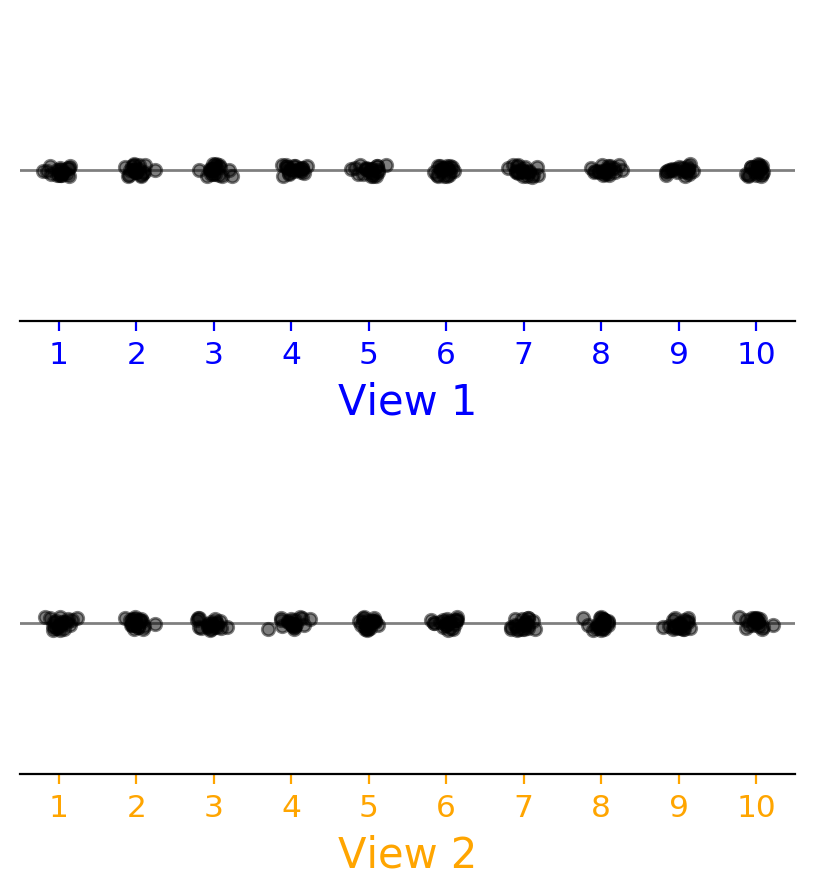}
\caption{
Marginally, each view follows a mixture model with $10$ clusters.
}
\label{fig:toy_data_view_marginal}
\end{subfigure}
\hfill
\begin{subfigure}[t]{0.24\textwidth}
\includegraphics[width=\linewidth, height=\linewidth]{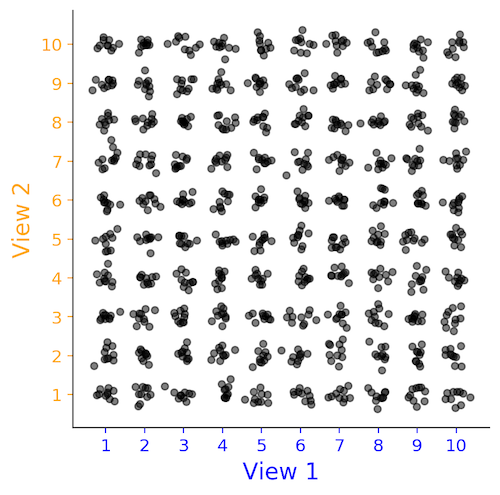}
\caption{
Scenario 1: Independent views; $100$ overall clusters.
}
\label{fig:toy_data_indep}
\end{subfigure}
\hfill
\begin{subfigure}[t]{0.24\textwidth}
\includegraphics[width=\linewidth, height=\linewidth]{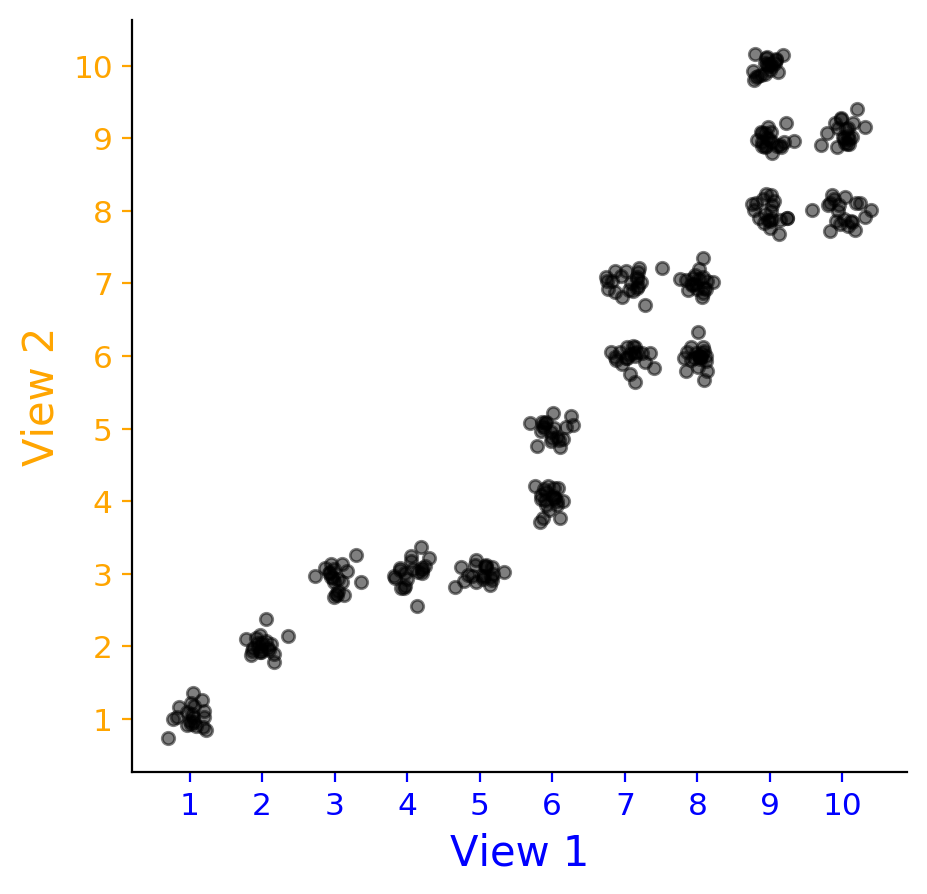}
\caption{
Scenario 2: Partial dependence between views; $16$ overall clusters.
}
\label{fig:toy_data_partial}
\end{subfigure}
\hfill
\begin{subfigure}[t]{0.24\textwidth}
\includegraphics[width=\linewidth, height=\linewidth]{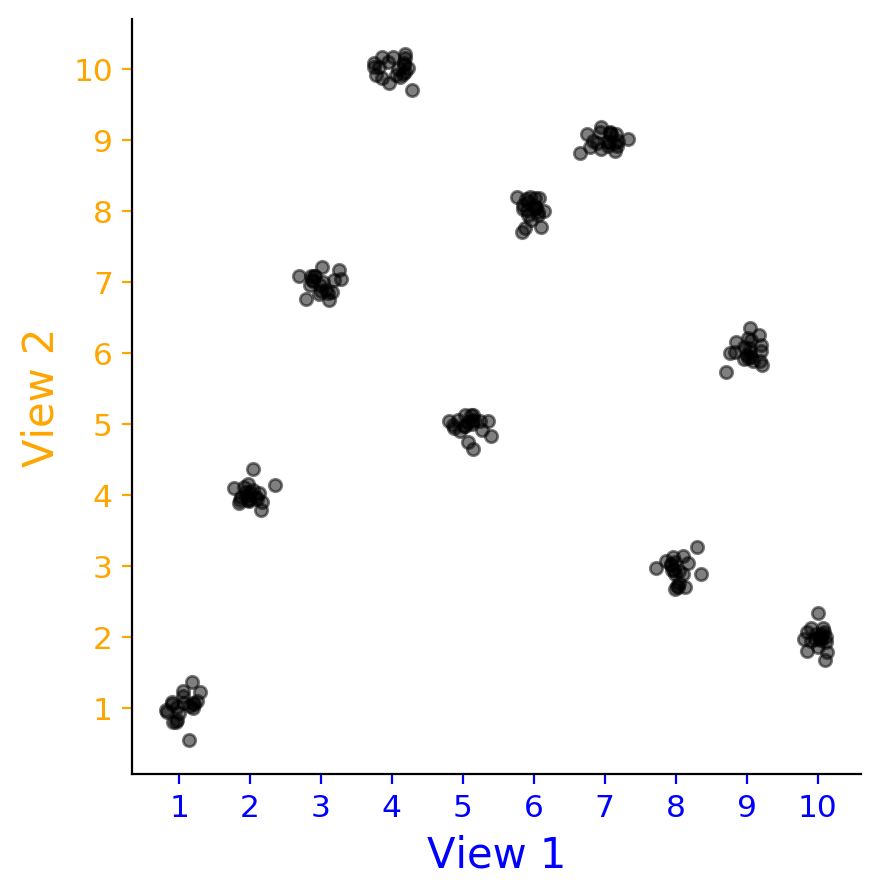}
\caption{
Scenario 3: Strong dependence between views; $10$ overall clusters.
This is the consensus clustering case. 
}
\label{fig:toy_data_joint}
\end{subfigure}
\caption{
Three scenarios for the joint distribution of two views.
The marginal distribution of each view is a one dimensional mixture model (Figure \ref{fig:toy_data_view_marginal}).
In Figure \ref{fig:toy_data_indep} all possible combinations of view 1 and view 2 clusters occur with equal probability.
In Figure \ref{fig:toy_data_partial} only some combinations of view 1 and view 2 clusters are possible.
In Figure \ref{fig:toy_data_joint} each cluster in the first view is matched with exactly one cluster in the second view.
}
\label{fig:toy_data}
\end{figure}

Suppose we are given $n$ samples $\{x_i \}_{i=1}^n$ with $x_i \in \mathbb{R}^{\sum_{v=1}^V d\vs{v}}$ from a $V$-view data set and have specified the number of view-specific clusters $K\vs{1}, \dots, K\vs{V}$.
If no additional assumptions are placed on $\pi$, we fit the model by maximizing the log likelihood of the observed data 
 \begin{equation}\label{eq:data_log_lik}
 \ell( \{x_i \}_{i=1}^n| \Theta, \pi) := \sum_{i=1}^n \log f(x_i | \Theta, \pi)
 \end{equation}
using an EM algorithm \citep{dempster1977maximum} that is detailed in Section \ref{ss:mvmm_em}, where
 \begin{equation}\label{eq:obs_data_pdf}
f(x | \Theta, \pi) :=  \sumoverclusters \pi_{k\vs{1}, \dots, k\vs{V}}  \prod_{v=1}^V  \phi\vs{v}(x\vs{v} | \Theta\vs{v}_{k\vs{v}})
\end{equation}
is the probability density function of the observed data.
The remainder of this paper focuses on simultaneously estimating the model parameters, $ \Theta, \pi$, as well as the sparsity structure of $\pi$.


\section{Sparsity inducing log penalty} \label{s:sparsity_log_pen}

This section develops a penalized likelihood approach to estimate the sparsity structure of $\pi$ that avoids the exponential search space of naive enumeration.
We assume  the number of view specific clusters, $K\vs{1}, \dots, K\vs{V}$, have been specified.

Consider fitting a standard, single-view mixture model with a sparsity inducing penalty $p(\cdot)$ (e.g. Lasso or SCAD) on the entries of the cluster membership probability vector, $\pi \in \mathbb{R}^{K}_+$. 
This raises several issues.
First, the Lasso penalty is constant since $\pi$ lives on the unit simplex.
Second, exact zeros in $\pi$ give a negative infinity the complete data log likelihood \eqref{eq:mvmm_joint_pdf}, though this issue does not arise in the observed data log-likelihood \eqref{eq:data_log_lik}.
If we use an EM algorithm to maximize the observed data log-likelihood the M-step involves the following optimization problem 
\begin{equation} \label{eq:single_view_m_step}
\begin{aligned}
& \underset{\pi \in \mathbb{R}^K}{\text{minimize}} & & - \sum_{k=1}^K a_k  \log(\pi_k) + \lambda \sum_{k=1}^K p(\pi_k) \\
& \text{subject to } & &  \pi \ge 0 \text{ and } \pi^T \mathbf{1}_K = 1,
\end{aligned}
\end{equation}
where $a \in \mathbb{R}^K_+$ is the output of the E-step (i.e. the expected cluster assignments).
The log in the first term of the objective function acts as a barrier function that prevents the solution from having zeros.

\citet{huang2017model} provides theoretical justification for using the penalty $p(\cdot) = \log(\delta + \cdot)$ for some small $\delta > 0$. 
The following theorem further justifies the use of this penalty for small $\delta$ by showing that we can approximate the solution with a quantity that has exact zeros.
This theorem also leads to a computationally efficient approximation for the M-step and suggests that the penalty is insensitive to the choice of $\delta$ for small values of $\delta$.

\begin{theorem} \label{thm:approx_soft}
Let $a_1, \dots, a_K \ge 0$, $\sum_{k=1}^K a_k =1$, and $0 < \lambda < \frac{1}{K}$.
Let  $z^{\delta} \in \mathbb{R}_{+}^K$ be a solution of the following problem for fixed $\delta > 0$,
\begin{equation}\label{eq:log_sparsity_problem_with_epsilon}
\begin{aligned}
& \underset{z}{\textup{minimize}}  & & - \sum_{k=1}^K a_k  \log(z_k)  + \lambda \sum_{k=1}^K  \log(\delta + z_k)  \\
& \textup{subject to} & & z \ge 0 \text{ and } z^T \mathbf{1}_K = 1.\\
\end{aligned}
\end{equation}
Then $\lim_{\delta \to 0} z^{\delta} = z^0  \in \mathbb{R}^{K}$ where
\begin{equation} \label{eq:log_sparsity_problem_pi}
z_k^{0} := \frac{(a_k - \lambda)_+}{\sum_{j=1}^K (a_j - \lambda)_+}  \text{ for each } k \in [K].
\end{equation}

\end{theorem}
This theorem says that for small $\delta$ the global minimizer of \eqref{eq:log_sparsity_problem_with_epsilon} is close to the normalized soft-thresholding operation \eqref{eq:log_sparsity_problem_pi}.
The condition $\lambda < \frac{1}{K}$ guarantees the denominator of \eqref{eq:log_sparsity_problem_pi} is non-zero. 
The soft-thresholding approximation presented in this theorem is proposed by \cite{huang2017model} and used as a heuristic \citep{yao2018robust,  yu2019new, bugdary2019online}; we prove Theorem \ref{thm:approx_soft} in Section \ref{ss:proof_log_pen}.

Returning to the MVMM, we consider the following penalized likelihood problem
\begin{equation} \label{eq:log_pen_max_log_lik}
\begin{aligned}
&\underset{\Theta, \pi}{\text{maximize}}& &  \ell( \{x_i\}_{i=1}^n|  \Theta, \pi)- \lambda \left(\sumoverclusters \log(\delta + \pi_{k\vs{1}, \dots, k\vs{V}}) \right),
\end{aligned}
\end{equation}
where $\ell$ is the observed data log likelihood \eqref{eq:data_log_lik} and $\delta >0$ is a small value.
This problem can be solved with an EM algorithm similar to the one derived for the unpenalized model.
The M-step of this EM algorithm solves a problem in the form of \eqref{eq:log_sparsity_problem_with_epsilon}. Based on Theorem \ref{thm:approx_soft} we approximate the M-step using the normalized soft-thresholding operation.
Details for this algorithm can be found in Section \ref{ss:log_pen_em}.


\section{Enforcing block diagonal constraints}\label{s:block_diag}

This section presents a constrained maximum likelihood approach to estimate $\pi$ under the restriction that $\pi$ has a block diagonal structure.
Sections \ref{ss:sym_lap} and \ref{ss:block_diag_opt} discuss optimization with block diagonal constraints in a general setting.
Section \ref{ss:mvmm_block_diag} presents the particular case of the multi-view mixture model.

For a fixed matrix we have to be careful about what ``block diagonal" means i.e. one could argue that the matrix $\text{diag}([1, 1, 0])$ has either 1, 2, or 3 blocks.
We take the convention that blocks must have at least one non-zero entry and anything that can be a block is a block; thus $\text{diag}([1, 1, 0])$ has 2 blocks.
For a matrix $X$ whose rows/columns are allowed to be permuted we say ``$X$ is block diagonal with $NB(X)$ blocks up to permutations" where
\begin{equation} \label{eq:num_blocks}
\text{NB}(X) :=\max \{B | \text{the rows/columns of } X \text{ can be permuted to create a } B \text{ block, block diagonal matrix}\}
\end{equation}
Any permutation of the rows/columns of $X$ which achieves the above maximum is called a \textit{maximally block diagonal permutation}.

\subsection{Spectral characterization of block diagonal matrices} \label{ss:sym_lap}

This section gives a spectral characterization of block diagonal matrices up to permutations.
Let $A \in \mathbb{R}^{n \times n}_+$ be the adjacency matrix of a weighted, undirected graph with no self loops.
The \textit{unnormalized Laplacian} is
\begin{equation}
L_{\text{un}}(A): = \text{diag}(\text{deg}(A)) - A
\end{equation}
where $\text{deg}(A) := A\mathbf{1}_n \in \mathbb{R}^n_+$ is the vector of the vertex degrees  \citep{von2007tutorial}.
The \textit{symmetric, normalized Laplacian} is
\begin{equation}\label{eq:sym_lap_def}
 L_{\text{sym}}(A) := I -  \text{diag}(\text{deg}(A) )^{-1/2} A  \text{ diag}(\text{deg}(A))^{-1/2}.  
\end{equation}
When $\text{deg}(A)$ has zeros, the inverse is taken to be the Moore-Penrose psueo-inverse thus the diagonal elements of $ L_{\text{sym}}(\cdot)$ are always equal to 1 even when there are degree zero (isolated) vertices.\footnote{This convention is not always followed \citep{von2007tutorial}, as discussed in Section \ref{ss:proof_sym_lap_spect_bd}.}
The eigenvalues of the symmetric Laplacian are equal\footnote{We have to be careful when $ \text{diag}(\text{deg}(A))$ is non-invertible; this issue is addressed in Section \ref{s:extremal}.} to the generalized eigenvalues of $(L_{\text{un}}(A), \text{diag}(\text{deg}(A))$.

For $X \in \mathbb{R}^{R \times C}_+$ let
$$A_{\text{bp}}(X)  := \begin{bmatrix}  0 & X \\ X^T & 0 \end{bmatrix}\in  \mathbb{R}^{(R + C)\times (R + C)} $$
be the adjacency matrix of the weighted, bipartite graph $G(X)$ whose edge weights are given by the entries of $X$ and whose vertex sets are the rows and columns of $X$ (see Figure \ref{fig:motiv_ex_pi_bipt_graph}).
Note a row or column of zeros in $X$ corresponds to an isolated vertex in the graph. 

Proposition \ref{prop:sym_lap_block_diag} shows the connected components of this bipartite graph with at least two vertices capture the block diagonal structure of $X$ up to permutations; these connected components are in turn captured by the spectrum of the symmetric, normalized Laplacian.

\begin{proposition} \label{prop:sym_lap_block_diag}
The following are equivalent for $1\le B + Z_{\text{row}} + Z_{\text{col}} \le \min(R, C)$
\begin{enumerate}

\item $X$ is block diagonal up to permutations with $B$ blocks and has $ Z_{\text{row}}$ rows and $Z_{\text{col}}$ columns of zeros.

\item $G(X)$ has $B$ connected components with at least two vertices and $Z_{\text{row}} + Z_{\text{col}}$ isolated vertices.

\item $L_{\text{sym}}(A_{\text{bp}}(X))$ has exactly $B$ eigenvalues equal to 0.

\item $L_{\text{un}}(A_{\text{bp}}(X))$ has exactly $B + Z_{\text{row}} + Z_{\text{col}} $ eigenvalues equal to 0.

\end{enumerate}
Additionally, the number of eigenvalues equal to 1 of the symmetric Laplacian is at least $2\cdot(Z_{\text{row}} +Z_{\text{col}})$.

\end{proposition}
Section \ref{s:multiarray} generalizes this proposition to block diagonal multi-arrays.

Proposition \ref{prop:sym_lap_block_diag} shows that the symmetric Laplacian gives more precise control over the block diagonal structure of a matrix than the unnormalized Laplacian does.
The number of 0 eigenvalues of the symmetric Laplacian is exactly the number of blocks while the number of zero eigenvalues of the unnormalized Laplacian only upper bounds the number of blocks (see Figure  \ref{fig:lap_spect_compare_mat}).

\begin{figure}[H]
 \centering
\begin{subfigure}[t]{0.32\textwidth}
\includegraphics[width=.95\linewidth, height=\linewidth]{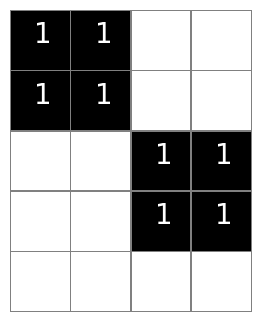}
\caption{
$X \in \{0,1\}_+^{5 \times 4}$ is block diagonal with two blocks and one row of zeros.
}
\label{fig:lap_spect_compare_mat}
\end{subfigure}
\hfill
\begin{subfigure}[t]{0.32\textwidth}
\includegraphics[width=\linewidth, height=\linewidth]{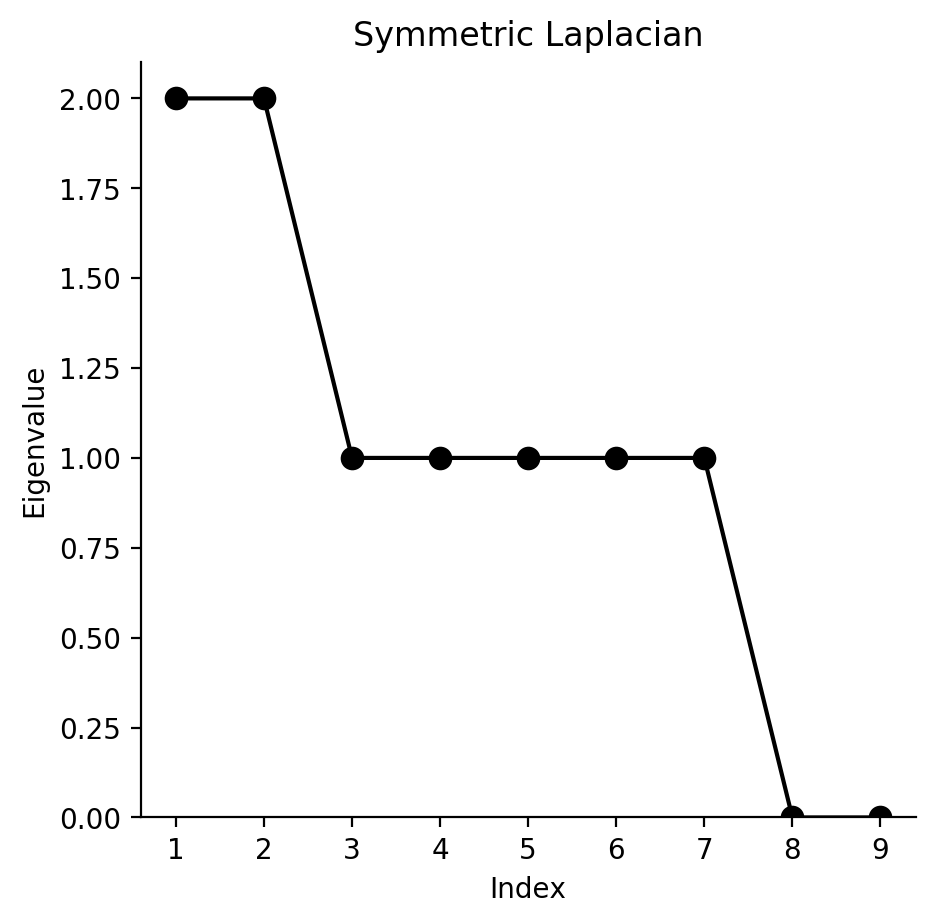}
\caption{
Spectrum of $L_{\text{sym}}(A_{\text{bp}}(X))$; two eigenvalues are equal to zero.
}
\label{fig:lap_spect_compare_mat_sym_evals}
\end{subfigure}
\hfill
\begin{subfigure}[t]{0.32\textwidth}
\includegraphics[width=\linewidth, height=\linewidth]{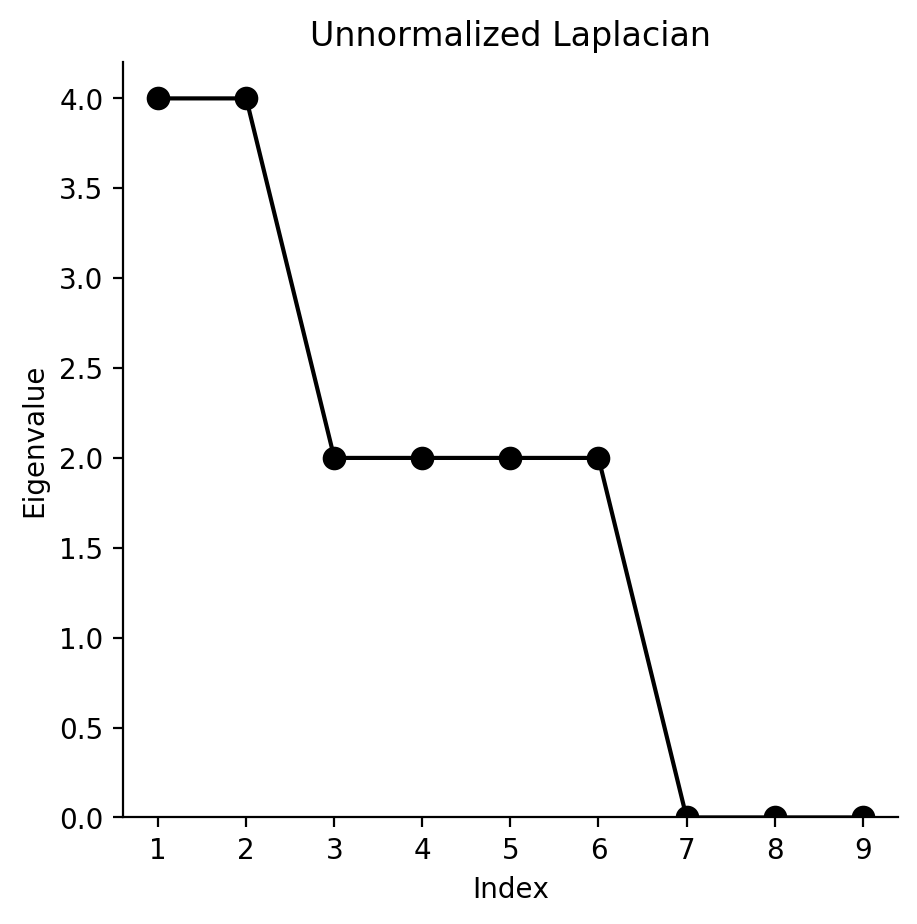}
\caption{
Spectrum of $L_{\text{un}}(A_{\text{bp}}(X))$; three eigenvalues are equal to zero.
}
\label{fig:lap_spect_compare_mat_un_evals}
\end{subfigure}
\caption{
The symmetric Laplacian's spectrum counts the blocks of a matrix up to permutations; the spectrum of the unnormalized Laplacian counts both blocks and zero rows/columns.
}
\label{fig:lap_spect_compare_mat}
\end{figure}

\subsection{Optimization with block diagonal constraints} \label{ss:block_diag_opt}

This section considers the following block diagonally constrained optimization problem 
\begin{equation} \label{prob:bd_constr}
\begin{aligned}
&\underset{X\in \mathbb{R}^{R \times C}}{\text{minimize}}   & & f(X)\\
& \text{subject to} & &  X \ge 0 \text{ and } X \text{ is block diagonal with at least } B \text{ blocks up to permutations}, \\
\end{aligned}
\end{equation}
where $f: \mathbb{R}^{R \times C} \to  \mathbb{R}$.
The naive approach to solving this problem involves iterating over all possible sparsity patterns with at least $B$ blocks up to permutations and is likely computationally infeasible.
Based on Proposition \ref{prop:sym_lap_block_diag}, we see Problem \eqref{prob:bd_constr} is equivalent to
\begin{equation} \label{prob:bd_eval_constr} 
\begin{aligned}
&\underset{X\in \mathbb{R}^{R \times C}}{\text{minimize}}   & & f(X)\\
& \text{subject to} & &  X \ge 0 \text{ and } L_{\text{sym}}(A_{\text{bp}}(X)) \text{ has at least } B \text{ eigenvalues equal to } 0. \\
\end{aligned}
\end{equation}
To impose the rank constraint, we consider the following related problem
\begin{equation}\label{prob:bd_eval_pen} 
\begin{aligned}
& \underset{X\in \mathbb{R}^{R \times C}}{\text{minimize}}   & & f(X) +   \alpha \sum_{j=1}^B  \lambda_{(j)} \left( L_{\text{sym}}(A_{\text{bp}}(X))\right)   \\
& \text{subject to} & &  X \ge 0,
\end{aligned}
\end{equation}
for a sufficiently large value of $\alpha$.
The non-linearity in $L_{\text{sym}}(\cdot)$ makes this problem computationally challenging.
We can replace this nonlinearity with linear terms using a variational characterization of generalized eigenvalues (Proposition \ref{prop:weighted_geval_extremal_representation} and Corollary \ref{cor:sym_evals_and_gevals}) to obtain,
\begin{equation} \label{prob:bd_extremal_rep}
\begin{aligned}
& \underset{X\in \mathbb{R}^{R \times C}, U \in \mathbb{R}^{(R + C) \times B}}{\text{minimize}}   & & f(X) +   \alpha \text{Tr} \left( U^T L_{\text{un}}(A_{\text{bp}}(X)) U\right)   \\
& \text{subject to} & & X \ge 0\\ 
&&&  U^T \text{diag}(\text{deg}(A_{\text{bp}}(X))) U = I_B,
\end{aligned}
\end{equation}
which typically has the same minimizers as \eqref{prob:bd_eval_pen}. 

\begin{proposition} \label{prop:soln_relations_global}
Problems \eqref{prob:bd_constr} and \eqref{prob:bd_eval_constr} are equivalent.
If $(X, U)$ is a global minimizer of  \eqref{prob:bd_extremal_rep} such that $\sum_{j=1}^B  \lambda_{(j)} \left( L_{\text{sym}}(A_{\text{bp}}(X))\right) = 0$, then $X$ is a global minimizer of \eqref{prob:bd_constr}, \eqref{prob:bd_eval_constr} and \eqref{prob:bd_eval_pen}.
\end{proposition}
Proposition \ref{prop:soln_relations_local} gives a similar statement for local solutions.

\begin{remark} \label{rem:bd_no_soln}
Problem \eqref{prob:bd_constr} is not guaranteed to have a solution.
For example, let $f(X) = ||X - A||_F^2$ for some matrix $A \in \mathbb{R}^{B \times B}$.
If $A=0$ then \eqref{prob:bd_constr} has no solution (e.g. consider $X_n = \frac{1}{n} I_B$ then $f(X_n) \to 0$ and each $X_n$ satisfies the constraints, but $X_n \to 0$ which does not satisfy the constraints). 
On the other hand, if $NB(A) \ge B$ then $X=A$ is the global solution.
\end{remark}

Informally, if the objective function $f$ does not encourage too many rows/columns to be 0, Problem \eqref{prob:bd_constr} will have a solution.
When $f$ does not have a solution, it may indicate that block diagonal constraints are not a good modeling choice.
For example, it does not make sense to ask for the nearest block diagonal matrix to the zero matrix.

Problem \eqref{prob:bd_extremal_rep} is amenable to an alternating minimization algorithm that alternates between updating $U$ and updating $X$.
When $X$ is fixed, a global solution for $U$ is given by an eigen-decomposition.
When $U$ is fixed, the second term in the objective and the second term in the constraints of Problem \eqref{prob:bd_extremal_rep} are linear in $X$.
This alternating algorithm is detailed in Section \ref{ss:sym_lap_pen_alt_algo} and includes the case where $f$ replaced by a surrogate function at each step.
While this algorithm is similar to the BSUM algorithm \citep{razaviyayn2013unified, kumar2019unified}, its convergence properties are more challenging to study due the non-convexity of $f$ and the non-linearly coupled constraints.
Section \ref{ss:algo_convergence} studies the convergence behavior of this alternating algorithm using Zangwill's convergence theory.

Section \ref{s:un_lap_bad} contrasts our approach based on the symmetric Laplacian with similar approached based on the unnormalized Laplacian \citep{nie2016constrained, nie2017learning, lu2018subspace, kumar2019unified}.
Section \ref{s:multiarray} shows the approach discussed in this section for matrices naturally extends to enforcing block diagonal constraints on multi-arrays.

\subsection{MVMM with block diagonal constraints} \label{ss:mvmm_block_diag}

This section presents a constrained maximum likelihood problem that imposes a block diagonal structure on $\pi$ for the MVMM for $V=2$ views.
We decompose $\pi = \epsilon \mathbf{1}\mathbf{1}^T + \blckdiag$ where $\epsilon > 0$ is a small constant and $\blckdiag$ is a block diagonal matrix.
The $\epsilon$ term lets the model have ``outliers" e.g. observations that do not fall cleanly in the block diagonal structure.
It is also useful for computational reasons to avoid issues with exact zeros similar to those discussed in Section \ref{s:sparsity_log_pen}.
In particular, we consider
\begin{equation} \label{prob:mvmm_bd_constr}
\begin{aligned}
&\underset{\Theta, \blckdiag}{\text{minimize}}   & &  -  \ell( \{x_i\}_{i=1}^n| \Theta, \epsilon \mathbf{1}_{K\vs{1}}\mathbf{1}_{K\vs{2}}^T + \blckdiag)  \\
& \text{subject to} & & \blckdiag \ge 0,\langle \blckdiag, \mathbf{1}_{K\vs{1}} \mathbf{1}_{K\vs{2}}^T \rangle = 1 - K\vs{1} K\vs{2} \epsilon \\ 
&&& \blckdiag \text{ has at least } B \text{ blocks up to permutations},
\end{aligned}
\end{equation}
where $\ell$ is the observed data  log-likelihood \eqref{eq:data_log_lik} and $0 < \epsilon < \frac{1}{K\vs{1} K\vs{2}}$.
Following Section \ref{ss:block_diag_opt}, we replace the block diagonal constraint with a penalty on the smallest generalized eigenvalues of $L_{\text{sym}}(A_{\text{bp}}(\blckdiag))$.
An alternating EM algorithm for the resulting problem is presented in Section \ref{ss:mvmm_bd_constr_algo}.
Each step of this algorithm requires an eigen-decomposition and solving a convex problem.
Based on the discussion in Section \ref{s:multiarray}, it is straightforward to extend block diagonal constraints to the case of $V \ge 2$ multi-view mixture models.


\section{Simulations}\label{s:simulations}

We examine the clustering performance of the log penalized MVMM (log-MVMM) and the block diagonally constrained MVMM (bd-MVMM) on a synthetic data example.
The data in this section are sampled from a $V=2$ view Gaussian mixture model where $\pi \in \mathbb{R}^{10 \times 10}$ has five $2 \times 2$ blocks (Figure \ref{fig:beads_2_5__1__.5__Pi_true} below) with $d\vs{1} = d\vs{2}=10$ features.
Each view cluster has an identity covariance matrix.
The cluster means are sampled from isotropic Gaussians with standard deviations $\sigma_{\text{mean}}\vs{1}$ and $\sigma_{\text{mean}}\vs{2}$ for the first and second views respectively.
These parameters control the difficulty of the clustering problem e.g. if are both large then the cluster means tend to be far apart.
In this section we set $\sigma_{\text{mean}}\vs{1} = 1$  and $\sigma_{\text{mean}}\vs{2} = .5$ meaning the clusters in the first view are better separated than those in the second view.
The simulations below are repeated 20 times with different seeds and the cluster means are sampled once for each Monte-Carlo repetition.

The log-MVMM model is fit for a range of $\lambda$ values and we assume the number of view clusters $K\vs{1} = K\vs{2} = 10$ are known.
The bd-MVMM is also fit for a range of number of blocks and we set $\epsilon=  0.01 \cdot \frac{1}{K\vs{1} \cdot K\vs{2}}$.
Both the log-MVMM and bd-MVMM are initialized by fitting the basic MVMM discussed in Section \ref{s:mvmm} for 10 EM iterations. 
All models fit in this section assume diagonal covariance matrices for each cluster and use a small amount of covariance regularization to prevent clusters from collapsing on a single observation. 
As baselines for comparison we also fit the basic MVMM (MVMM) as well as a mixture model on the concatenated data (cat-MM).

\begin{figure}[H]
 \centering
\begin{subfigure}[t]{0.3\textwidth}
\includegraphics[width=\linewidth]{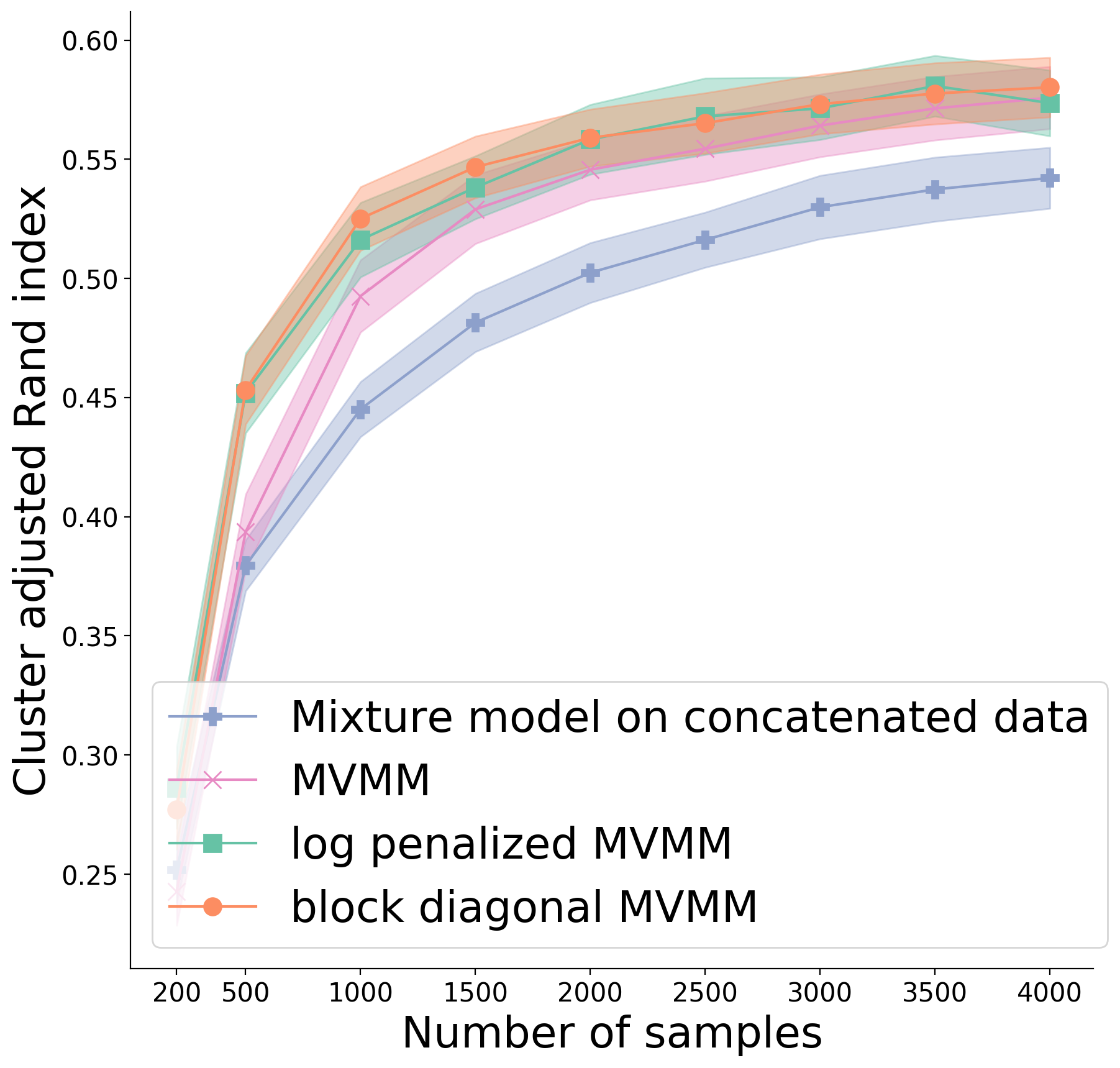}
\caption{
ARI comparing the predicted vs. true cluster labels.
}
\label{fig:beads_2_5__1__.5__n_samples_vs_test_overall_ars_at_truth}
\end{subfigure}
\hfill
\begin{subfigure}[t]{0.3\textwidth}
\includegraphics[width=\linewidth]{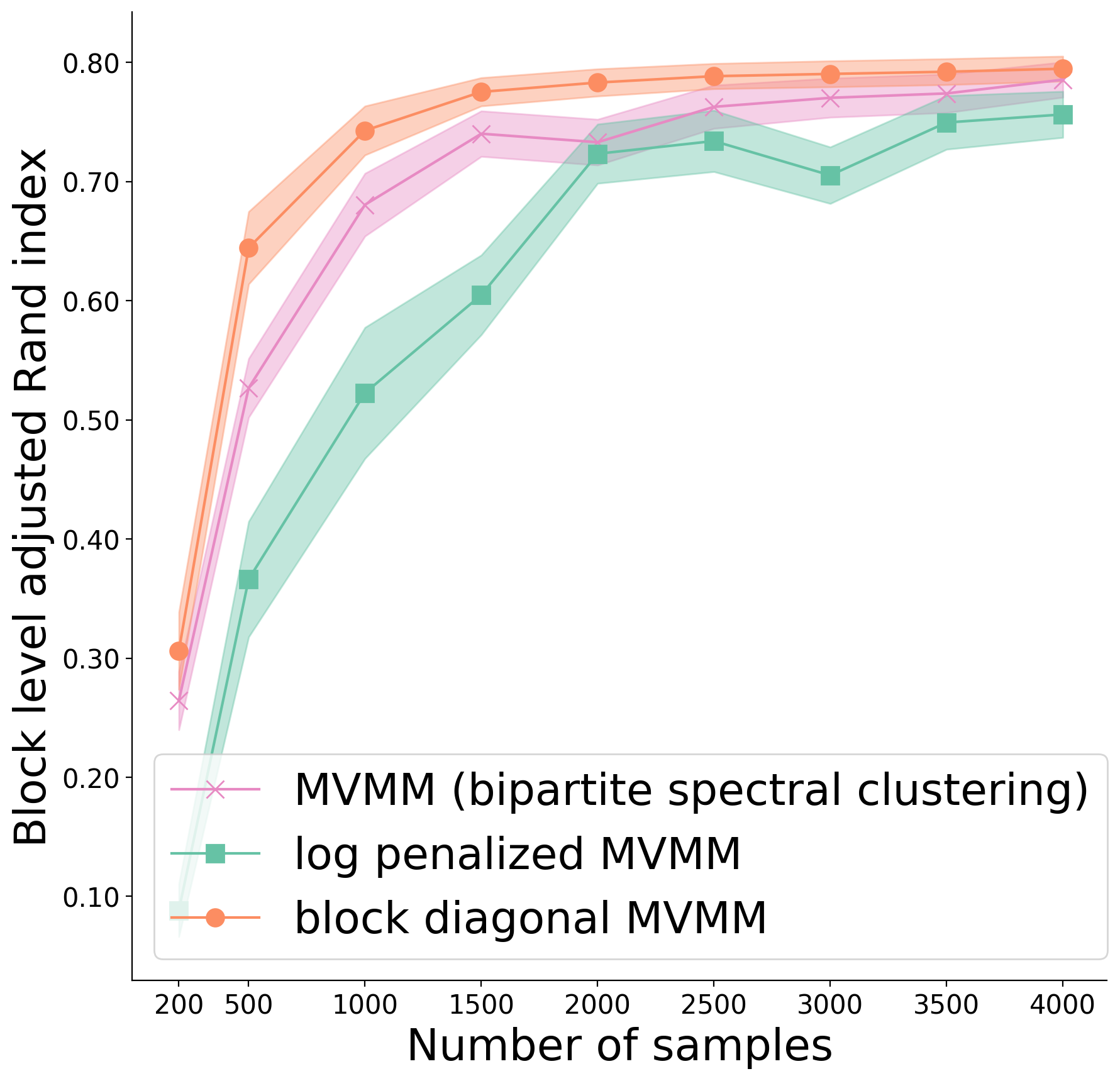}
\caption{
ARI comparing the predicted vs. true block level labels.
}
\label{fig:beads_2_5__1__.5__n_samples_vs_test_community_restr_ars_at_truth}
\end{subfigure}
\caption{
Clustering performance at the true hyper-parameter values; 20 components for log-MVMM, 20 components for cat-MM and 5 blocks for bd-MVMM.
The lines show the Monte-Carlo means; the shaded areas show $\pm \frac{1}{\sqrt{20}}$ times the Monte-Carlo standard deviation.
}
\label{fig:beads_2_5__1__.5_eval_at_truth}
\end{figure}

We first compare each model when the true parameter values are known e.g. total number of components for log-MVMM\footnote{If the true value does not show up in the tuning sequence we pick the model with the closest value.} and the true number of blocks for the bd-MVMM.
Figure \ref{fig:beads_2_5__1__.5_eval_at_truth} shows the results for a range of training sample sizes  ($n=$ 200, 500, 1000, 1500, 2000, 2500, 3000, 3500, 4000).
Recall the \textit{adjusted Rand index} (ARI) measures how well a vector of predicted cluster labels corresponds to a vector of true cluster labels where large values mean better correspondence \citep{rand1971objective}.

Figure \ref{fig:beads_2_5__1__.5__n_samples_vs_test_overall_ars_at_truth} shows the ARI of each model's predicted cluster labels compared to the true cluster labels for an independent test set (note there are $|\text{supp}(\pi)| = 20$ true clusters). 
Here the bd-MVMM and log-MVMM perform better than the two baselines (MVMM and cat-MVMM) for a range of sample sizes.
The prior information about the sparsity structure of $\pi$ helps these two models estimate the cluster parameters.
The performance gap is larger for smaller sample sizes and narrows with enough data.
Note the cat-MM catches up slowly because it does not take the view structure into account.

Figure \ref{fig:beads_2_5__1__.5__n_samples_vs_test_community_restr_ars_at_truth} evaluates the models' ability to find the block structure of the $\pi$ matrix.
Here we group clusters together that are in the same block i.e. there are 5 true block clusters.
For the log-MVMM and bd-MVMM we predict block cluster labels based on the block structure of the estimated $\widehat{\pi}$ and $\widehat{\blckdiag}$ matrices respectively.
As a baseline for comparison we apply bipartite spectral clustering \citep{dhillon2001co} to the estimated $\widehat{\pi}$ matrix from the MVMM. 
Here the bd-MVMM performs the best, which is not surprising because it was designed to target this kind of structure.
The log-MVMM struggles because small mistakes on the $\widehat{\pi}$ matrix can cause two blocks to be linked.
Once the sample size grows large enough the MVMM eventually comes close to the bd-MVMM.

Figures \ref{fig:beads_2_5__1__.5__Pi_est_n2000} and \ref{fig:beads_2_5__1__.5__Pi_est_n2500}  show the estimated $\widehat{\blckdiag}$ matrix from one Monte-Carlo repetition. 
For smaller sample sizes the block diagonal structure is almost correct (Figure  \ref{fig:beads_2_5__1__.5__Pi_est_n2000}).
With more samples the bd-MVMM finds the correct block diagonal structure (Figure  \ref{fig:beads_2_5__1__.5__Pi_est_n2500}).

\begin{figure}[H]
 \centering
\begin{subfigure}[t]{0.3\textwidth}
\includegraphics[width=.5\linewidth, height=.5\linewidth]{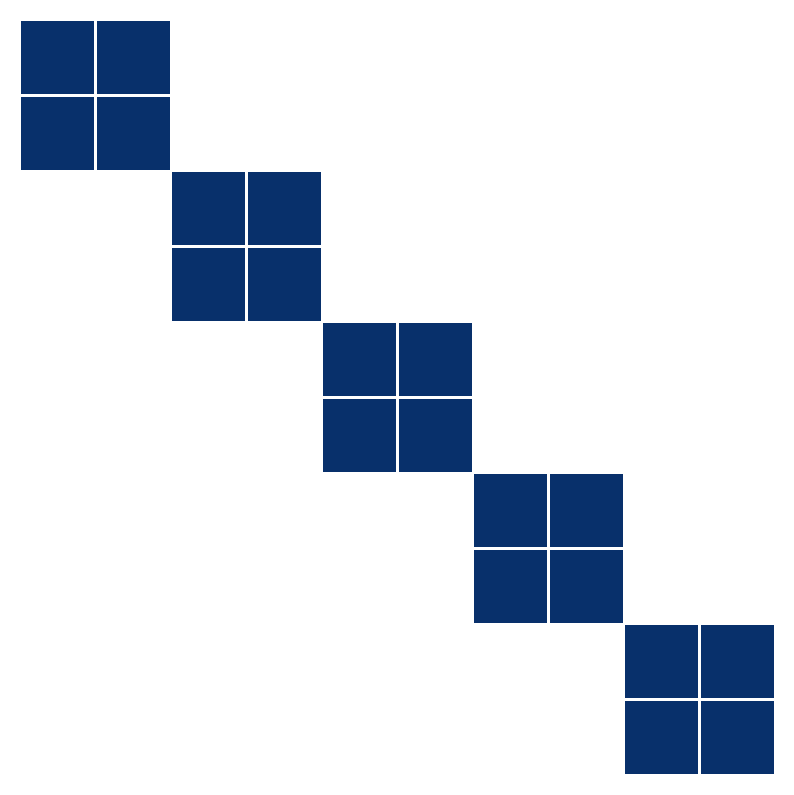}
\caption{
True $\pi \in \mathbb{R}^{10 \times 10}$.
}
\label{fig:beads_2_5__1__.5__Pi_true}
\end{subfigure}
\begin{subfigure}[t]{0.3\textwidth}
\includegraphics[width=.5\linewidth,  height=.5\linewidth]{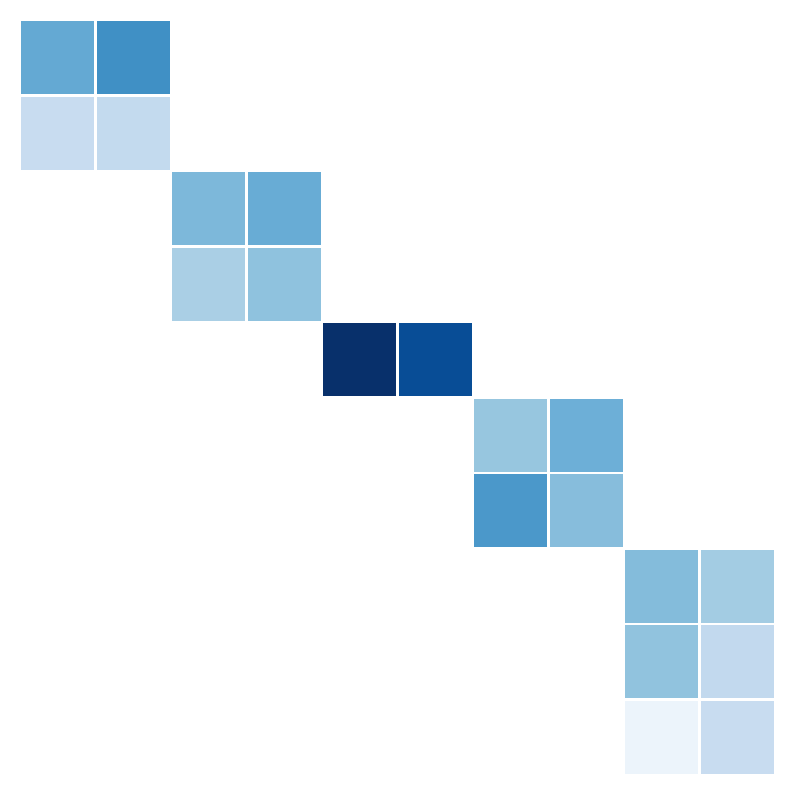}
\caption{
Estimated $\widehat{\blckdiag}$ with $n=2,000$ samples.
}
\label{fig:beads_2_5__1__.5__Pi_est_n2000}
\end{subfigure}
\begin{subfigure}[t]{0.3\textwidth}
\includegraphics[width=.5\linewidth,  height=.5\linewidth]{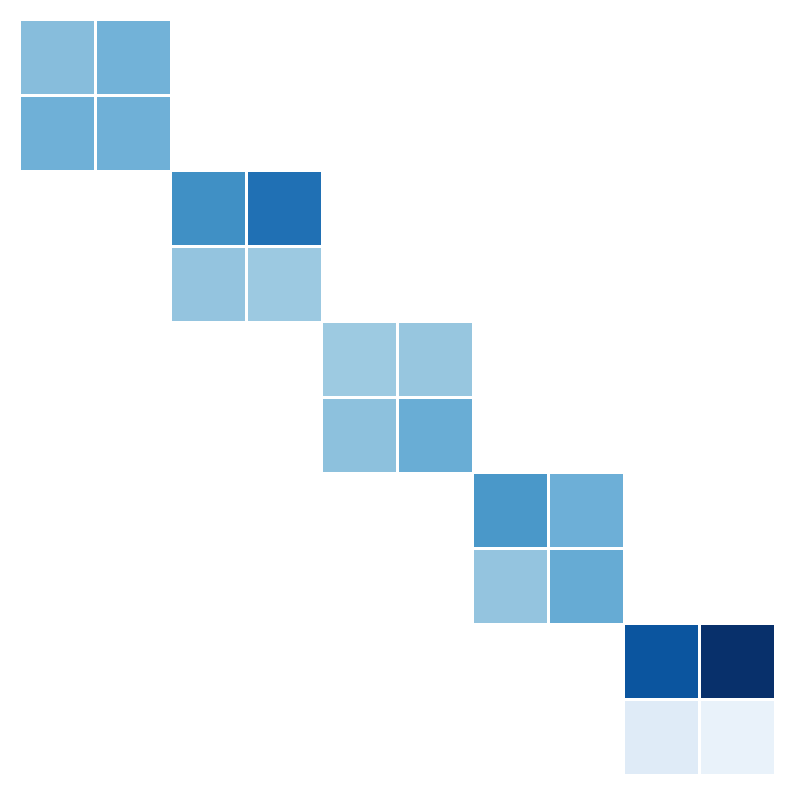}
\caption{
Estimated $\widehat{\blckdiag}$ with $n=2,500$ samples.
}
\label{fig:beads_2_5__1__.5__Pi_est_n2500}
\end{subfigure}
\caption{
True $\pi$ and estimated $\widehat{\blckdiag}$ matrices.
The cluster labels have been permuted to reveal the block diagonal structure.
}
\label{fig:beads_2_5__1__.5}
\end{figure}

Next we evaluate the models after performing model selection using a modified BIC criteria \citep{schwarz1978estimating}.
After fitting log-MVMM for a range of $\lambda$ values we select the best model using the following BIC criterion suggested by \citep{huang2017model}
\begin{equation} \label{eq:log_pen_bic}
\text{BIC} = 2 \sum_{i=1}^n \ell( x_i | \widehat{\Theta}, \widehat{\pi}) - \left(\text{dof}(\widehat{\Theta}) + |\text{supp}(\widehat{\pi})| - 1\right) \log(n),
\end{equation}
where $\widehat{\Theta}$ and $\widehat{\pi}$ are the estimated cluster parameters and $\pi$ matrices respectively and $\text{dof}(\cdot)$ is the number of degrees of freedom of the cluster parameters.
\citet{huang2017model} provides results about the consistency of this model selection procedure for single view Gaussian mixture models.
For the bd-MVMM we use a similar formula except the support of $\widehat{\pi}$ is replaced with $|\text{supp}(\widehat{\blckdiag})|$.

Figure \ref{fig:beads_2_5__1__.5__n_samples_vs_n_comp_est_at_bic_selected} shows the BIC estimated number of components for log-MVMM and bd-MVMM.
The bd-MVMM does a good job with model selection (e.g. it usually picks 5 blocks), but log-MVMM tends to select too few clusters.
Figure \ref{fig:beads_2_5__1__.5__n_samples_vs_test_overall_ars_at_bic_selected} and \ref{fig:beads_2_5__1__.5__n_samples_vs_test_community_restr_ars_at_bic_selected} are similar to Figures  \ref{fig:beads_2_5__1__.5__n_samples_vs_test_overall_ars_at_truth} and \ref{fig:beads_2_5__1__.5__n_samples_vs_test_community_restr_ars_at_truth}, but the BIC selected parameter values are used instead of the true values.
Here bd-MVMM still outperforms the MVMM, but by a smaller margin.

\begin{figure}[H]
 \centering
 \begin{subfigure}[t]{0.3\textwidth}
\includegraphics[width=\linewidth,  height=\linewidth]{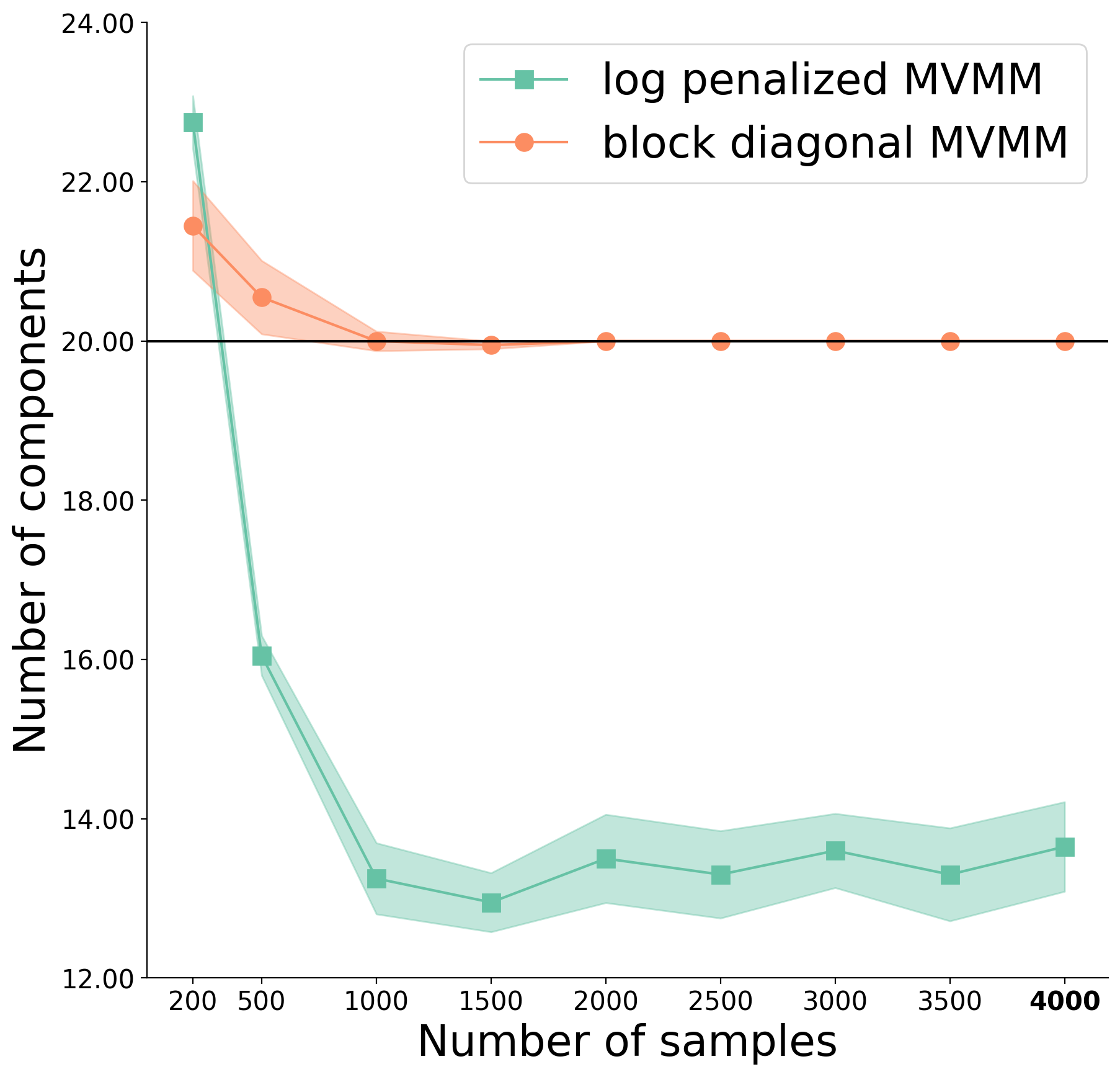}
\caption{
Estimated total number of clusters.
}
\label{fig:beads_2_5__1__.5__n_samples_vs_n_comp_est_at_bic_selected}
\end{subfigure}
\hfill
\begin{subfigure}[t]{0.3\textwidth}
\includegraphics[width=\linewidth,  height=\linewidth]{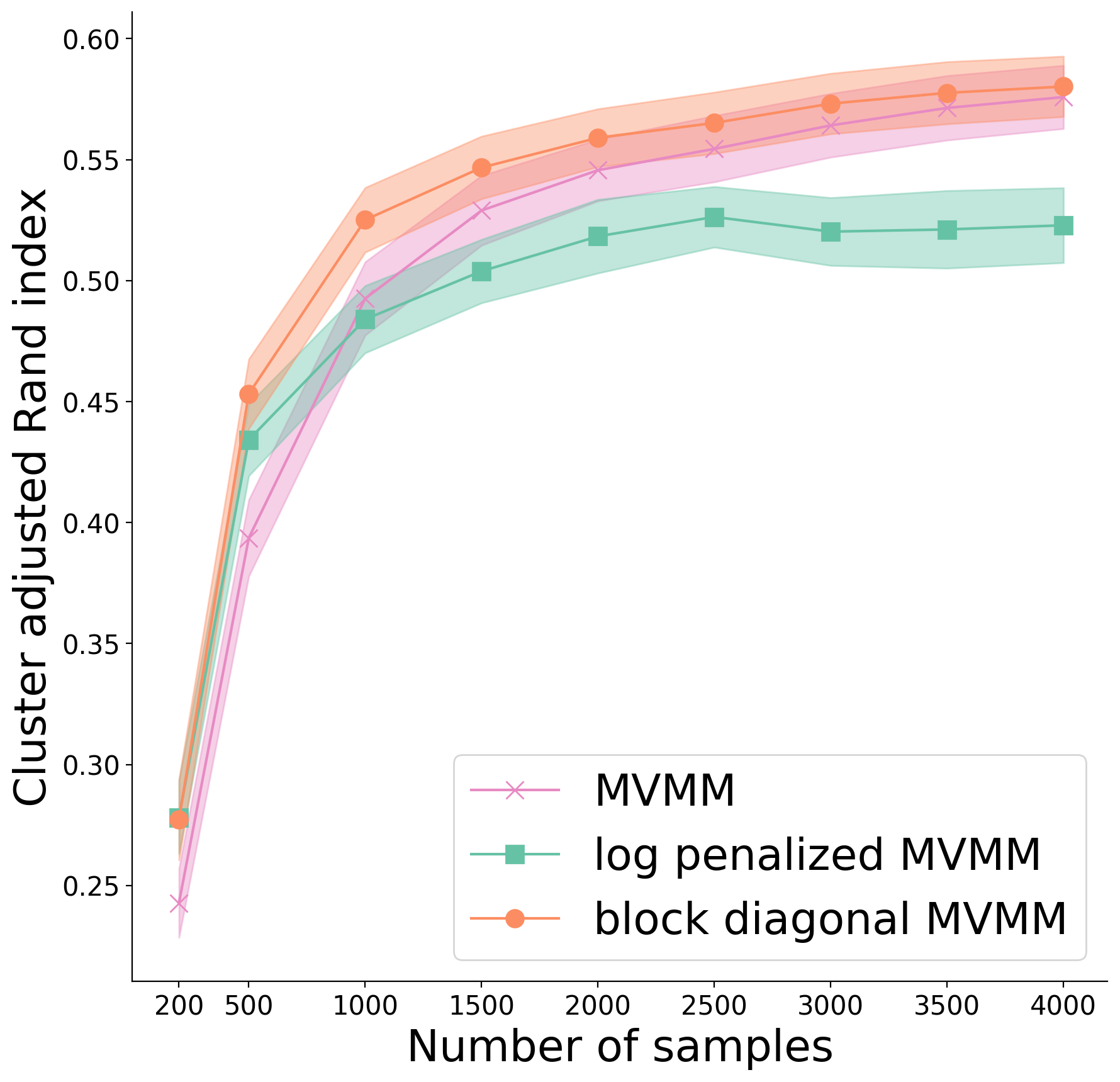}
\caption{
ARI of predicted vs. true cluster labels.
}
\label{fig:beads_2_5__1__.5__n_samples_vs_test_overall_ars_at_bic_selected}
\end{subfigure}
\hfill
\begin{subfigure}[t]{0.3\textwidth}
\includegraphics[width=\linewidth,  height=\linewidth]{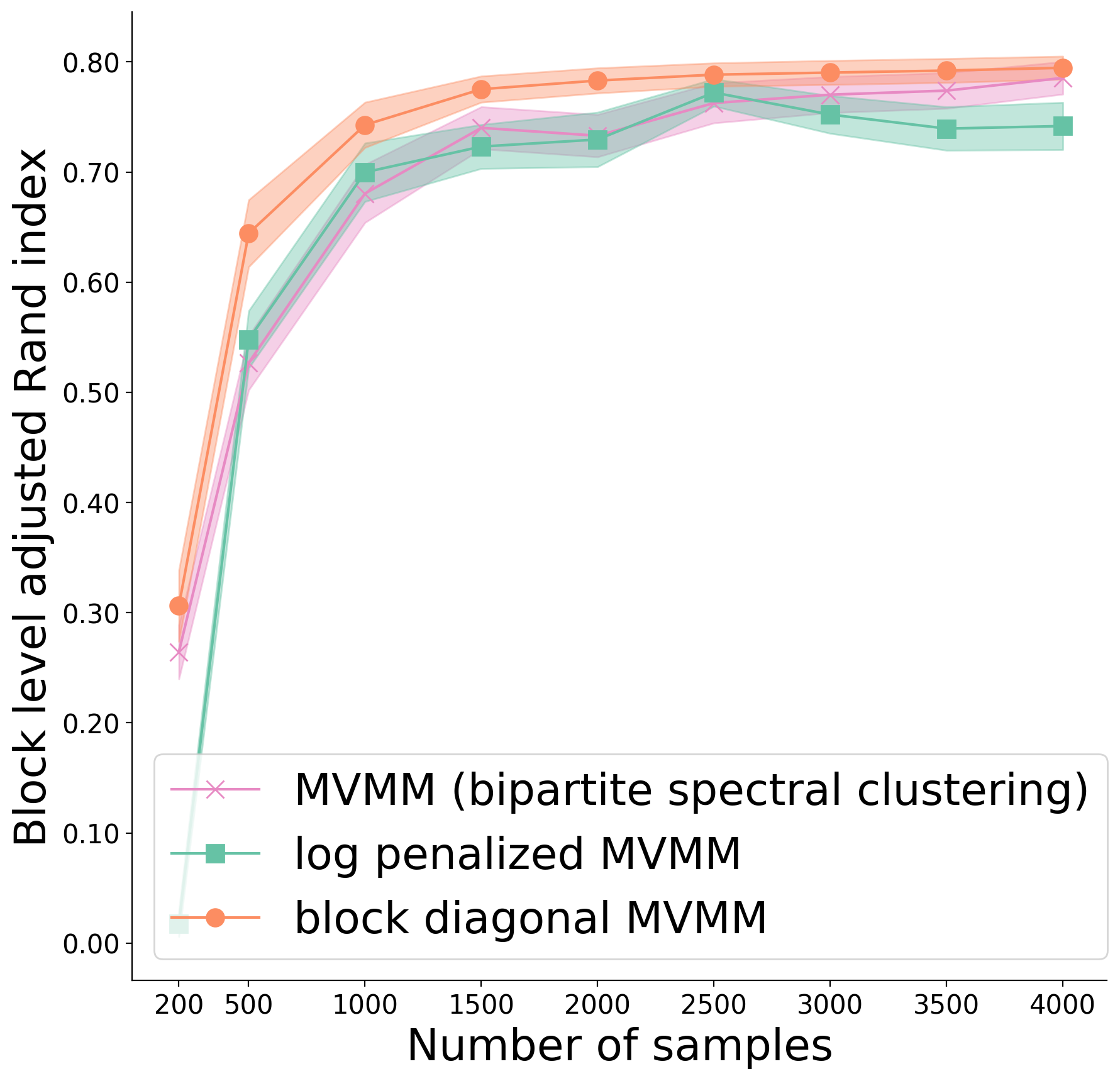}
\caption{
ARI of predicted vs. true block level labels.
}
\label{fig:beads_2_5__1__.5__n_samples_vs_test_community_restr_ars_at_bic_selected}
\end{subfigure}
\caption{
Clustering performance of the BIC selected models for log-MVMM and bd-MVMM. 
}
\label{fig:beads_2_5__1__.5_eval_at_bic}
\end{figure}

This section focuses on the case when the signal to noise level is different in each view.
Additional simulations examining different $\pi$ matrices and different noise levels are shown in Section \ref{s:add_sims}.
These additional simulations show that when the noise level is the same in each view ($\sigma_{\text{mean}}\vs{1} =\sigma_{\text{mean}}\vs{2}$) then log-MVMM and bd-MVMM perform much closer to the MVMM.


\section{Real data examples} \label{s:real_data}

This section applies the block diagonal MVMM to two different data sets.
While more detailed analysis is beyond the scope of this paper we provide additional results and figures in the online supplementary material.

\subsection{TCGA breast cancer}\label{ss:tcga}

The TCGA breast cancer study \citep{cancer2012comprehensive} collects data from 1,027 breast cancer patients on multiple genomic platforms including: RNA expression (RNA), microRNA (miRNA), DNA methylation (DNA) and copy number (CP).
We closely follow the data processing guidelines from \citet{hoadley2018cell}, leaving us with 3,217 RNA features, 383 miRNA features, 3,139 DNA features and 3,000 CP features.\footnote{We were unable to obtain the feature list for copy number so we selected the top 3,000 features with the largest variance.}
Missing values are filled in using 5 nearest neighbors imputation \citep{troyanskaya2001missing}.
We first determine the number of clusters in each view by fitting a Gaussian mixture model (with diagonal covariances) to each view marginally; BIC selects 10 RNA clusters, 11 miRNA clusters, 25 DNA clusters and 32 CP clusters.

Next we fit a $V=2$ view block diagonal MVMM to the following pairings: RNA vs. miRNA, RNA vs. DNA and RNA vs. CP.
BIC selects 1 block for RNA vs. miRNA, 1 block for RNA vs. DNA and 3 blocks for RNA vs. CP.
Figure \ref{fig:cp_bd_pi} shows the estimated $\widehat{\blckdiag}$ matrices for RNA vs. CP.
The block diagonal structure of these matrices suggests there is strong jointly defined subtypes in the RNA and CP views.
Note there is still joint information in RNA/miRNA and RNA/DNA since the estimated $\widehat{D}$ matrices are not rank one.

\begin{figure}[H]
 \centering
 \begin{subfigure}[t]{0.3\textwidth}
\includegraphics[width=\linewidth, height=1.5\linewidth]{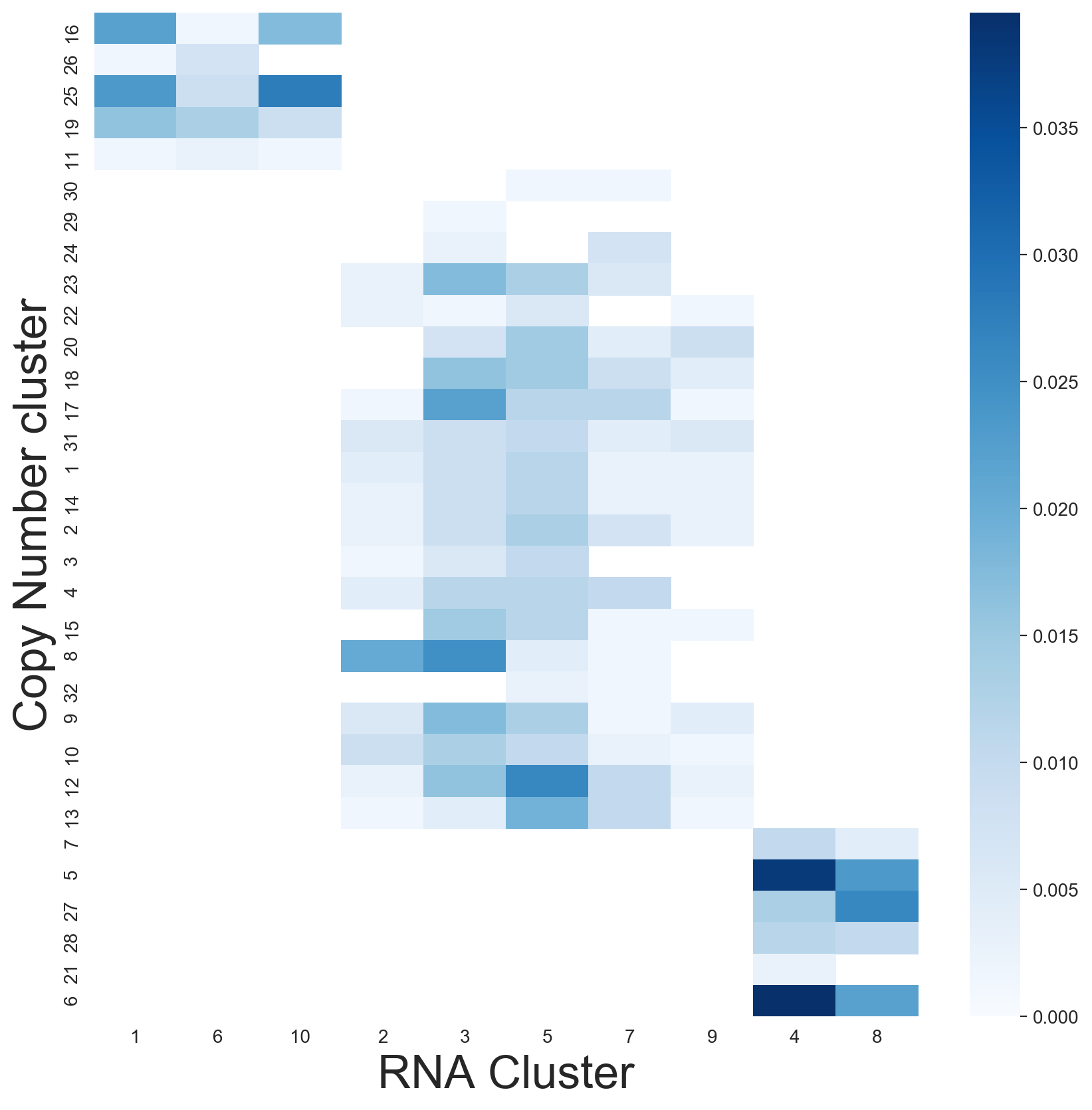}
\caption{
Estimate $\widehat{D}$ matrix for RNA vs. CP.
The rows and columns are permuted to reveal the 3 blocks.
}
\label{fig:cp_bd_pi}
\end{subfigure}
\hfill
\begin{subfigure}[t]{0.3\textwidth}
\includegraphics[width=\linewidth, height=1.5\linewidth]{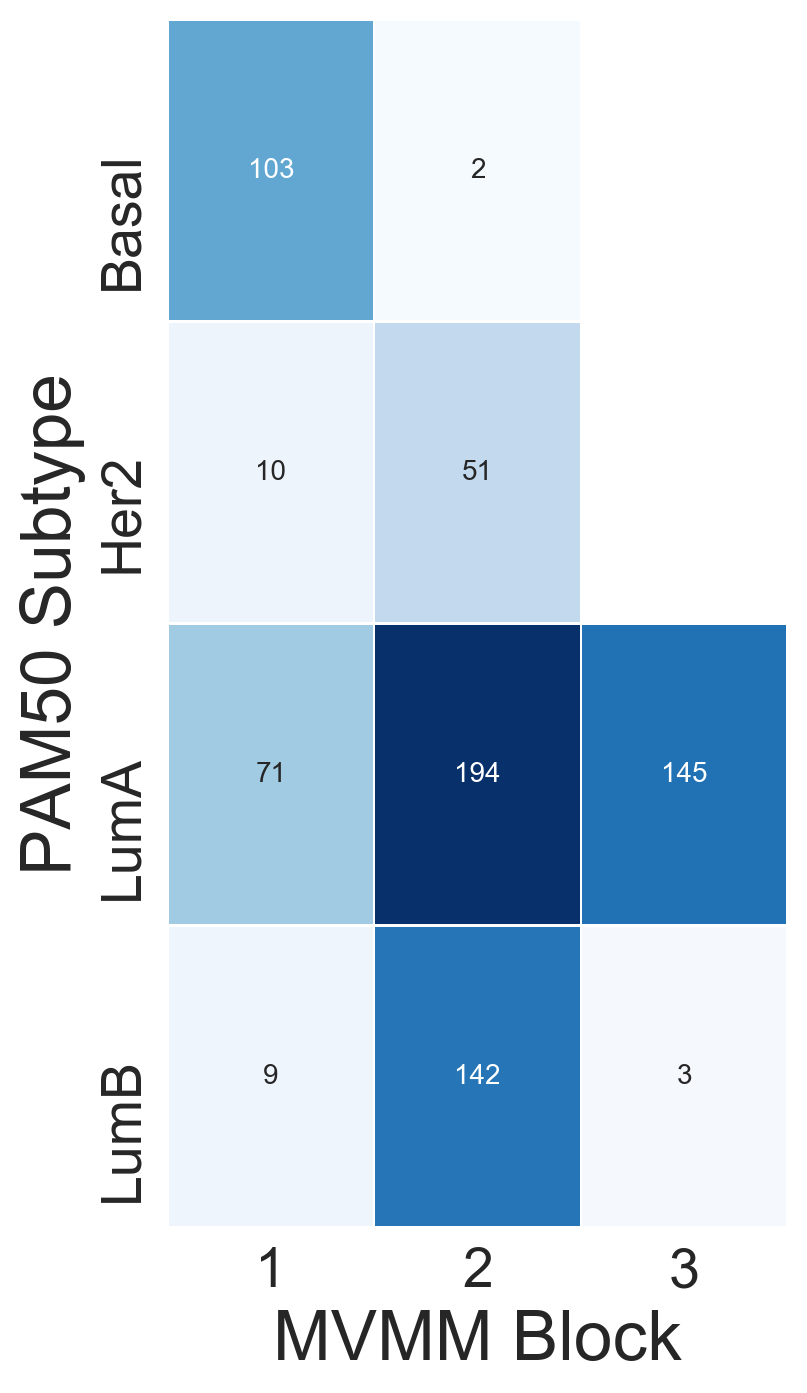}
\caption{
Contingency table for predicted block labels vs. known PAM50 subtype labels.
}
\label{fig:cp_block_vs_pam50_subtype}
\end{subfigure}
\hfill
\begin{subfigure}[t]{0.3\textwidth}
\includegraphics[width=\linewidth, height=1.5\linewidth]{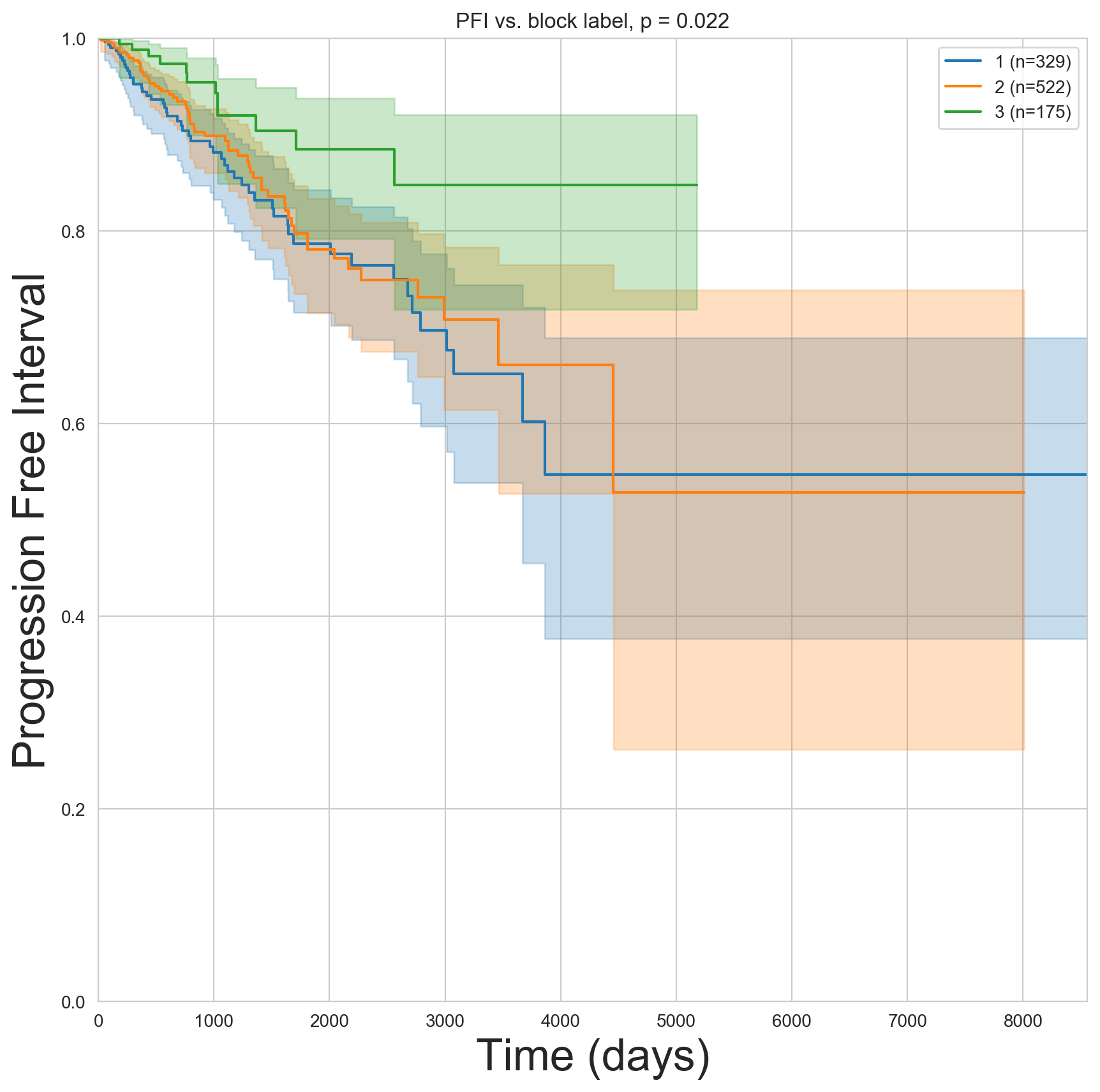} 
\caption{
Kaplan-Meier curves comparing the RNA-CP blocks against PFI.
A LogRank test finds the block label are statistically significantly related to PFI.
}
\label{fig:cp_block_vs_survial} 
\end{subfigure}
\caption{
The first block picks out Basal like tumors as well as a few Luminal A tumors.
The third block picks out Luminal A tumors that tend to have better survival based on PFI.
}
\label{fig:tcga_figs}
\end{figure}

We next investigate the RNA/CP blocks using two additional clinical variables: PAM50 subtype (Basal-like, Luminal A, Luminal B, Her2-enriched) and survival measured by \textit{progression free interval} (PFI) as recommended by \citep{liu2018integrated}.
Figure \ref{fig:cp_block_vs_pam50_subtype} shows block 1 picks out the Basal-like subtype, which is known to have a strong genomic signal in each platform \citep{cancer2012comprehensive, hoadley2018cell}.
Figures \ref{fig:cp_block_vs_pam50_subtype} and \ref{fig:cp_block_vs_pam50_subtype} show block 3 picks out Luminal A tumors that have better survival.

\subsection{Neuron cell types} \label{ss:cell_types}

Integrative clustering has become increasingly important for neuron subtype discovery \citep{gouwens2019classification, gouwens2020toward}.
Neuroscientists are now able to collect a variety of data modalities from individual mouse neurons including transcriptomic, morphological and electrophysiological features.
We apply the bd-MVMM to an excitatory mouse neuron data set obtained from the Allen Institute \citep{gouwens2020toward}.
This two-view data set has 44 electrophysiological (EPHYS) features and 69 transcriptomic features (RNA) available for $n= 4,269$ excitatory neurons.
The EPHYS features were obtained using sparse PCA on 12 raw electrophysiological time series recordings in an awake mouse as discussed in \citep{gouwens2019classification}.
The RNA features are obtained by first selecting the $4,019$ most differentially expressed genes as in \citep{tasic2018shared}, applying a log transform then extracting the top $69$ PCA features.
This PCA  rank was selected using the singular value thresholding method discussed in \citep{gavish2014optimal}.
We first determine the number of clusters in each view by fitting a Gaussian mixture model (with diagonal covariances) to each view marginally; BIC selects 47 EPHYS clusters and 41 RNA clusters.

\begin{figure}[H]
 \centering
 \begin{subfigure}[t]{0.5\textwidth}
\includegraphics[width=1\linewidth, height=1\linewidth]{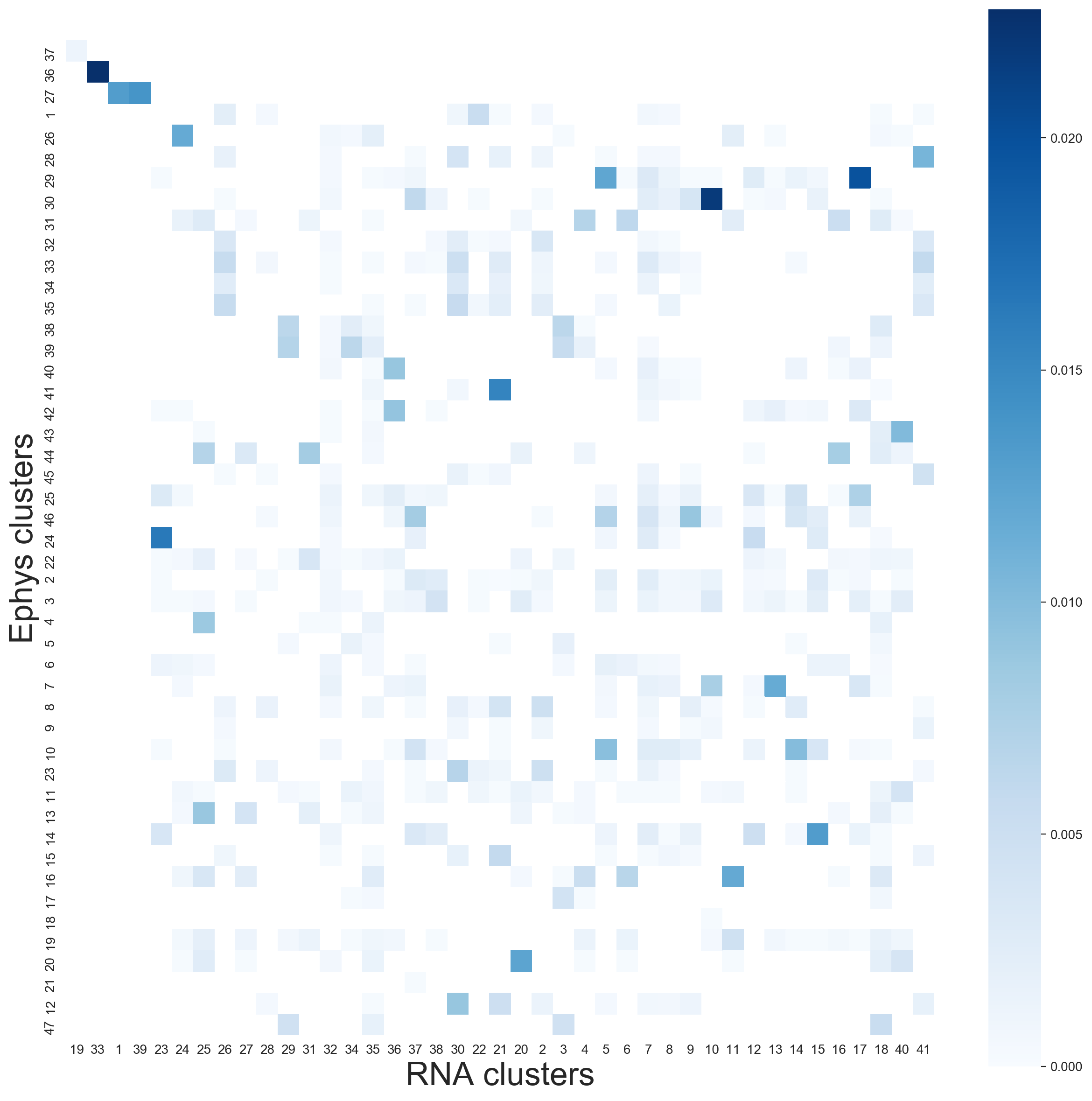}
\caption{
Estimated $\widehat{\blckdiag}$ matrix for EPHYS vs. RNA.
The rows/columns are permuted to reveal the block diagonal structure.
}
\label{fig:mouse_et_pi}
\end{subfigure} %
\hfill
\begin{subfigure}[t]{0.48\textwidth}
\centering
\includegraphics[width=1\linewidth, height=1\linewidth]{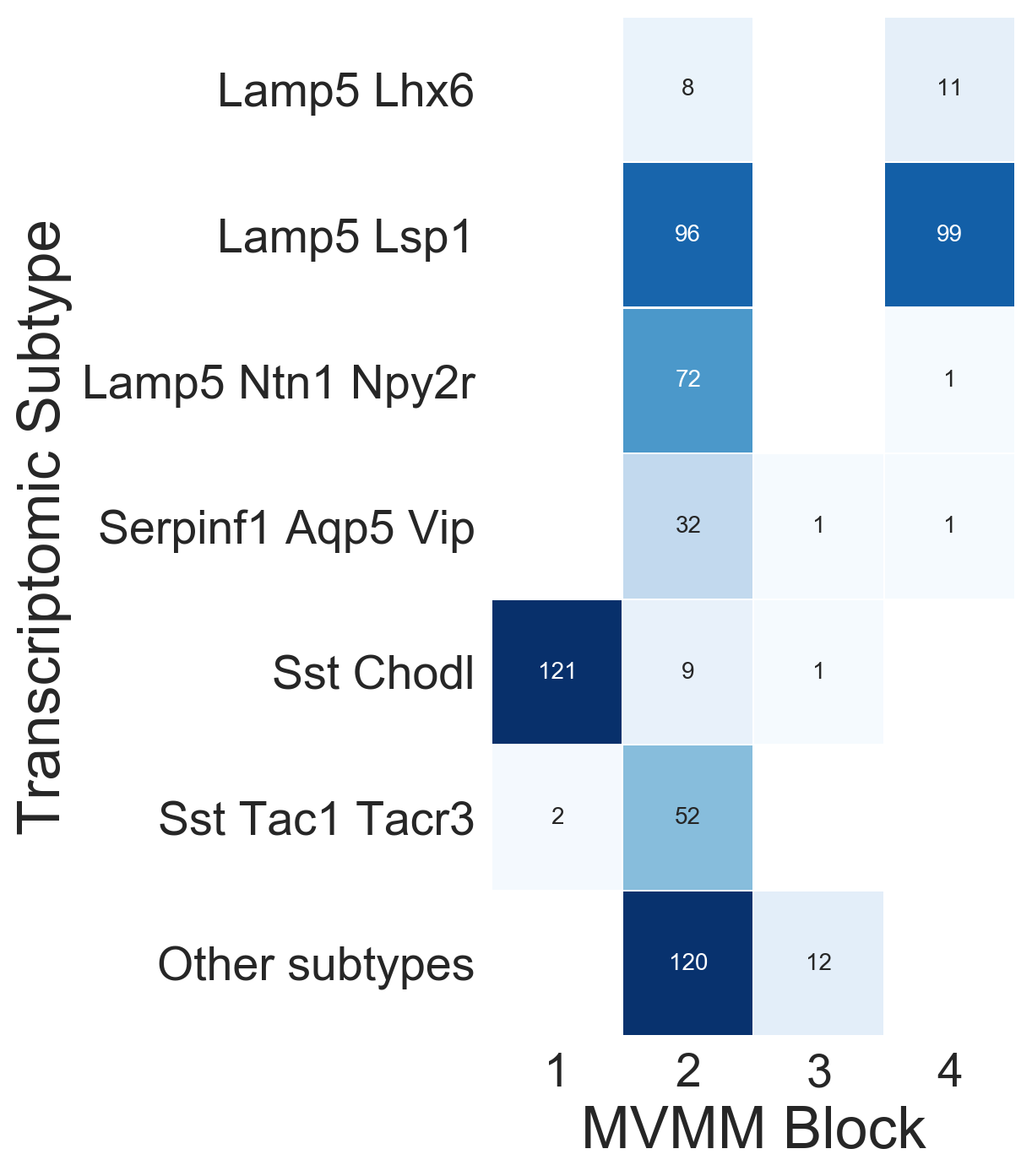}
\caption{
Predicted block labels vs. known transcriptomic subtypes. 
Block 1 picks out the ``Sst Chodl" subtype; block 4 picks out ``Lampp5 Lhx6" and ``Lamp5 Lsp1" subtypes.
}
\label{fig:block_vs_transcr_subtype}
\end{subfigure}
\caption{
MVMM blocks for EPHYS vs. RNA.
}
\label{fig:mouse_et_blocks}
\end{figure}

\begin{figure}[H]
 \centering
 \begin{subfigure}[t]{0.66\textwidth}
\includegraphics[width=1\linewidth, height=1\linewidth]{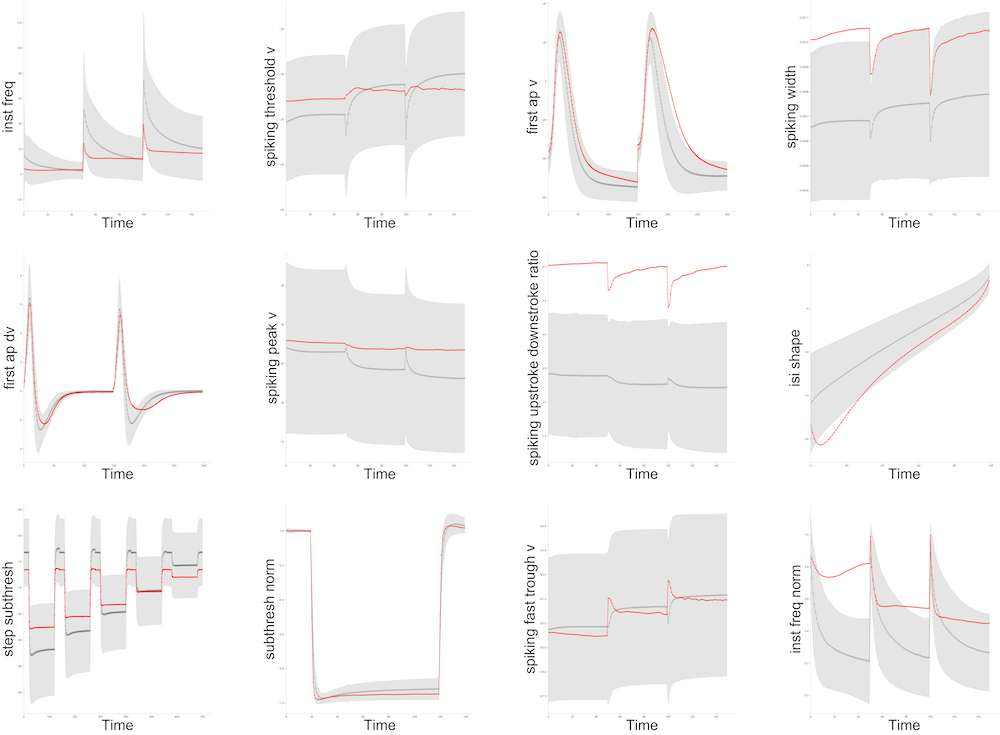}
\caption{
Visualization of EPHYS cluster 36.
The grey line shows the overall mean for each time series and the grey shaded area shows $\pm$ 1 standard deviation.
The red lines show the cluster mean of the raw EPHYS recordings as described in the below caption.
}
\label{fig:ephys_curve_cl_36}
\end{subfigure} %
\hfill
\begin{subfigure}[t]{0.32\textwidth}
\centering
\includegraphics[width=1\linewidth, height=2\linewidth]{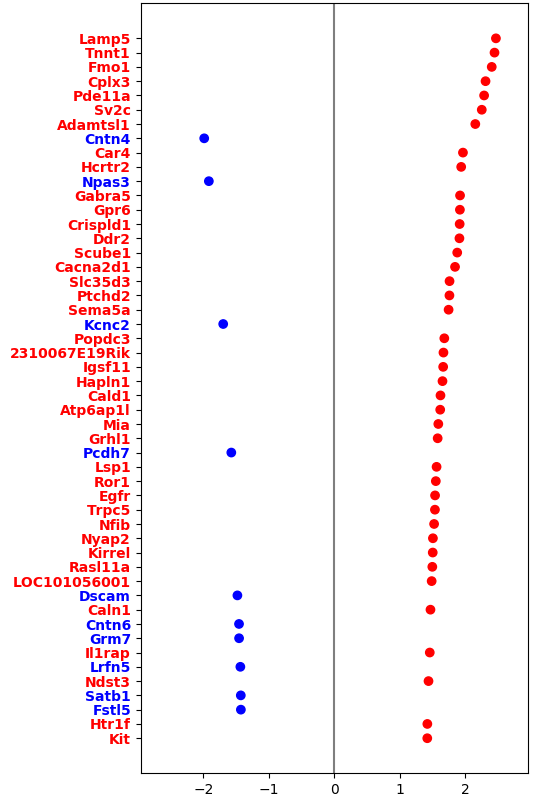}
\caption{
Visualization of the top 50 genes for RNA cluster 33.
The values shown are the standardized difference of the cluster mean minus the overall mean scaled by the overall sample standard deviation.
}
\label{fig:transcr_cl_33}
\end{subfigure}
\caption{
Block 4 identifies EPHYS cluster 36 (Figure \ref{fig:ephys_curve_cl_36}) with RNA cluster 33 (Figure \ref{fig:transcr_cl_33}).
Both figures show visualizations of the cluster means of the raw variables.
While the clustering algorithm was run on PCA features, we show the means on the scale of the original features.
To represent the mean on the raw data scale we compute a weighted average of all observations, where the weights are given by the cluster prediction probabilities (i.e. this is essentially the M-step for the Gaussian mean parameter).
}
\label{fig:mouse_et_block_4}
\end{figure}

We next fit a block diagonal MVMM to this two-view data set and select 4 blocks using BIC.
There is one large $38 \times 43$  block (i.e. 38 RNA clusters and 43 RNA clusters) and the other blocks are $1\times 1$, $1\times 1$ and $2 \times 1$.
This block diagonal structure suggest there are a handful of jointly well defined clusters while the rest of the information is mixed between the two views.
Previous research has identified 60 RNA clusters \citep{tasic2018shared}, which we compare to the predicted block labels found by the MVMM (Figure \ref{fig:block_vs_transcr_subtype}).
This figure shows, for example, block 1 picks out the ``Sst Chodl" subtype and block 4 picks out ``Lampp5 Lhx6" and ``Lamp5 Lsp1" subtypes.

Figure \ref{fig:mouse_et_block_4} takes a closer look at block 4 that identifies RNA cluster 33 with EPHYS cluster 36.
Figure \ref{fig:ephys_curve_cl_36} shows a visualization of EPHYS cluster 36's mean for each of the raw EPHYS response variables. 
This cluster, for example, has a higher than average ``spiking width" and ``spiking upstroke downstroke ratio" responses.
Figure \ref{fig:transcr_cl_33} shows the RNA cluster 33's mean on the scale of the standardized residual from the overall mean (i.e. the value shown for each variable is $\frac{\text{cluster mean} - \text{overall mean}}{\text{sample standard deviation}}$).



\section{Conclusion}
We presented two methods to estimate the sparsity structure of the $\Pi$ matrix for the multi-view mixture model.
The log-MVMM presented in Section \ref{s:sparsity_log_pen} makes no assumption about the structure of the sparsity while the bd-MVMM presented in Section \ref{s:block_diag} assumes there is a block diagonal structure.
These methods allow scientists to explore how cluster information is spread across multi-view data sets.

The simulations in Sections \ref{s:simulations} and \ref{s:add_sims} show the modified BIC criteria often works well for the block diagonal MVMM, but tends to select too few clusters for the log penalized MVMM.
Future work may establish better model selection methods e.g. based on \citet{chen2008extended} or \citet{fu2020estimating}.

The main computational bottleneck for the block diagonal MVMM is the convex Problem \eqref{prob:mvmm_block_diag_extremal_mstep} in the M-step.
For simplicity we use an off the shelf \textit{second order cone program} solver \citep{domahidi2013ecos, diamond2016cvxpy}.
A better algorithm may significantly speed up this step.


\section*{Acknowledgements}

We thank Nathan Gowens, Katherine Hoadley, Jonathan Williams, and Daniela Witten for their insightful discussion and guidance.
This material is based upon work supported by the National Science Foundation under Award No. 1902440.

\appendix

\section{Block diagonal multi-arrays} \label{s:multiarray}

The approach discussed in Section \ref{s:block_diag} for enforcing block diagonal constraints on matrices extends to multi-arrays $X\in \mathbb{R}^{d_1 \times \dots \times d\vs{V}}$.
Consider the following problem
\begin{equation} \label{eq:multi_array_blck_diag_opt_problem}
\begin{aligned}
&\underset{X\in \mathbb{R}^{d_1 \times \dots \times d\vs{V}}}{\text{minimize}}   & & f(X)\\
& \text{subject to} & & X \ge 0 \text{ and } X \text{ is block diagonal with at least } B \text{ blocks up to permutations}. \\
\end{aligned}
\end{equation}

\begin{figure}[H]
\begin{subfigure}[t]{0.32\textwidth}
\includegraphics[width=\linewidth]{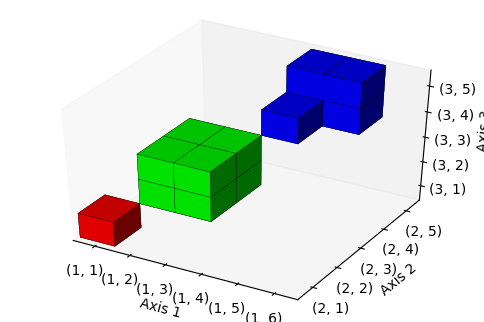} 
\caption{
The support of $X$ where the highlighted entries are non-zero.
Note there is a 2d slice of 0s.
}
\label{fig:block_diag_multiarray}
\end{subfigure}
\hfill
\begin{subfigure}[t]{0.32\textwidth}
\includegraphics[width=\linewidth]{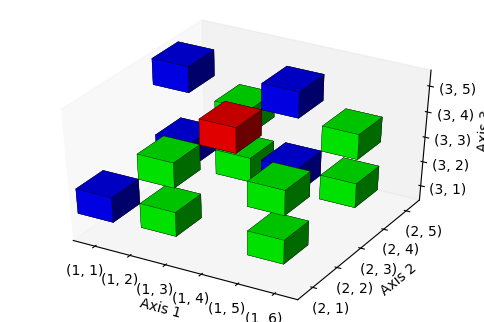} 
\caption{
Same as Figure \ref{fig:block_diag_multiarray}, but the axes have been permuted.
}
\label{fig:block_diag_multiarray_perm}
\end{subfigure}
\hfill
\begin{subfigure}[t]{0.32\textwidth}
\includegraphics[width=\linewidth]{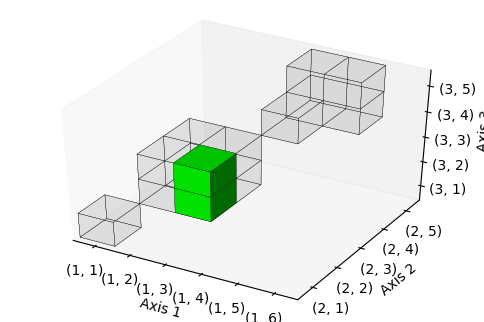}
\caption{
The weight of edge $\{(1, 3), (2, 3)\}$ is $A_{(1, 3), (2, 3)}  = X_{3, 3, 2} + X_{3, 3, 3}$.
}
\label{fig:multiarray_edge_summing}
\end{subfigure}

\caption{
A block diagonal multi-array $X \in \mathbb{R}^{6 \times 5 \times 5}$ with three blocks up to permutations.
}
\label{fig:block_diag_multiarray_viz}
\end{figure}

First we extend definitions for matrices given in Section \ref{s:block_diag} to multi-arrays.
In gory detail, a block of a multi-array is a $V$-hypercube of coordinates, $\mathcal{B} = [L_1, U_1] \times \dots \times [L_V, U_V] \subseteq \mathbb{Z}^V$,
such that $X_{i_1, \dots, i_V} = 0$ if there is a $k \in [V]$ such that $i_k \notin [L_k, U_k]$ but $i_j \in [L_j, U_j]$ for any $j\neq k$.
Figure \ref{fig:block_diag_multiarray} shows a block diagonal multi-array $X \in \mathbb{R}^{6 \times 5 \times 5}$ with three blocks e.g. the upper right block is $[4, 5] \times  [4, 5]  \times  [4, 5]$.
For a fixed block diagonal multi-array, we take the convention that blocks must have at least one non-zero entry and anything that can be a block is a block.
For a multi-array $X$ whose axes are allowed to be permuted we say ``$X$ is block diagonal with $NB(X)$ blocks up to permutations'' where
\begin{equation*} \label{eq:num_blocks_multiarray}
\text{NB}(X) :=\max \{B | \text{the axes of } X \text{ can be permuted to create a } B \text{ block, block diagonal multi-array}\}
\end{equation*}
Any permutation of the axes of $X$ which achieves the above maximum is called a \textit{maximally block diagonal permutation}.
The multi-arrays in Figures \ref{fig:block_diag_multiarray} and \ref{fig:block_diag_multiarray_perm} both have three blocks up to permutations.

Next we construct a graph that captures the permutation invariant block diagonal structure of a multi-array.
Let $G(X)$ be a weighted, $V$-partite graph whose vertex set is
$$
\mathcal{V} := \{ (v, k) | k \in  [d\vs{v}] \text{ and } v \in [V] \}
$$
i.e. the generalization of rows and columns to multi-arrays.
There can only be an edge between two vertices on different axes i.e. $(a, j\vs{a})$ and $(b, j\vs{b})$ where $a \neq b$. 
The weight of such an edge is given by
$$
A_{(a, j\vs{a}), (b, j\vs{b})} (X) := \sum_{k\vs{1}=1}^{K\vs{1}} \dots  \sum_{k\vs{a-1}=1}^{K\vs{a-1}}   \sum_{k\vs{a+1}=1}^{K\vs{a+1}}  \dots  \sum_{k\vs{b-1}=1}^{K\vs{b-1}}  \sum_{k\vs{b+1}=1}^{K\vs{b+1}} \dots  \sum_{k\vs{V}=1}^{K\vs{V}} X_{k\vs{1}, \dots, j\vs{a}, \dots, j\vs{b}, \dots k\vs{V}} 
$$
i.e. summing over all entries of $X$ where the $a$th axis is fixed at $j\vs{a}$ and the $b$th axis is fixed at $j\vs{b}$.
Here $A(X) \in \mathbb{R}^{\sum_{v=1}^V d\vs{V} \times \sum_{v=1}^V d\vs{V} }$ is the adjacency matrix of $G(X)$.
This adjacency matrix $A(X)$ is equivalent to the hypergraph adjacency matrix given in \citet{zhou2007learning}.

The edge $\{(a, j_a), (b, j_b)\}$ is present in $G(X)$ if and only if there there is a tuple $(k\vs{1}, \dots, j\vs{a}, \dots, j\vs{b}, \dots, k\vs{V})$ such that $X_{k\vs{1}, \dots, j\vs{a}, \dots, j\vs{b}, \dots, k\vs{V}} \neq 0$.
For example in Figure \ref{fig:block_diag_multiarray} the edge $\{(2, 4), (3, 4)\}$ is present, but the edge $\{(2, 4), (3, 5) \}$ is not.
A vertex $(v, k)$ is isolated if $X_{j\vs{1}, \dots, j\vs{V}} = 0$ for all $\{j\vs{1}, \dots, j\vs{V} | j\vs{v} = k\}$ i.e. there is a $V-1$ dimensional slice of zeros (e.g. $(1, 6)$ is the only isolated vertex in Figure \ref{fig:block_diag_multiarray}).

The symmetric Laplacian of this graph captures the block diagonal structure of $X$ as follows.
\begin{proposition} \label{prop:sym_lap_block_diag_multiarray}
The following are equivalent for $1 \le B +  \sum_{v=1}^V Z\vs{v} \le \min(d\vs{1}, \dots, d\vs{V})$
\begin{enumerate}

\item $X$ is block diagonal up to permutations with $B$ blocks and has  $Z\vs{v}$ $V-1$ dimensional slices of zeros on the $v$th axis.

\item $G(A)$ has $B$ connected components with at least two vertices and $ \sum_{v=1}^V Z\vs{v}$ isolated vertices.

\item $L_{\text{sym}}(A(X))$ has exactly $B$ eigenvalues equal to 0.

\item $L_{\text{un}}(A_{\text{bp}}(X))$ has exactly $B + \sum_{v=1}^V Z\vs{v} $ eigenvalues equal to 0.

\end{enumerate}
Additionally, the number of eigenvalues equal to 1 of the symmetric Laplacian is at least $\sum_{v=1}^V Z\vs{v}$.  

\end{proposition}

We now have that Problem \eqref{eq:multi_array_blck_diag_opt_problem} is equivalent to
\begin{equation} \label{eq:sym_lap_constr_problem_multiarray}
\begin{aligned}
&\underset{X\in \mathbb{R}^{d\vs{1} \times \dots \times d\vs{v}}}{\text{minimize}}   & & \ell(X)\\
& \text{subject to} & &  X \ge 0 \text{ and } L_{\text{sym}}(A(X)) \text{ has } B \text{ eigenvalues equal to 0}.
\end{aligned}
\end{equation}
Following Section \ref{ss:block_diag_opt}, solve the related problem 
\begin{equation} \label{prob:bd_multiarray_extremal_rep}
\begin{aligned}
& \underset{X\in \mathbb{R}^{d\vs{1} \times \dots \times d\vs{V}}, U\in \mathbb{R}^{\sum_{v=1}^V d\vs{v} \times B}}{\text{minimize}}   & & f(X) +   \alpha \text{Tr} \left( U^T L_{\text{un}}(A(X)) U\right)   \\
& \text{subject to} & & X \ge 0 \text{ and } U^T \text{diag}(\text{deg}(A(X))) U = I_B,
\end{aligned}
\end{equation}
for a sufficiently large value of $\alpha$.
Note that $A(\cdot)$ is a linear function so the second term in the objective and the constraints are linear in $A$.
An alternating algorithm similar to the one discussed in Section \ref{ss:sym_lap_pen_alt_algo} can be used to solve this problem.


\section{Extremal characterization of weighted sums of generalized eigenvalues}  \label{s:extremal}

For a pair of symmetric matrices $A, B \in \mathbb{R}^{n \times n}$ denote the matrices whose columns are the largest generalized $K$ eigenvectors of $(A, B)$ by
$$\mathcal{GE}_K(A, B) := \{ U |  A  U_k  = \lambda_k B U_k \text{ for } k \in [K],  U^T B U = I_K\} \subseteq \mathbb{R}^{n \times K} $$
where $\lambda_1 \ge \dots \ge \lambda_K$ are the largest generalized eigenvalues of $(A, B)$.
Similarly, let $\mathcal{GE}_{(K)}(A, B)$ be the analogous set for the smallest $K$ generalized eigenvalues.
The following proposition shows that the generalized eigenvalues of $(A, B)$ can still be well defined when $B$ has a non-trivial kernel and can be ordered (since they are real).
\begin{proposition} \label{prop:gevals_well_def_sing_B}
If $\text{ker}(B) \subseteq \ker(A)$ and $m = n - \text{dim}(\text{ker}(B))$ then $(A, B)$ has $m$ real generalized eigenvalues.
These generalized eigenvalues are given by the eigenvalues of $B^{-1/2} A B^{-1/2}$ excluding the eigenvalues whose eigenvectors live in the kernel of $B$ where the inverse is taken to be the Moore-Penrose pseudo inverse.
 \end{proposition}

We adapt a Proposition from \citet{marshall1979inequalities} to obtain an extremal characterization for weighted sums of the largest (smallest) generalized eigenvalues.
This is a generalization of the famous Fan's theorem \citep{fan1949theorem}.
\begin{proposition} \label{prop:weighted_geval_extremal_representation}
Let $A, B \in \mathbb{R}^{n \times n}$ be symmetric.
Assume $B$ is positive semi-definite, $\text{ker}(B) \subseteq \text{ker}(A)$ and $K \le n - \text{dim}(\text{ker}(B))$.
If $w \in \mathbb{R}^K$, such that $w_1 \ge w_2 \ge \dots \ge w_K $, then
\begin{equation} \label{eq:weighted_sum_largest_eval_extremal}
\begin{aligned}
& \sum_{j=1}^K w_j \lambda_{j}(A, B)  = &    & \underset{U\in \mathbb{R}^{n \times K}}{\text{maximum}}  & &  \text{Tr}\left( U^T A U \text{diag}(w) \right)\\ 
& & & \text{subject to} & &  U^T B U = I_K,
\end{aligned}
\end{equation}
where the maximum is attained by any matrix in $U \in \mathbb{R}^{n \times K} \in \mathcal{GE}_K(A, B)$.
Similarly,
\begin{equation}\label{eq:weighted_sum_smallest_eval_extremal}
\begin{aligned}
& \sum_{j=1}^K w_j \lambda_{(j)}(A, B)   = &    & \underset{U\in \mathbb{R}^{n \times K}}{\text{minimum}}  & &  \text{Tr}\left( U^T A U \text{diag}(w) \right)\\ 
& & & \text{subject to} & &  U^T B U = I_K,
\end{aligned}
\end{equation}
where the minimum is attained by any matrix  $U \in \mathbb{R}^{n \times K} \in \mathcal{GE}_{(K)}(A, B)$.
\end{proposition}

This proposition allows $U$ to be low rank (as opposed to the full $n \times n$ matrix), permits weighted sums of generalized eigenvalues and says $B$ does not have to be full rank.
Note $w$ is allowed to have negative entries which may be of interest for some applications (e.g. a penalty that encourages some eigenvalues to be large).

Proposition \ref{prop:gevals_well_def_sing_B} shows how the eigenvalues of $L_{\text{sym}}(A_{\text{bp}} (X))$ are related to the generalized eigenvalues of $(L_{\text{un}}(A_{\text{bp}} (X)),\text{diag}( \text{deg}(A_{\text{bp}} (X))))$; the latter are the subset of the former excluding the one eigenvalues that come from degree zero nodes.
 The following corollary shows how problems \eqref{prob:bd_eval_pen}  and \eqref{prob:bd_extremal_rep} and are related; they are the same so long as the solution does not have too many rows/columns of zeros.
 \begin{corollary}  \label{cor:sym_evals_and_gevals}
 Let $X \in \mathbb{R}^{R \times C}+$ and $1 \le K \le \min(R, C)$.
 Let $\widetilde{R}$ denote the number of non-zero rows of $X$ (similarly for $\widetilde{C}$).
 If $K \le \widetilde{R} + \widetilde{C}$ then
 $$
 \lambda_{(k)}\left( L_{\text{sym}}(A_{\text{bp}} (X))\right) = \lambda_{(k)} \left (L_{\text{un}}(A_{\text{bp}} (X)),\text{diag}( \text{deg}(A_{\text{bp}} (X)))\right)
 $$
 for each $1 \le k \le K$. 
  \end{corollary}


\section{Alternating algorithm for block diagonal constraints}  \label{s:wnn}

This section considers the following weighted nuclear norm regularized problem for some $1 \le K \le \min(R, C)$,
\begin{equation}\label{prob:sym_lap_eval_pen} 
\begin{aligned}
& \underset{X\in \mathbb{R}^{R \times C}}{\text{minimize}}   & & f(X) +   \alpha \sum_{j=1}^K  w_j \lambda_{(j)} \left( L_{\text{sym}}(A_{\text{bp}}(X))\right)   \\
& \text{subject to} & &  X \ge 0, 
\end{aligned}
\end{equation}
where $w \in \mathbb{R}^{K}_+$ is a positive weight vector with $w_1 \ge \dots \ge w_K$ e.g. $w_k \propto \frac{1}{k}$.
Problem \eqref{prob:bd_eval_pen} is of course recovered by setting $w = \mathbf{1}_K$.
 
The mild generalization of \eqref{prob:bd_eval_pen} allows us to put more weight on smaller eigenvalues, which can lead to better estimators  \citep{chen2013reduced, gu2014weighted}.
In some applications one might also want to consider \eqref{prob:sym_lap_eval_pen} where $\alpha$ is the (continuous) hyper-parameter controlling the amount of block diagonal regularization instead of \eqref{prob:bd_eval_constr}, which has the (discrete) hyper-parameter of the number of blocks.
We implemented this idea in the context of the block diagonal MVMM simulations in Section \ref{s:simulations}.
Unfortunately, this simulation leads to a null result; we found that continuous block diagonal regularization (e.g. where $w$ is exponentially or polynomially decaying) was not faster or more accurate than the block diagonally constrained method of Section \ref{ss:mvmm_block_diag}.

The following proposition makes the connection between the constrained \eqref{prob:bd_constr} and  \eqref{prob:bd_eval_pen}/\eqref{prob:bd_extremal_rep}.
While global or local solutions to \eqref{prob:bd_constr} are difficult to find,  these points are typically contained in a larger set of points (global or local solutions to \eqref{prob:bd_extremal_rep}), which are easier to find.

\begin{proposition} \label{prop:soln_relations_local}

\begin{enumerate}

\item Suppose $X$ is a global (local) solution of \eqref{prob:bd_eval_pen} such that $ \sum_{j=1}^B  \lambda_{(j)} \left( L_{\text{sym}}(A_{\text{bp}}(X))\right) = 0$. Then
 $X$ is a  global (local) solution of \eqref{prob:bd_constr}.

\item 
Suppose $X$ is a global (local) solution of \eqref{prob:bd_eval_pen} such that the largest $B$ row sums and the largest $B$ column sums of $X$ are strictly positive.
Then there exists a coordinate-wise minimizer $U$  such that $(X, U)$ is a global (local) minimizer of \eqref{prob:bd_extremal_rep}.
\end{enumerate}
\end{proposition}

The first claim shows that if we can find a solution to the penalized Problem \eqref{prob:bd_eval_pen} that is block diagonal (i.e. $\alpha$ is large enough to induce the rank constraint) then we have a solution to \eqref{prob:bd_constr}.
The second claim shows the solutions of \eqref{prob:bd_eval_pen} are typically\footnote{If $\alpha$ is large, the solutions to  \eqref{prob:bd_eval_pen} typically satisfy the column/row sum condition in the second claim (since zero rows/columns give large eigenvalues of 1 by Proposition \ref{ss:proof_sym_lap_spect_bd}).} also solutions to the extremal representation problem \eqref{prob:bd_extremal_rep}; the latter are easier for our algorithm to locate.

Section \ref{ss:sym_lap_pen_alt_algo} presents an alternating algorithm for \eqref{prob:sym_lap_eval_pen} and Section \ref{ss:algo_convergence} discusses convergence properties of this algorithm.

\subsection{Alternating algorithm for \eqref{prob:sym_lap_eval_pen}} \label{ss:sym_lap_pen_alt_algo}

Following Proposition \ref{eq:weighted_sum_largest_eval_extremal}, we reformulate  the weighted nuclear norm problem \eqref{prob:sym_lap_eval_pen}  as
\begin{equation} \label{prob:sym_lap_pen_extremal_rep}
\begin{aligned}
& \underset{X\in \mathbb{R}^{R \times C}, U\in \mathbb{R}^{(R + C) \times K}}{\text{minimize}}   & & f(X) +   \alpha \text{Tr} \left( U^T L_{\text{un}}(A_{\text{bp}}(X)) U \text{diag}(w) \right)   \\
& \text{subject to} & & X \ge 0, U^T \text{diag}(\text{deg}(A_{\text{bp}}(X))) U = I_K.
\end{aligned}
\end{equation}

\subsubsection{$U$ subproblem}

For fixed $X$, the $U$ subproblem  in \eqref{prob:sym_lap_pen_extremal_rep} is a generalized eigen-problem. 
Corollary \ref{cor:eigen_subproblem} shows a global solution of this problem can be obtained through a low rank SVD of a smaller matrix.
For $X \in \mathbb{R}^{R \times C}_+$ let
\begin{equation}\label{eq:tsym}
T_{\text{sym}}(X) := \text{diag}(X \mathbf{1}_C)^{-1/2} X \text{diag}(X^T \mathbf{1}_R)^{-1/2} \in \mathbb{R}^{R + C}.
\end{equation}
When $X$ has 0 rows or columns the inverse is taken to be the Moore-Penrose pseudo-inverse.
Note this matrix is the upper diagonal elements of $L_{\text{sym}}(A_{\text{bp}}(X))$.

The Lasso penalty in the $X$ update discussed below can lead to exact zeros and, in principle, can introduce some rows/columns of zeros.
The following corollary shows that the $U$ update can handle the case when $X$ has some rows/columns of 0s.
Note that if the initial value of Algorithm \ref{algo:wnn_alt_algo} satisfies the condition $K \le \widetilde{R}  + \widetilde{C}$ then the output of each successive step will also satisfy this condition.

\begin{corollary} \label{cor:eigen_subproblem}
For $X \in \mathbb{R}^{R \times C}_+$ consider the following problem,
\begin{equation} \label{prob:geig}
\begin{aligned}
& \underset{U\in \mathbb{R}^{(R + C) \times K}}{\text{minimize}}   & & \text{Tr} \left( U^T L_{\text{un}}(A_{\text{bp}}(X)) U \text{diag}(w) \right)   \\
&  \text{subject to} & & U^T \text{diag}(\text{deg}(A_{\text{bp}}(X))) U = I_K\\
\end{aligned}
\end{equation}
for some $K \le R + C$.

\textbf{Case 1}: Suppose $X$ has no rows or columns of zeros.
Let $U_{\text{left}} \in \mathbb{R}^{R \times \min(R, C)}$ and  $U_{\text{right}} \in \mathbb{R}^{C \times \min(R, C)}$ be the matrix of the left and right singular vectors of $T_{\text{sym}}(X)$.
Let $U_{\text{left}, j}$ denote the left singular vector corresponding to the $j$th largest singular value and let  $U_{\text{left}, (j)}$ denote the left singular vector corresponding to the $j$th smallest singular value.
If $R \ge C$ let $Q \in \mathbb{R}^{R \times (\max(R, C) - \min(R, C))}$ be a orthonormal basis matrix of $\text{col-span}(U_{\text{left}} )^{\perp}$. 
If $R < C$ let $Q \in \mathbb{R}^{C \times (\max(R, C) - \min(R, C))}$ be a orthonormal basis matrix of $\text{col-span}(U_{\text{right}})^{\perp}$.

Let the columns of $U^* \in \mathbb{R}^{(R + C) \times K}$ be given by $U^*_k = \text{diag}(\text{deg}(A_{\text{bp}}(X)))^{-1/2} \Xi$ where
\begin{equation}\label{eq:geig_soln}
\Xi = 
\begin{cases}
\begin{bmatrix} U_{\text{left}, k} \\ U_{\text{right}, k}  \end{bmatrix}  & 1 \le k \le \min(R, C)\\[1pt]\\
\begin{bmatrix} Q_j \\ \mathbf{0}_C  \end{bmatrix}  & k = \min(R, C) + j, \text{ for } 1 \le j \le \max(R, C) - \min(R, C), \text{ and } R \ge C \\[1pt]\\
\begin{bmatrix}  \mathbf{0}_R \\ Q_j  \end{bmatrix}  & k = \min(R, C) + j, \text{ for } 1 \le j \le \max(R, C) - \min(R, C), \text{ and } R < C \\[1pt]\\
\begin{bmatrix} U_{\text{left}, (j)} \\ - U_{\text{right}, (j)}  \end{bmatrix}  & k = \max(R, C) + j, \text{ for } j \ge 1.
\end{cases}
\end{equation}
Then $U^*$ is a global minimizer of \eqref{prob:geig}.

\textbf{Case 2}: Suppose $X$ has $ \widetilde{R}$ and $ \widetilde{C}$ non-zero rows and columns and $K \le  \widetilde{R} +  \widetilde{C}$.
Let $\widetilde{X} \in \mathbb{R}^{ \widetilde{R} \times  \widetilde{C}}$ denote $X$ after removing the zero rows and columns and let $\widetilde{U}$ be the solution obtained using \eqref{eq:geig_soln} applied to $\widetilde{X}$.
Then a global solution of \eqref{prob:geig} can be obtained by adding appropriate zero rows to $\widetilde{U}$.

\end{corollary}

\subsubsection{$X$ subproblem}  \label{ss:var_subproblem}

For fixed $U$, the constraints and second term in the objective of Problem \eqref{prob:sym_lap_pen_extremal_rep} are linear in $X$.
Let matrix $M(U, w) \in \mathbb{R}^{R \times C}$ be the matrix such that
$$\text{Tr} \left( U^T L_{\text{un}}(A_{\text{bp}}(X)) U  \text{diag}(w)\right) = \langle X, M(U, w) \rangle.$$
Writing $U = \begin{bmatrix}U_{\text{rows}}\\ U_{\text{col}} \end{bmatrix}$ where $U_{\text{rows}} \in \mathbb{R}^{R \times K}$ and $U_{\text{cols}} \in \mathbb{R}^{C \times K}$, we see the $r, c$th element of $M(U, w)$ is 
\begin{equation} \label{eq:M}
[M(U, w)]_{rc} = || \text{diag}(w)^{1/2} \left( U_{\text{rows}}(r, :) - U_{\text{cols}}(c, :) \right) ||_2^2.
\end{equation}

Let 
\begin{equation} \label{eq:c_diag}
c_{\text{diag}}(U) := U \odot U \in \mathbb{R}^{(R + C) \times K}
\end{equation}
be the matrix whose elements are the squares of $U$; this matrix gives the diagonal elements of the linear equality constraint of \eqref{prob:sym_lap_pen_extremal_rep}.
Also let 
\begin{equation}\label{eq:c_utri}
c_{\text{utri}}(U) \in \mathbb{R}^{(R + C) \times {K \choose 2}} \text{ be the matrix whose columns are given by } U_{\ell} \odot U_{j}, 1 \le \ell < j \le K.
\end{equation}
This matrix gives the upper-triangular of the linear equality constraints of \eqref{prob:sym_lap_pen_extremal_rep}; the lower triangular constraints are redundant.
Note some of the constraints of $c_{\text{utri}}(U)$ may be redundant\footnote{E.g. when $A_{\text{bp}}(X)$ has multiple connected components Proposition \ref{prop:sym_lap_spectrum_ccs} gives one source of redundancy.} and can be removed to improve numerical performance.

For fixed $U$, the $X$ subproblem for Problem \eqref{prob:sym_lap_pen_extremal_rep} is given by
\begin{equation} \label{prob:wnn_x_update}
\begin{aligned}
&\underset{X}{\text{minimize}}   & & f(X) +\alpha  \langle X, M(U, w) \rangle\\
& \text{subject to} & & X \ge 0\\
& & &  \begin{bmatrix} c_{\text{diag}}(U)^T  \\ c_{\text{utri}}(U)^T \end{bmatrix} \text{diag} (\text{deg} ( A_{\text{bp}}(X)) )=   \begin{bmatrix} \mathbf{1}_{K} \\ \mathbf{0}_{ {K \choose 2}}  \end{bmatrix}.
\end{aligned}
\end{equation}

If $f$ is convex then \eqref{prob:wnn_x_update} is a convex problem because the second term in the objective and the constraints are linear.
Because $X$ is constrained to be positive, the second term in the objective puts a weighted lasso penalty on the entries of $X$ whose weights are given by $M(U, w)$.

For complicated objective functions (e.g. the log-likelihood of a mixture model) the full $X$ updates may be computationally intractable.
We therefore consider surrogate updates obtained by replacing $f$ with a surrogate function that has the same first order behavior \citep{razaviyayn2013unified}.
\begin{definition} \label{def:surrogate_fcn}
A surrogate function $Q(X | Y)$ satisfies $Q(X | X) = f(X)$, $Q(X | Y) \ge f(Y)$, $Q(\cdot | Y)$ is continuous and assume $\frac{d}{dX}Q(X | Y) \big |_{X=Y} = \frac{d}{dX} f(X) \big |_{X=Y}  $ for all $X, Y$.
\end{definition}
Given the current guess, $X_{\text{current}}$, we update $X$ by solving the following problem,
\begin{equation} \label{prob:wnn_x_update_maj}
\begin{aligned}
&\underset{X}{\text{argmin}}   & & Q(X | X_{\text{current}}) + \alpha  \langle X, M(U, w) \rangle\\
& \text{subject to} & & X \ge 0 \\ 
& & &  \begin{bmatrix} c_{\text{diag}}(U)^T  \\ c_{\text{utri}}(U)^T \end{bmatrix} \text{diag}( \text{deg} ( A_{\text{bp}}(X)) )=   \begin{bmatrix} \mathbf{1}_{K} \\ \mathbf{0}_{ {K \choose 2}}  \end{bmatrix}.
\end{aligned}
\end{equation}

\subsubsection{Alternating algorithm}

Let $\textsc{Update-X}(X_{\text{current}}, U)$ be an algorithm that solves either the full update \eqref{prob:wnn_x_update} or the surrogate update \eqref{prob:wnn_x_update_maj}.

\begin{algorithm}[H] \label{algo:wnn_alt_algo} 
\DontPrintSemicolon
  
\KwInput{ $\alpha \ge 0$, $K \le \min(R, C)$, $w \in \mathbb{R}^{K}_+$}
\KwOutput{$X$}
  
Initialize $X^0$.
  
\While{Stopping criteria not satisfied}
{
$U^{s+1} \leftarrow $ smallest $K$ generalized eigenvectors of \tcp*{Computed as in Corollary \ref{cor:eigen_subproblem}} 
\begin{equation}\label{eq:wnn_alt_algo_eig}
\left( L_{\text{sym}} (A_{\text{bp}}(X^{s})), \text{diag}(\text{deg}(A_{\text{bp}}(X^s))) \right)
\end{equation}



$X^{s+1} \leftarrow \textsc{Update-X}(X^s, U^{s+1})$ 
  
$ s \leftarrow s + 1$
  
}
\caption{Alternating algorithm for the weighted nuclear norm Problem \eqref{prob:sym_lap_pen_extremal_rep}}
\end{algorithm}

\begin{remark}
If \textsc{Update-X} solves either \eqref{prob:wnn_x_update} or \eqref{prob:wnn_x_update_maj}, each step of this algorithm decreases the objective function of \eqref{prob:sym_lap_pen_extremal_rep} (and \eqref{prob:sym_lap_eval_pen}).
In the former case, Algorithm \ref{algo:wnn_alt_algo} is an alternating minimization algorithm while in the latter case it is a \textit{block successive upper bound minimization algorithm} with coupled constraints between blocks \citep{razaviyayn2013unified}.
\end{remark}

\subsubsection{Algorithm Intuition}

The second term in \eqref{prob:wnn_x_update} puts a weighted lasso penalty on the entries of $X$.
These weights, which come from \eqref{eq:M}, encourage $X$ to be more block diagonal.

Suppose $X_{\text{current}} \in \mathbb{R}^{R + C}_+$ is exactly block diagonal up to permutations with $B$ blocks, $K=B$ and $w = \mathbf{1}_B$.
Let $\mathbf{1}_{A_1}, \dots, \mathbf{1}_{A_B} \in \mathbb{R}^{R + C}$ denote the indicator vectors of the blocks and let $d_b = \mathbf{1}_{R + C}^T \mathbf{1}_{A_b}$ be the total degree of the $b$th block for each $b \in [B]$.
By Proposition \ref{prop:sym_lap_spectrum_ccs},
$$U = \begin{bmatrix} \frac{1}{\sqrt{d_1}} \mathbf{1}_{A_1} & \dots & \frac{1}{\sqrt{d_B}} \mathbf{1}_{A_B} \end{bmatrix}$$
is a $U$ global minimizer of \eqref{prob:sym_lap_pen_extremal_rep}.
In this case
\begin{equation*}
[M(U, \mathbf{1}_B)]_{r, c} = 
\begin{cases}
0  & \text{if row } r \text{ and column } c \text{ are in the same block}\\
\frac{1}{d_{\text{row}}(r)} + \frac{1}{d_{\text{col}}(c)}  & \text{if row } r \text{ and column } c \text{ are in different blocks},
\end{cases}
\end{equation*}
where $d_{\text{row}}(r) = d_b$ where the $r$th row belongs to the $b$th block (similarly for $d_{\text{col}}(r)$).
The second term in \eqref{prob:wnn_x_update} only penalizes edges that go between blocks and does not penalize edges within a block.

If $X_{\text{current}}$ has a row or column of zeros the corresponding eigenvalue of the symmetric Laplacian will be 1 (i.e. large) and the corresponding eigenvector will not be included in the smallest $K$ eigenvectors that comprise $U$.
Therefore, the algorithm does not want to encourage rows/columns of zeros.

\subsection{Convergence of Algorithm \ref{algo:wnn_alt_algo} } \label{ss:algo_convergence}

We show Algorithm \ref{algo:wnn_alt_algo} converges to a coordinate-wise minimizer when \textsc{Update-X} does a full update by solving \eqref{prob:wnn_x_update}.
If \textsc{Update-X} does a surrogate update solving \eqref{prob:wnn_x_update_maj} Algorithm \ref{algo:wnn_alt_algo} converges to a coordinate-wise stationary point (defined below).
The non-linear coupled constraints of Problem \eqref{prob:wnn_x_update_maj} make the convergence analysis tricky e.g. the BSUM framework does not apply \citep{razaviyayn2013unified}.

Consider a constrained optimization problem with two blocks of variables; let $f(x, y)$ be the objective function and let $g(x, y), h(x, y)$ be vector valued functions corresponding to the equality and inequality constraints (all functions are assumed to be continuous).
Let $I(x, y)$ denote the indicator function of the constraint set $\{ (x, y) | g(x, y) = 0, h(x, y) \le 0 \}$.
Recall a \textit{stationary point} of an optimization problem is one that satisfies the KKT conditions \citep{boyd2004convex}; assuming appropriate constraint qualification all local minimizers are stationary points.

\begin{definition}
Let $\mathcal{L}(y) := \{ x* | x^* \text{ is a local minimizer of } \underset{x}{minimize}\; f(x, y) + I(x, y)\}$ be the set of $x$ coordinate local minimizers for fixed $y$.
Let $\mathcal{S}(y) := \{x^* | x^*$ is a stationary point of $ \underset{x}{minimize}\; f(x, y) + I(x, y) \}$ be the set of $x$ coordinate stationary points for fixed $y$.
Let $\mathcal{G}(x) := \{ y* | y^* \in \underset{y}{argmin}\; f(x, y) + I(x, y)\}$ be the set of $y$ coordinate global minimizers for fixed $x$.
Finally let,
\begin{equation}
\mathcal{LG} := \{ (x, y) | x \in \mathcal{L}(y), y \in \mathcal{G}(x)\}
\end{equation}
\begin{equation}
\mathcal{SG} := \{ (x, y) | x \in \mathcal{S}(y), y \in \mathcal{G}(x)\}
\end{equation}
denote the set of $x, y$ pairs where $x$ is a coordinate-wise local minimizer (stationary point) and $y$ is a coordinate-wise minimizer.
\end{definition}

\begin{assumption} \label{assumpt:basic}
Assume the objective function $f: \mathbb{R}^{R \times C}_+ \to \mathbb{R}$ is continuous and the level set $S_{X^0} :=  \{X | f(X) \le f(X^0) + \alpha w^T \mathbf{1}_K \}$ is compact where $X^0$ is the point at which the algorithm is initialized. 
\end{assumption}

\begin{assumption} \label{assumpt:degree_lower_bound}
Assume there exists an $\eta > 0$ such that the iterates, $X^s$, of Algorithm \ref{algo:wnn_alt_algo} are contained in the set $\mathcal{RC}_{\eta}  := \{X | \text{deg}(A_{\text{bp}}(X)) \ge  \eta \mathbf{1}_{R + C} \}$ for large enough $s$.
\end{assumption}

This technical assumption typically hold in practice since the algorithm does not encourage rows/columns of zeros as discussed above.
Alternatively, this assumption can be enforced by adding the linear constraints $\text{deg}(A_{\text{bp}}(X)) \ge \eta \mathbf{1}_{R + C}$ to   \eqref{prob:sym_lap_eval_pen} and \eqref{prob:sym_lap_pen_extremal_rep}.
The updates for the algorithm still work even when some rows/columns of $X$ are identically zero as long as there are at least $K$ total non-zero rows/columns at each step\footnote{We lose  the convergence guarantees, however, because the sequence is not guaranteed to be in a compact set.}.

\begin{assumption} \label{assumpt:x_update}
Assume one of the following,
\begin{enumerate}
\item $\textsc{Update-X}$ returns a global minimizer of the full update problem \eqref{prob:wnn_x_update}.
 
\item There exists a surrogate function $Q$ and $\textsc{Update-X}$ returns a global minimizer of the surrogate update problem \eqref{prob:wnn_x_update_maj}.
\end{enumerate}
\end{assumption}

\begin{proposition} \label{prop:conv}
Let $\{X^s, U^s \}_{s=1}^{\infty}$ be any sequence generated by Algorithm \ref{algo:wnn_alt_algo} and suppose Assumptions \ref{assumpt:basic} and \ref{assumpt:degree_lower_bound} hold.
Under Assumption \ref{assumpt:x_update}.1  all limit points of $\{X^s, U^s\}_{s=1}^{\infty}$ are elements of $\mathcal{LG}$.
Under Assumption \ref{assumpt:x_update}.2, all limit points of $\{X^s, U^s\}_{s=1}^{\infty}$ are elements of $\mathcal{SG}$.
In addition, $\lim_{s \to \infty}   f(X^s) +   \alpha \sum_{j=1}^K  w_j \lambda_{(j)} \left( L_{\text{sym}}(A_{\text{bp}}(X^s))\right)  \to  f(X^*) +   \alpha \sum_{j=1}^K  w_j \lambda_{(j)} \left( L_{\text{sym}}(A_{\text{bp}}(X^*))\right) $ for all limit points $X^*$.
\end{proposition}

Proposition \ref{prop:conv} does not guarantee we find a minimizer of \eqref{prob:sym_lap_eval_pen}.
The following proposition shows these local minimizers are contained in the solution set we actually are guaranteed to find.
\begin{proposition} \label{prop:soln_set_containment}
Suppose $X$ is a local minimizer of \eqref{prob:sym_lap_eval_pen}.
Then $ \text{ there exists a } U  \text{ such that } (X, U) \in \mathcal{LG} \subseteq \mathcal{SG}$, where $ \mathcal{LG}$ and $\mathcal{SG}$ correspond to Problem \eqref{prob:sym_lap_pen_extremal_rep}.
If $X$ is a global minimizer of \eqref{prob:sym_lap_eval_pen}, then there exists a $U$ such that $(X, U)$ is a global minimizer of \eqref{prob:sym_lap_pen_extremal_rep}.

\end{proposition}


\section{Choice of Laplacian} \label{s:un_lap_bad}

Many existing approaches to imposing block diagonal constraints use the unnormalized Laplacian instead of the symmetric Laplacian \citep{feng2014robust, nie2016constrained, nie2017learning, lu2018subspace, kumar2019unified}.
This section shows that approaches based on the unnormalized Laplacian require stronger modeling assumptions and do not have computational advantages over our approach based on the symmetric Laplacian.

\begin{figure}[H]
\begin{subfigure}[t]{0.32\textwidth}
\includegraphics[width=\linewidth]{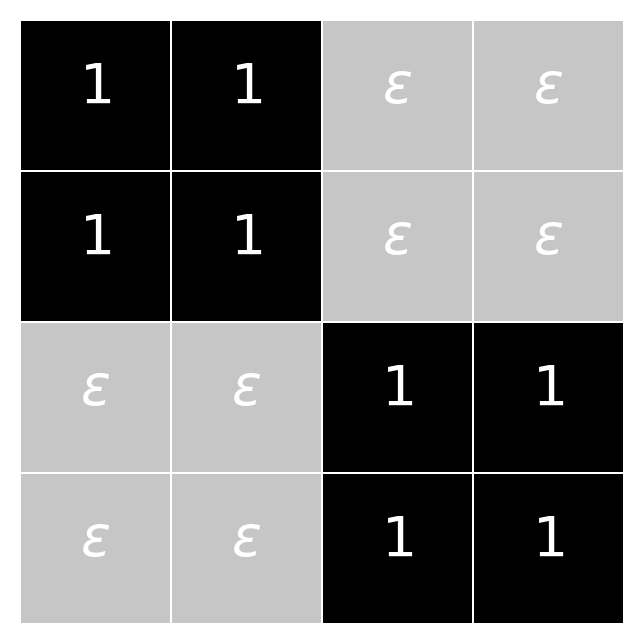} 
\caption{
$X_{\epsilon} \in \mathbb{R}^{4 \times 4}$ is the matrix with two $2\times 2$ blocks of ones and whose off-diagonal elements are equal to $\epsilon$.  
}
\label{fig:un_vs_sym__block_diag__mat}
\end{subfigure}
\hfill
\begin{subfigure}[t]{0.32\textwidth}
\includegraphics[width=\linewidth]{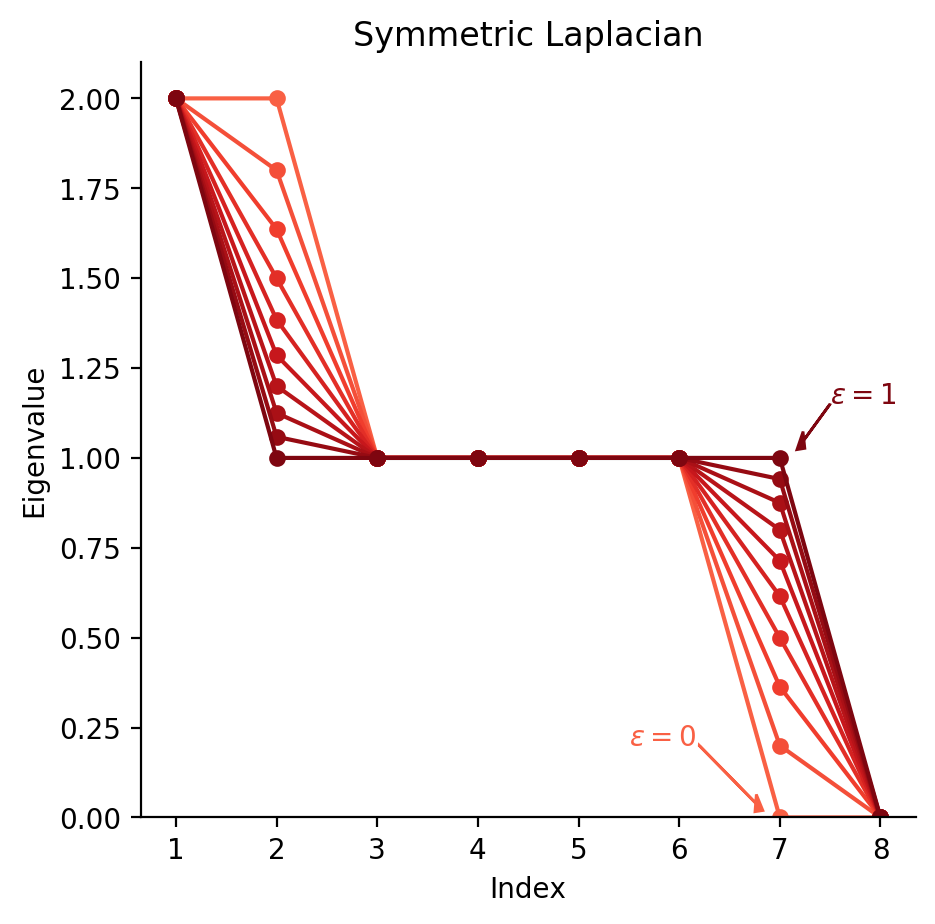} 
\caption{
Spectrum of $L_{\text{sym}}(A_{\text{bp}}(X_{\epsilon}))$ for a range of values of $\epsilon \in [0, 1]$.
}
\label{fig:un_vs_sym__block_diag__sym_evals}
\end{subfigure}
\hfill
\begin{subfigure}[t]{0.32\textwidth}
\includegraphics[width=\linewidth]{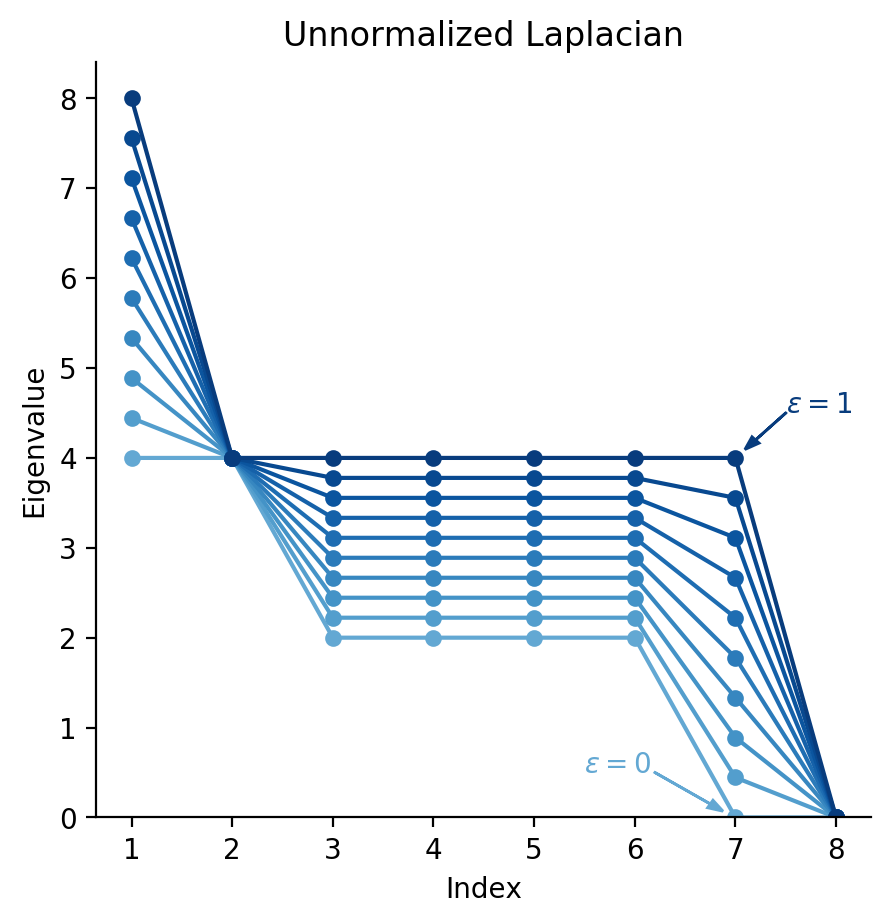}
\caption{
Spectrum of $L_{\text{un}}(A_{\text{bp}}(X_{\epsilon}))$ for a range of values of  $\epsilon \in [0, 1]$.
}
\label{fig:un_vs_sym__block_diag__un_evals}
\end{subfigure}%

\caption{
As $X_{\epsilon}$ approaches a 2 block, block diagonal matrix an eigenvalue of both the symmetric and unnormalized Laplacian approaches 0.
}
\label{fig:un_vs_sym_evals_block_diag}
\end{figure}

\begin{figure}[H]
\begin{subfigure}[t]{0.32\textwidth}
\includegraphics[width=\linewidth]{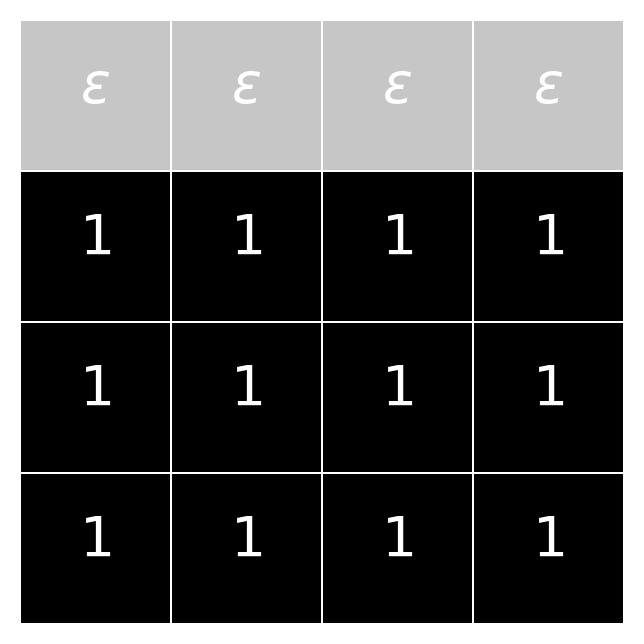} 
\caption{
$X_{\epsilon} \in \mathbb{R}^{4 \times 4}$ is the matrix whose first row is equal to $\epsilon$ and whose remaining elements are equal to 1.
}
\label{fig:un_vs_sym__zero_row__mat}
\end{subfigure}
\hfill
\begin{subfigure}[t]{0.32\textwidth}
\includegraphics[width=\linewidth]{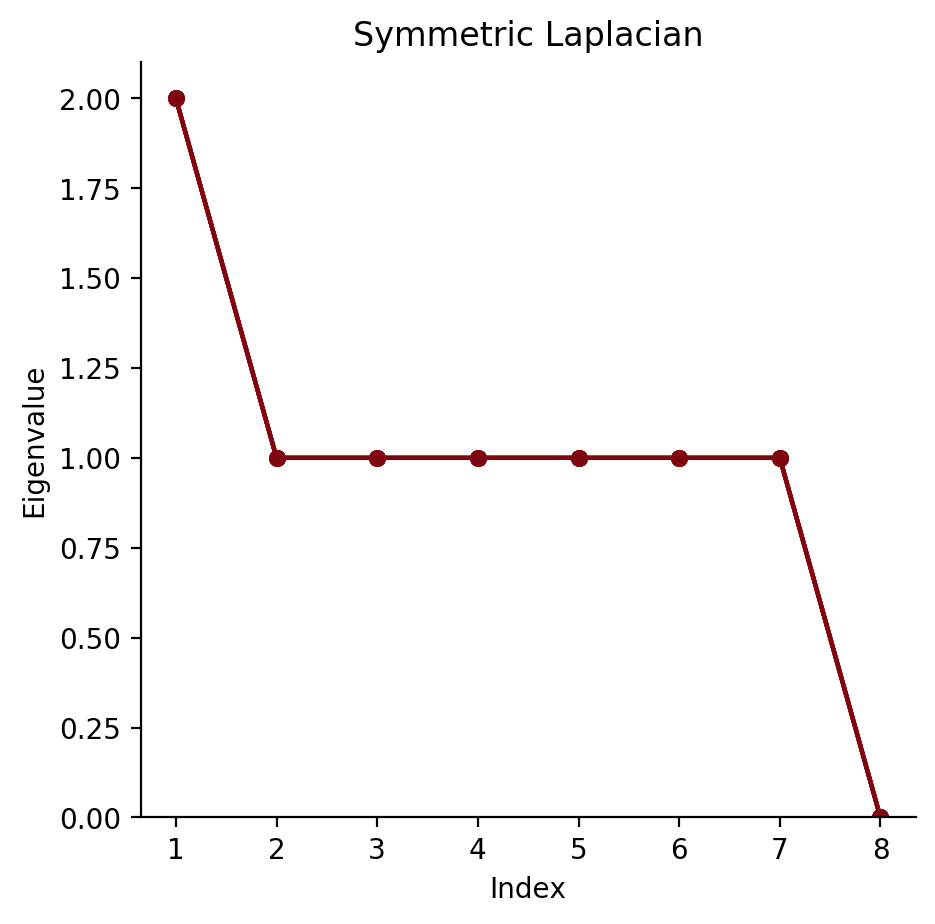} 
\caption{
Same as Figure \ref{fig:un_vs_sym__block_diag__sym_evals}.
Here the spectrum is the same for every value of $\epsilon$.
}
\label{fig:un_vs_sym__zero_row__sym_evals}
\end{subfigure}
\hfill
\begin{subfigure}[t]{0.32\textwidth}
\includegraphics[width=\linewidth]{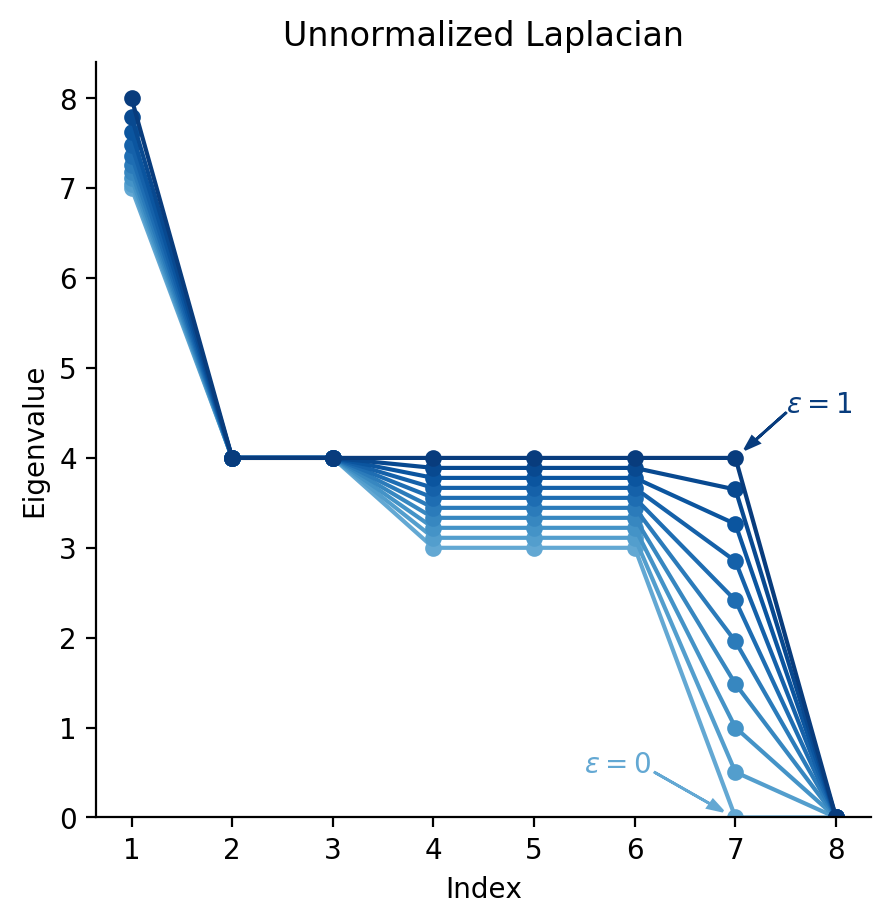}
\caption{
Same as Figure \ref{fig:un_vs_sym__block_diag__un_evals}.
}
\label{fig:un_vs_sym__zero_row__un_evals}
\end{subfigure}%

\caption{
As a row approaches 0, an eigenvalue of the unnormalized Laplacian approaches 0; the spectrum of the symmetric Laplacian is unaffected.
}
\label{fig:un_vs_sym_evals_zero_row}
\end{figure}

For any $X \in \mathbb{R}^{R \times C}_+$ Proposition \ref{prop:sym_lap_block_diag} shows
\begin{equation*} \label{eq:sym_lap_evals_equal_nb}
L_{\text{sym}}(A_{\text{bp}}(X)) \text{ has exactly }  B \text{ eigenvalues equal to } 0 \iff  X \text{ has exactly } B \text{ blocks up to permutations}
\end{equation*}
while
\begin{equation*}\label{eq:un_lap_evals_atmost_nb}
L_{\text{un}}(A_{\text{bp}}(X)) \text{ has exactly } B \text{ eigenvalues equal to } 0 \iff  X \text{ has at most } B \text{ blocks up to permutations}.
\end{equation*}

\begin{remark}
Consider replacing $L_{\text{sym}}(\cdot)$ with $L_{\text{un}}(\cdot)$ in \eqref{prob:bd_eval_constr}. 
We observed that in practice, using the unnormalized Laplacian for the block diagonal MVMM often leads to unsatisfactory solutions with too many rows/columns of 0s.
\end{remark}

Figures \ref{fig:un_vs_sym_evals_block_diag} and \ref{fig:un_vs_sym_evals_zero_row} illustrate the difference between $L_{\text{sym}}(A_{\text{bp}}(X))$ and $L_{\text{un}}(A_{\text{bp}}(X))$. 
When the symmetric Laplacian has small eigenvalues, then $X$ is close to block diagonal.
When the unnormalized Laplacian has small eigenvalues, $X$ is either close to block diagonal or has rows/columns of zeros.

It is easier to enforce the constraint ``at least $B$ eigenvalues are 0'' as opposed to exactly $B$ eigenvalues are 0.
For the symmetric Laplacian this inequality constraint leads to
\begin{equation*}
L_{\text{sym}}(A_{\text{bp}}(X)) \text{ has at least } B \text{ eigenvalues equal to } 0 \iff  X \text{ has at least } B \text{ blocks up to permutations}.
\end{equation*}
If the inequality constraint is placed on the eigenvalues of the unnormalized Laplacian we cannot make a corresponding statement about the block diagonal structure of the matrix.

One approach to ensuring the exact correspondence between the 0 eigenvalues of the unnormalized Laplacian and the block diagonal structure of $X$ is to constrain the degrees to be a known, non-zero constant.
Let $c \in \mathbb{R}^{R \times C}_+$, with $c >0$ then
\begin{equation}\label{eq:un_lap_evals_fixed_deg_exact_nb}
\begin{aligned}
& L_{\text{un}}(A_{\text{bp}}(X)) \text{ has exactly } B \text{ eigenvalues equal to } 0 \text{ and } \text{deg}(A_{\text{bp}}(X)) = c\\
 \iff \; &  X \text{ has at exactly } B \text{ blocks up to permutations and } \text{deg}(A_{\text{bp}}(X)) = c.
\end{aligned}
\end{equation}
Assuming the degrees are known allows one to use the unnormalized Laplacian, but requires stronger modeling assumptions.

Using the unnormalized Laplacian with the fixed degree constraint does not provide computational advantages over our approach based on the symmetric Laplacian.
Each step of the alternating algorithm for the symmetric Laplacian discussed in Section \ref{ss:sym_lap_pen_alt_algo} computes an eigen-decomposition then solves the linearly perturbed subproblem \eqref{prob:wnn_x_update}.
A similar algorithm for the unnormalized Laplacian can also be developed \citep{nie2016constrained}.
The eigen-decomposition for the unnormalized Laplacian requires computing the smallest $K$ eigenvectors of an $\mathbb{R}^{(R + C) \times (R + C)}$ matrix.
On  the other hand, the eigen-decomposition for the symmetric Laplacian can be obtained by computing the largest $K$ singular vectors of a smaller $\mathbb{R}^{R + C}$ matrix (Corollary \ref{cor:eigen_subproblem}). 
Additionally, when the fixed degree constraint is applied for the unnormalized Laplacian the corresponding linearly perturbed subproblem is in the same form as \eqref{prob:wnn_x_update}  (i.e. has linear constraints). 

Note that $L_{\text{sym}}(A_{\text{bp}}(\cdot))$ is not a continuous function near degree zero nodes due to the inverse so we have to be careful about how we use it.
In practice, we find this discontinuity is not a major issue and is not even present in the extremal formulation of the problem \eqref{prob:bd_extremal_rep}.
Minimizing the eigenvalues of the symmetric Laplacian tends not to encourage rows or columns to be zero, unlike the unnormalized Laplacian (see Figure \ref{fig:un_vs_sym_evals_zero_row}).


\section{EM algorithms for the multi-view mixture model}\label{ss:additional_mvm_em}

This section provides EM algorithms to fit the various multi-view mixture model problems described in the body of the paper.
Many of the computations (e.g. the E-step and the M-step for the cluster parameters $\Theta$) can be done using standard single-view mixture model algorithms.
This means we can base implementations of the MVMM EM algorithms off of pre-existing mixture modeling software such as sklearn \citep{pedregosa2011scikit}.

\subsection{EM algorithm for the MVMM} \label{ss:mvmm_em}

We fit the MVMM described in Section \ref{s:mvmm} by minimizing the negative observed data log likelihood 
\begin{equation} \label{prob:max_obs_log_lik_mvmm}
\begin{aligned}
& \underset{\Theta, \pi}{\text{minimize}} & & - \ell( \{x_i \}_{i=1}^n| \Theta, \pi) 
\end{aligned}
\end{equation}
using an EM algorithm. 
The E-step constructs a surrogate function for the original objective function at the current guess and the M-step minimizes this surrogate function \citep{lange2000optimization}.

In detail, at the $s$th step, given current parameter estimates $(\Theta^s, \pi^s)$ the E-step constructs
\begin{equation} \label{eq:em_surrogate}
\begin{aligned}
Q^s\left(\Theta, \pi \right) & := E\left [ \sum_{i=1}^n \log f \left( x_i, y_i |  \Theta^s,  \pi^s \right)  \right] \\
& = \sum_{i=1}^n   \sumoverclusters \gamma(k\vs{1}, \dots, k\vs{V} | x_i ) \log \left( \pi_{k\vs{1}, \dots, k\vs{V}} \prod_{v=1}^V  \phi\vs{v}(x\vs{v}_i | \Theta^{(v)}_{k\vs{v}})  \right), 
\end{aligned}
\end{equation}
where $f$ is the complete data pdf \eqref{eq:mvmm_joint_pdf} and the responsibilities are
\begin{equation}\label{eq:em_resp}
\begin{aligned}
\gamma^s(k\vs{1}, \dots, k\vs{V} | x_i )  & := E\left[ P(y= (k\vs{1}, \dots, k\vs{V}) | x_i ) \right] =  \frac{\pi^s_{k\vs{1}, \dots, k\vs{V}}  \prod_{v=1}^V  \phi\vs{v}(x_i\vs{v} | \Theta^{(v), s}_{k\vs{v}})}{ f( x_i | \Theta^s, \pi^s ) } \text{ for each } i \in [n].
\end{aligned}
\end{equation}
The parameters are then updated in the M-step by solving 
$$\Theta^{s+1}, \pi^{s+1} =  \underset{\Theta, \pi}{\text{argmax }} Q^s\left(\Theta, \pi \right) $$
From \eqref{eq:em_surrogate} we see that this optimization problem splits into $V + 1$ separate problems; one for $\pi$ and one for each set of view cluster parameters $\Theta\vs{v}$, $v=1, \dots, V$.
The $\pi$ update has an analytical solution given by $ \pi^{s + 1} = a$ where $a\in R^{K\vs{1} \times \dots \times K\vs{V}}$ with
\begin{equation} \label{eq:a_resp}
a_{k\vs{1}, \dots, k\vs{V}} := \frac{1}{n} \sum_{i=1}^n \gamma^s(k\vs{1}, \dots, k\vs{V} | x_i).
\end{equation}
The cluster parameters for the $v$th view are updated by solving the following weighted maximum likelihood problem,
\begin{equation} \label{prob:mvmm_m_ste_view_clust_params}
\begin{aligned}
& \underset{\{\Theta^{(v)}_{k}\}_{k=1}^{K\vs{v}}}{\text{minimize}} & &  - \sum_{i=1}^n   \sum_{k=1}^{K\vs{v}} \gamma(k\vs{1}, \dots, k\vs{V} | x_i ) \log \left( \phi\vs{v}(x\vs{v}_i | \Theta^{(v)}_{k})  \right),
\end{aligned}
\end{equation}
Note this problem is in exactly the same form as the M-step for a standard, single-view mixture model making it straightforward to use pre-existing EM implementations.

\begin{algorithm}[H] \label{algo:mvmm_em}
\DontPrintSemicolon
  
\KwInput{$K\vs{1}, \dots, K\vs{V}$}
\KwData{$ \{x_i \}_{i=1}^n$}
\KwOutput{$\Theta, \pi$}
Initialize $\Theta^0 = \{\Theta^{(v), 0}\}_{v=1}^V, \pi^0$.
  
\While{Stopping criteria not satisfied}
{

$Q^s(\cdot), a^s \leftarrow$ E-step$\left( \{x_i \}_{i=1}^n, \Theta^s, \pi^s \right)$ \tcp*{From \eqref{eq:em_surrogate} and  \eqref{eq:em_resp}}

\For{for v=1, \dots, V}{ 
	$\Theta^{(v), s+1} \leftarrow \text{argmin}_{\Theta\vs{v}} Q^s(\Theta\vs{v})$ \tcp*{Solves a problem in the form of \eqref{prob:mvmm_m_ste_view_clust_params}}
}

$\pi^{s+1} \leftarrow  a^s$  
  
$ s \leftarrow s + 1$
  
}
\caption{EM algorithm for the MVMM}
\end{algorithm}

The view specific cluster parameters, $\Theta$, can be initialized using standard mixture model initialization strategies. 
We initialize the $\pi$ matrix so that the entries all have the same value.
We terminate the algorithm when the objective function has stopped decreasing.

\subsection{EM algorithm for the log penalized MVMM} \label{ss:log_pen_em}

This section discusses an EM algorithm for the log-penalized problem \eqref{eq:log_pen_max_log_lik}.
This EM algorithm is similar to the one described in Section \ref{ss:mvmm_em}, but the M-step solves a different problem.
At each step we majorize the log-likelihood with $Q^s(\Theta, \pi)$ from \eqref{eq:em_surrogate}. 
The updates for $\Theta$ are the same as in Section \ref{ss:mvmm_em}.
The update of $\pi$ leads to the following problem
\begin{equation} \label{eq:mvmm_m_step}
\begin{aligned}
& \underset{\pi \in \mathbb{R}^{K\vs{1} \times \dots \times K\vs{V}}}{\text{minimize}} & & - \sum_{k\vs{1}=1}^{K\vs{1}} \dots \sum_{k\vs{V}=1}^{K\vs{V}} a_{k\vs{1}, \dots, k\vs{V}}  \log( \pi_{k\vs{1}, \dots, k\vs{V}}) + \lambda  \log(\delta + \pi_{k\vs{1}, \dots, k\vs{V}}) \\
& \text{subject to } & &  \pi \ge 0 \text{ and }\sum_{k\vs{1}=1}^{K\vs{1}} \dots \sum_{k\vs{V}=1}^{K\vs{V}}  \pi_{k\vs{1}, \dots, k\vs{V}}= 1,
\end{aligned}
\end{equation}
where $a$ is given by \eqref{eq:a_resp}.
Based on Theorem \ref{thm:approx_soft} we approximate the solution to this problem with the normalized soft-thresholding operation
\begin{equation}\label{eq:mvmm_soft_thresh}
\pi_{k\vs{1}, \dots, k\vs{V}} = \frac{(a_{k\vs{1}, \dots, k\vs{V}} - \lambda)_+}{  \sum_{j\vs{1}=1}^{K\vs{1}} \dots  \sum_{j\vs{V}=1}^{K\vs{V}} (a_{j\vs{1}, \dots, j\vs{V}} - \lambda)_+}.
\end{equation}

\begin{algorithm}[H] \label{algo:mvmm_em_log_pen}
\DontPrintSemicolon
  
\KwInput{$K\vs{1}, \dots, K\vs{V}$, $0 < \lambda  < \frac{1}{\prod_{v=1}^V K\vs{v}}$ }
\KwData{$ \{x_i \}_{i=1}^n$}
\KwOutput{$\Theta, \pi$}
Initialize $\Theta^0 = \{\Theta^{(v), 0}\}_{v=1}^V, \pi^0$.
  
\While{Stopping criteria not satisfied}
{

$Q^s(\cdot), a^s \leftarrow$ E-step$\left( \{x_i \}_{i=1}^n, \Theta^s, \pi^s \right)$ \tcp*{From \eqref{eq:em_surrogate} and  \eqref{eq:em_resp}}

\For{for v=1, \dots, V}{ 
	$\Theta^{(v), s+1} \leftarrow \text{argmin}_{\Theta\vs{v}} Q^s(\Theta\vs{v})$ \tcp*{Solves a problem in the form of \eqref{prob:mvmm_m_ste_view_clust_params}}
}

$\pi^{s+1} \leftarrow$ normalized soft-thresholding applied to $a^s$ as in \eqref{eq:mvmm_soft_thresh} \tcp*{Approximates Problem \eqref{eq:mvmm_m_step}}

$ s \leftarrow s + 1$
  
}
\caption{EM algorithm for the log-penalized MVMM, Problem \eqref{eq:log_pen_max_log_lik} }
\end{algorithm}

Algorithm \ref{algo:mvmm_block_diag_extremal_em} is initialized by running a few EM steps for the unconstrained MVMM using the algorithm discussed in Section \ref{ss:mvmm_em}.
We terminate the algorithm when the objective function of \eqref{eq:log_pen_max_log_lik} has stopped decreasing. 
We specify a small value of $\delta$ (e.g. $10^{-6}$) to monitor the convergence of the objective function, but this value of $\delta$ plays no role in the EM updates due to \eqref{eq:mvmm_soft_thresh}.

\subsection{EM algorithm for the block diagonally constrained MVMM} \label{ss:mvmm_bd_constr_algo}

Following Sections \ref{ss:block_diag_opt} and \ref{s:wnn} we replace \eqref{prob:mvmm_bd_constr} with the following related problem for a sufficiently large value of $\alpha$,
\begin{equation} \label{prob:mvmm_block_diag_extremal}
\begin{aligned}
&\underset{\Theta, \blckdiag, U}{\text{minimize}}   & &  -  \ell( \{x_i \}_{i=1}^n| \Theta, \epsilon \mathbf{1}_{K\vs{1}}\mathbf{1}_{K\vs{2}}^T + \blckdiag) +  \alpha \text{Tr} \left( U^T L_{\text{un}}(A_{\text{bp}}(\blckdiag)) U\right)  \\
& \text{subject to} & & \blckdiag \ge 0,  \langle \blckdiag, \mathbf{1}_{K\vs{1}} \mathbf{1}_{K\vs{2}}^T  \rangle = 1 - K\vs{1} K\vs{2} \epsilon  \\
&&& U^T \text{diag}(\text{deg}(A_{\text{bp}}(\blckdiag))) U = I_B.
\end{aligned}
\end{equation}

We can solve this problem by alternating between updating $U$ and $(\Theta, \pi)$.
The $U$ variable is updated with an eigen-decomposition as in Corollary \ref{cor:eigen_subproblem}.
To update $(\Theta, \blckdiag)$ at the $s$th step we majorize the log-likelihood with $Q^s(\Theta, \epsilon \mathbf{1}_{K\vs{1}}\mathbf{1}_{K\vs{2}}^T+ \blckdiag)$ from \eqref{eq:em_surrogate}. 
The update for $\Theta$ is the same as in Section \ref{ss:mvmm_em}.
The M-step for $\blckdiag$ solves the following convex problem
\begin{equation}\label{prob:mvmm_block_diag_extremal_mstep}
\begin{aligned}
&\underset{\blckdiag}{\text{minimize}}   & & -\sum_{k\vs{1}=1}^{K\vs{1}}   \sum_{k\vs{2}=1}^{K\vs{2}} a_{k\vs{1} k\vs{2}}  \log (\epsilon + \blckdiag_{k\vs{1} k\vs{2}} ) + \alpha  \langle  \blckdiag, M(V, \mathbf{1}_B) \rangle\\
& \text{subject to} & & \blckdiag \ge 0, \langle \blckdiag, \mathbf{1}_{K\vs{1}} \mathbf{1}_{K\vs{2}}^T \rangle =  1 - K\vs{1} K\vs{2} \epsilon\\ 
& & &  \begin{bmatrix} c_{\text{diag}}(V)^T  \\ c_{\text{utri}}(V)^T \end{bmatrix} \text{diag}(\text{deg} ( A_{\text{bp}}( \blckdiag))) =   \begin{bmatrix} \mathbf{1}_{B} \\ \mathbf{0}_{ {B \choose 2}}  \end{bmatrix},
\end{aligned}
\end{equation}
where $a$ is from \eqref{eq:a_resp} and $M(U, \mathbf{1}_B), c_{\text{diag}}(U), c_{\text{utri}}(U)$ are from \eqref{eq:M}, \eqref{eq:c_diag}, \eqref{eq:c_utri}. 
Let \textsc{Update-D}(U) be an algorithm that solves the convex Problem \eqref{prob:mvmm_block_diag_extremal_mstep}. 

\begin{algorithm}[H] \label{algo:mvmm_block_diag_extremal_em}
\DontPrintSemicolon
  
 \KwInput{ $K\vs{1}, K\vs{2}$, $1 \le B \le \min(K\vs{1}, K\vs{2})$,  $0 < \epsilon <  \frac{1}{K\vs{1} K\vs{2}}$, $\alpha >0$}

\KwData{$ \{x_i \}_{i=1}^n$}
\KwOutput{$\Theta, \blckdiag$}
  
Initialize $\Theta^0 = \{\Theta^{(v), 0}\}_{v=1}^2, \blckdiag^0$.

Initialize $\alpha$ \tcp*{E.g. from \eqref{eq:alpha_heruistic} below} 

\While{Block diagonal stopping criteria is not satisfied}
{

\While{Optimization convergence stopping criteria not satisfied}
{

$U^{s+1} \leftarrow $ smallest $B$ generalized eigenvectors of \tcp*{Computed as in Corollary \ref{cor:eigen_subproblem}} 
$$\left( L_{\text{un}} (A_{\text{bp}}(\blckdiag^{s})), \text{diag}(\text{deg}(A_{\text{bp}}(\blckdiag^s))) \right)$$

$Q^s(\cdot) \leftarrow$ E-step$\left( \{x_i \}_{i=1}^n, \Theta^s, \epsilon \mathbf{1}_{K\vs{1}}\mathbf{1}_{K\vs{2}}^T +  \blckdiag^s \right)$  \tcp*{From \eqref{eq:em_surrogate} and  \eqref{eq:em_resp}}

\For{for v=1, 2}{ 
	$\Theta^{(v), s+1} \leftarrow \text{argmin}_{\Theta\vs{v}} Q^s(\Theta\vs{v})$ \tcp*{Solves a problem in the form of \eqref{prob:mvmm_m_ste_view_clust_params}}
}

$\blckdiag^{s+1} \leftarrow    \textsc{Update-D}(U^{s+1})$

$ s \leftarrow s + 1$
  
}

Increase $\alpha$ \tcp*{E.g. $\alpha \leftarrow 2 * \alpha$}
}
\caption{EM algorithm for the block diagonally constrained MVMM, Problem \eqref{prob:mvmm_block_diag_extremal}}
\end{algorithm}

Algorithm \ref{algo:mvmm_block_diag_extremal_em} is initialized by running a few EM steps for the unconstrained MVMM using the algorithm discussed in Section \ref{ss:mvmm_em}.
Each step of the inner loop of Algorithm \ref{algo:mvmm_block_diag_extremal_em} is guaranteed to decrease the objective function of \eqref{prob:mvmm_block_diag_extremal} therefore, we stop the inner loop when the objective function has stopped decreasing.
The convergence results discussed in Section \ref{s:wnn} apply to the inner loop of Algorithm \ref{algo:mvmm_block_diag_extremal_em}.

For a given value of $\alpha$, the solution output by  Algorithm \ref{algo:mvmm_block_diag_extremal_em} may have too few 0 eigenvalues; in this case we increase $\alpha$  (e.g. multiplying it by 2) and re-run the inner loop. 
The following proposition motivates a heuristic choice for the initial value of $\alpha$ as well as an initializer for \textsc{Update-D}.
\begin{proposition} \label{prop:log_lasso_heuristic}
Let $a, b, \epsilon > 0$.
The unique global minimizer,
\begin{equation}\label{prob:meowmoewmoew}
\begin{aligned}
x^* = \; &\underset{x \in \mathbb{R}}{\text{argmin}}   & & - a \log(x + \epsilon) + b x\\
& \text{subject to} & & x \ge 0
\end{aligned}
\end{equation}
is given by
\begin{equation*}
x^* = 
\begin{cases}
\frac{a}{b} - \epsilon & \text{ if } \frac{a}{b} - \epsilon > 0 \\
0  & \text{otherwise.}
\end{cases}
\end{equation*}
\end{proposition}

Let $(\Theta^0, D^0)$ be the initial guess in Algorithm \ref{algo:mvmm_block_diag_extremal_em}.
By ignoring constraints the solution to \eqref{prob:mvmm_block_diag_extremal_mstep} can be approximated by
\begin{equation*} \label{eq:D_approx}
D_{k\vs{1}, k\vs{2}}^* \approx  \left( \frac{a_{k\vs{1}, k\vs{2}}}{\alpha M(U, \mathbf{1}_B)_{k\vs{1}, k\vs{2}}} - \epsilon \right)_+,
\end{equation*}
where $a$ and  $U$ are obtained from $(\Theta^0, D^0)$ as above.
This suggests the following guess\footnote{This median value gives a rough estimate for the scale of $\alpha$ at which terms are set to 0.} for $\alpha$
\begin{equation}\label{eq:alpha_heruistic}
\alpha = c \cdot \text{median} \left(  \left \{ \frac{a_{k\vs{1}, k\vs{2}}} {\epsilon M(U, \mathbf{1}_B)_{k\vs{1}, k\vs{2}} }  \right \}_{k\vs{1} \in [K\vs{1}], k\vs{2} \in [K\vs{2}]}\right)
\end{equation}
for some small value of $c < 0$ e.g. $c = 0.01$.

Algorithm \ref{algo:mvmm_block_diag_extremal_em} can be sensitive to the initial choice of $\alpha$ and how fast $\alpha$ is increased.
Informally, if $\alpha$ is too large, the algorithm may converge quickly to a bad local minimizer.
If $\alpha$ is too small the algorithm will take longer to converge.


\section{Additional simulations} \label{s:add_sims}

This section expands on the simulations presented in Section \ref{s:simulations}.
The setup here is similar to the setup in Section \ref{s:simulations}.
Here we look at three different $\Pi$ matrices (Figure \ref{fig:sim_pi_mats}) and at two different signal to noise settings.
In the first setting the views have uneven signal to noise ratio where $\sigma_{\text{mean}}\vs{1} = 1$ and  $\sigma_{\text{mean}}\vs{2} = 0.5$ (i.e. the first view clusters are better separated than the second view clusters).
In the second setting the views have even signal to noise ratios where $\sigma_{\text{mean}}\vs{1} = \sigma_{\text{mean}}\vs{2} = 1$.
The figures below examine cluster level performance at the true parameter values (similar to Figure \ref{fig:beads_2_5__1__.5__n_samples_vs_test_overall_ars_at_truth}), block level performance at the true parameter values (similar to Figure \ref{fig:beads_2_5__1__.5__n_samples_vs_test_community_restr_ars_at_truth}) and the BIC estimated number of components (similar to Figure \ref{fig:beads_2_5__1__.5__n_samples_vs_n_comp_est_at_bic_selected}).
The details of these figures are explained in Section \ref{s:simulations}.

\begin{figure}[H]
 \centering
\begin{subfigure}[t]{0.2\textwidth}
\includegraphics[width=\linewidth,  height=\linewidth]{beads_2_5__1__.5/pi_true}
\caption{
$\Pi \in \mathbb{R}^{10 \times 10}$ with five $2 \times 2$ blocks.
All entries are equal.
}
\label{fig:beads_2_5__1__.5__pi_true__app}
\end{subfigure}
\hfill
\begin{subfigure}[t]{0.2\textwidth}
\includegraphics[width=\linewidth,  height=\linewidth]{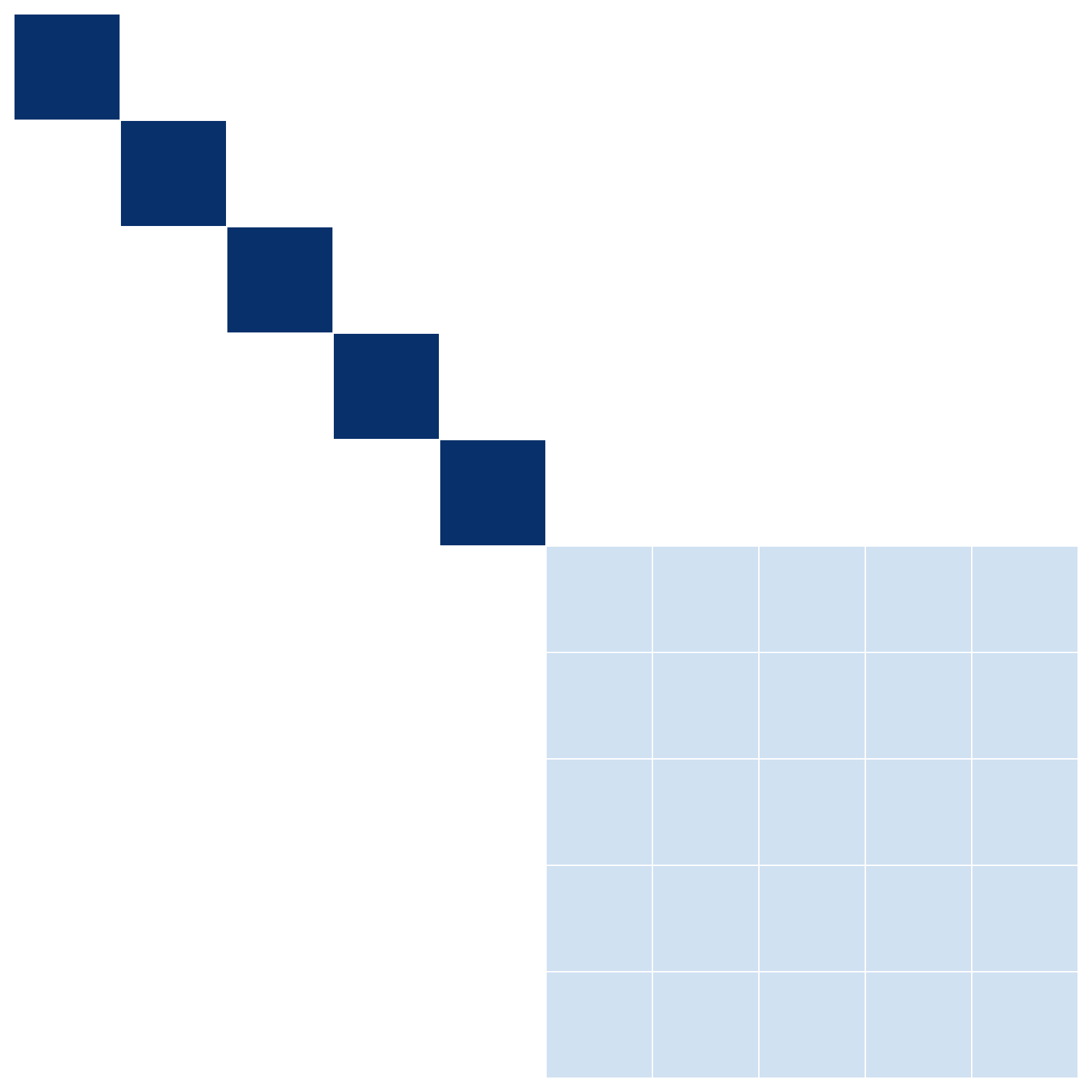}
\caption{
$\Pi \in \mathbb{R}^{10 \times 10}$ with five $1 \times 1$ blocks and one $5 \times 5$ block.
All blocks have the same total weight.
}
\label{fig:lollipop_5_5__1_.5_pi_true}
\end{subfigure}
\hfill
 \begin{subfigure}[t]{0.2\textwidth}
\includegraphics[width=\linewidth,  height=\linewidth]{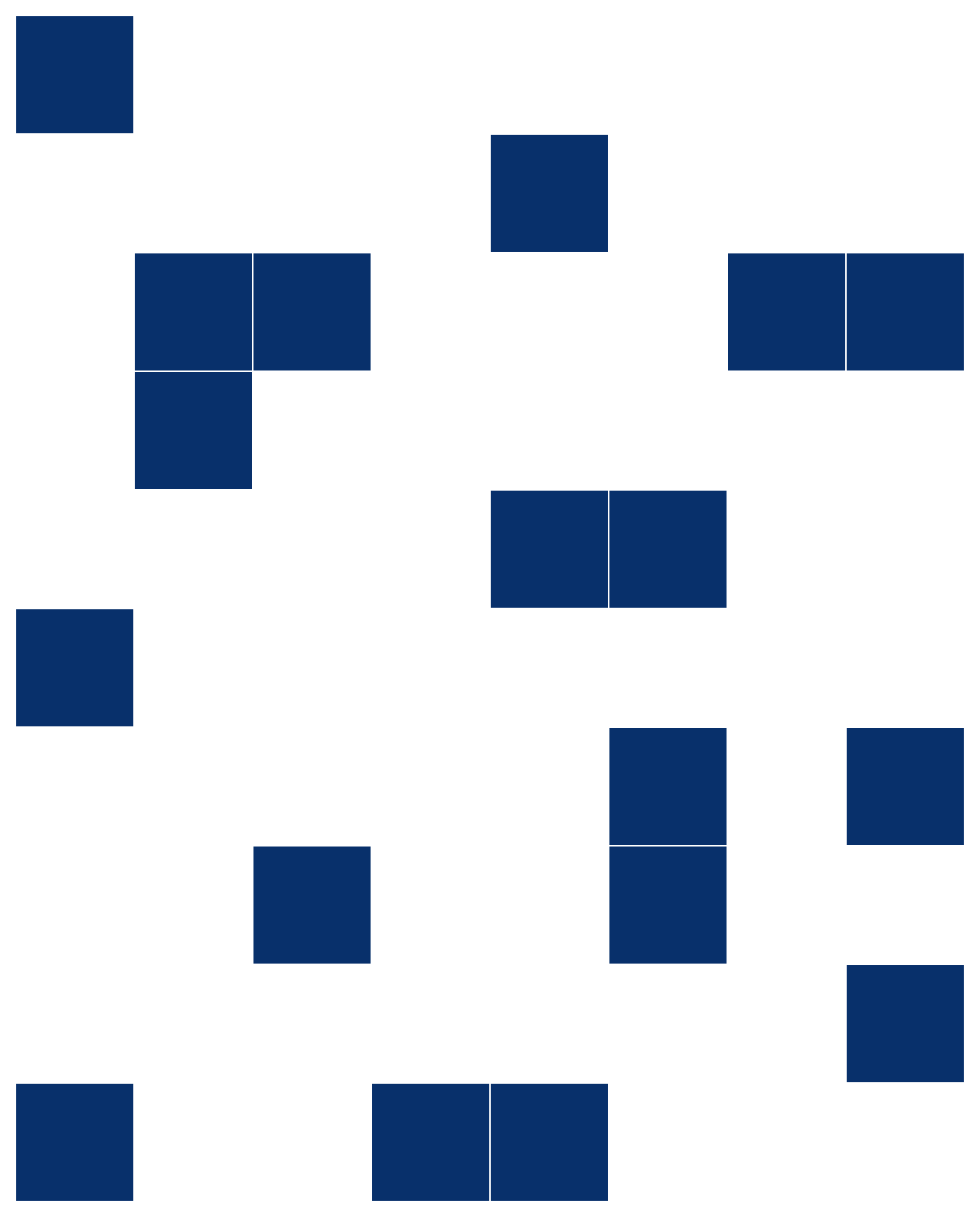}
\caption{
$\Pi \in \mathbb{R}^{10 \times 10}$  with $18$ randomly selected, non-zero entries with equal values.
}
\label{fig:sparse_pi_3__1_.5__pi_true}
\end{subfigure}
\caption{
Three different $\Pi$ matrices examined by simulations in this section.
}
\label{fig:sim_pi_mats}
\end{figure}

The overall takeaway is that the bd-MVMM and log-MVMM usually outperform the MVMM in the uneven setting.
In the even setting the bd-MVMM and log-MVMM sometimes still have an edge over the MVMM, but all three models are much closer together.
The log-MVMM sometimes struggles with the block level performance because small errors in the support of the estimated $\Pi$ can merge blocks together.
The BIC criteria often, but not always works well for the bd-MVMM.
This BIC criteria is usually biased towards selecting too few components for the log-MVMM.

Figure \ref{fig:beads_2_5__1__1_results} shows the results for the $\Pi$ matrix shown in \ref{fig:beads_2_5__1__.5__pi_true__app}.
The top row shows the uneven setting and the bottom row shows the even setting.
In the uneven setting the log-MVMM and bd-MVMM out-perform the MVMM in both cluster level and block level performance.
In the even setting the log-MVMM and bd-MVMM are only slightly better than the MVMM.
In the uneven setting the log-MVMM struggles with BIC based model selection, though it works well in the even setting.

Figure \ref{fig:lollipop_results} shows the results for the $\Pi$ matrix shown in \ref{fig:lollipop_5_5__1_.5_pi_true}.
In the uneven setting the bd-MVMM performs the best on both the cluster level and block level labels, however, the log-MVMM struggles with the block labels.
In this uneven setting BIC does not work well for either model and selects too few clusters in both cases.
BIC may struggle with the bd-MVMM because the individual clusters in the large $5\times 5$ block have smaller weights and therefore breaking this block up into several blocks does not harm the model fit as much as with the other $\Pi$.
In the even setting all three models perform similarly though the bd-MVMM has a slight edge at smaller sample sizes.
In this even setting BIC works well for bd-MVMM, but is still biased down for log-MVMM.

Figure \ref{fig:sparse_pi_results} shows the results for the sparse $\Pi$ matrix shown in \ref{fig:sparse_pi_3__1_.5__pi_true}.
Again in the uneven setting the log-MVMM out performs the MVMM, but in the even setting the two models are much closer together.
BIC is still biased towards too few clusters for this setting.

\begin{figure}[H]
 \centering
 \begin{subfigure}[t]{0.2\textwidth}
\includegraphics[width=\linewidth,  height=\linewidth]{beads_2_5__1__.5/n_samples_vs_test_overall_ars_at_truth}
\label{fig:beads_2_5__1__.5__n_samples_vs_test_overall_ars_at_truth__app}
\end{subfigure}
\hfill
\begin{subfigure}[t]{0.2\textwidth}
\includegraphics[width=\linewidth,  height=\linewidth]{beads_2_5__1__.5/n_samples_vs_test_community_restr_ars_at_truth}
\label{fig:beads_2_5__1__.5__n_samples_vs_test_community_restr_ars_at_truth__app}
\end{subfigure}
\hfill
 \begin{subfigure}[t]{0.2\textwidth}
\includegraphics[width=\linewidth,  height=\linewidth]{beads_2_5__1__.5/n_samples_vs_n_comp_est_at_bic_selected}
\label{fig:beads_2_5__1__.5__n_samples_vs_n_comp_est_at_bic_selected__app}
\end{subfigure}
\newline
\begin{subfigure}[t]{0.2\textwidth}
\includegraphics[width=\linewidth,  height=\linewidth]{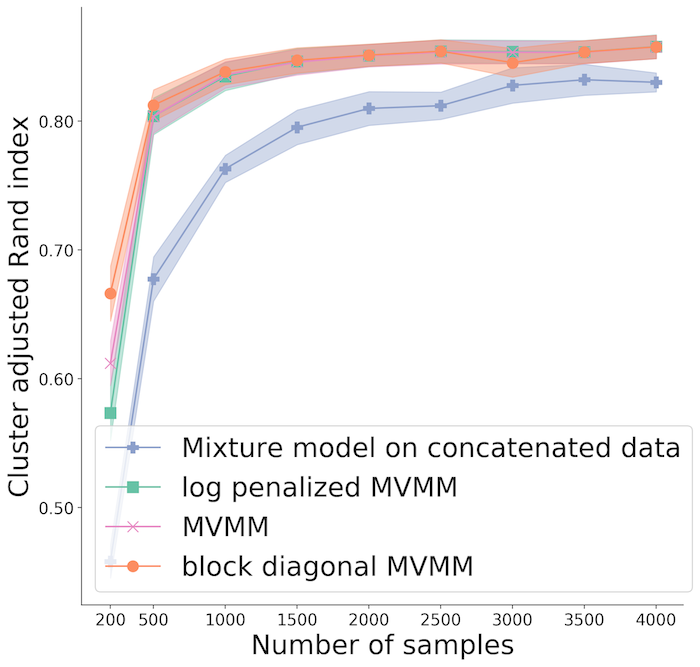}
\label{fig:beads_2_5__1__1__n_samples_vs_test_overall_ars_at_truth__app}
\end{subfigure}
\hfill
\begin{subfigure}[t]{0.2\textwidth}
\includegraphics[width=\linewidth,  height=\linewidth]{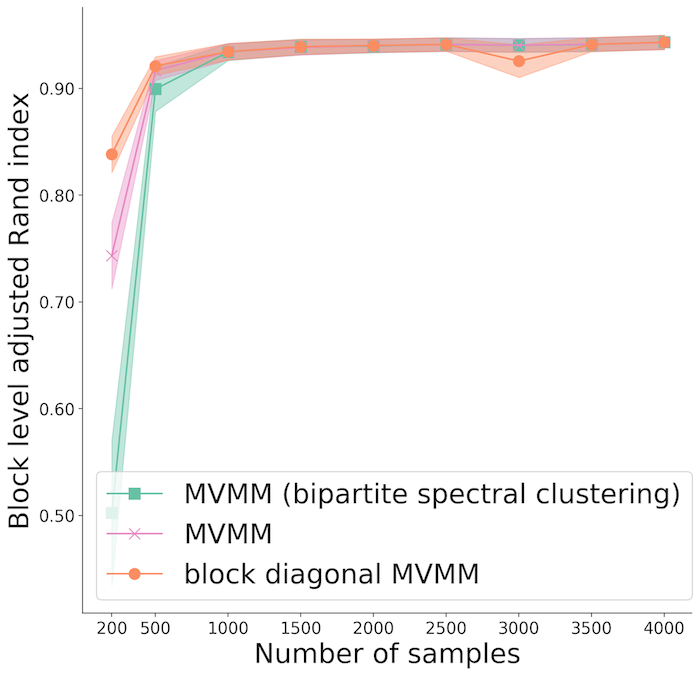}
\label{fig:beads_2_5__1__1__n_samples_vs_test_community_restr_ars_at_truth__app}
\end{subfigure}
\hfill
 \begin{subfigure}[t]{0.2\textwidth}
\includegraphics[width=\linewidth,  height=\linewidth]{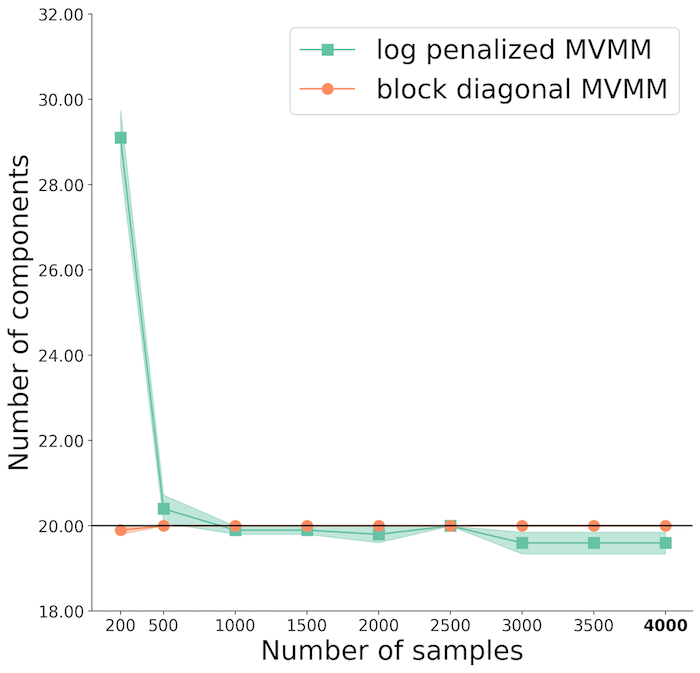}
\label{fig:beads_2_5__1__1__n_samples_vs_n_comp_est_at_bic_selected__app}
\end{subfigure}
\caption{
Results for the 5 block $\Pi$ matrix shown in Figure \ref{fig:beads_2_5__1__.5__pi_true__app}.
In the top row the view signal to noise ratios are uneven with  $\sigma_{\text{mean}}\vs{1} = 1$ and  $\sigma_{\text{mean}}\vs{2} = 0.5$.
In the bottom row the view signal to noise ratios are even with  $\sigma_{\text{mean}}\vs{1} =  \sigma_{\text{mean}}\vs{2} = 1$.
The first two columns examine the cluster label and block label performance at the true hyper-parameter values.
The third column examines the BIC estimated number of components.
}
\label{fig:beads_2_5__1__1_results}
\end{figure}

\begin{figure}[H]
 \centering
 \begin{subfigure}[t]{0.2\textwidth}
\includegraphics[width=\linewidth,  height=\linewidth]{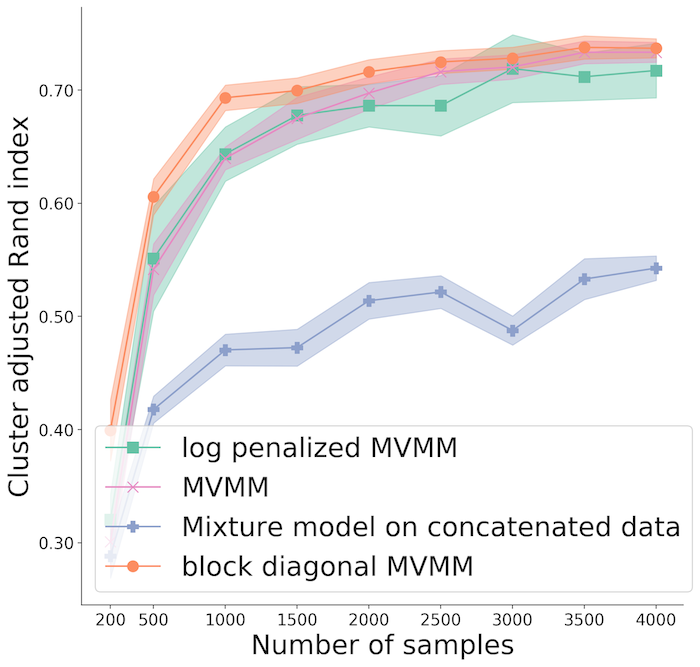}
\label{fig:lollipop_5_5__1_.5__n_samples_vs_test_overall_ars_at_truth}
\end{subfigure}
\hfill
\begin{subfigure}[t]{0.2\textwidth}
\includegraphics[width=\linewidth,  height=\linewidth]{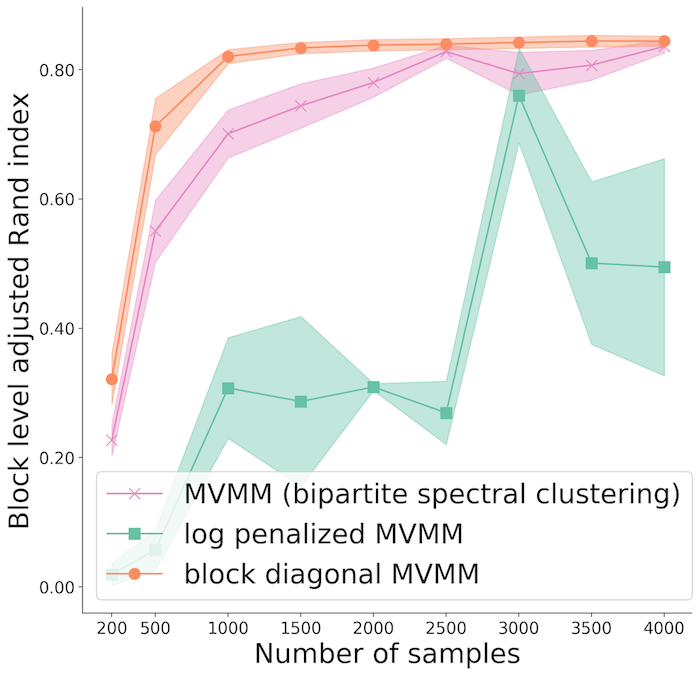}
\label{fig:lollipop_5_5__1_.5__n_samples_vs_test_community_restr_ars_at_truth}
\end{subfigure}
\hfill
 \begin{subfigure}[t]{0.2\textwidth}
\includegraphics[width=\linewidth,  height=\linewidth]{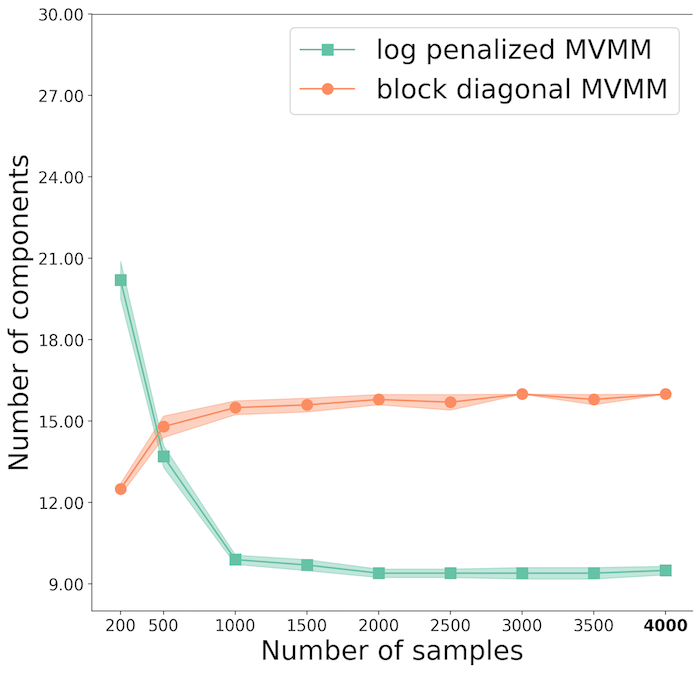}
\label{fig:lollipop_5_5__1_.5__n_samples_vs_n_comp_est_at_bic_selected}
\end{subfigure}
\newline
\begin{subfigure}[t]{0.2\textwidth}
\includegraphics[width=\linewidth,  height=\linewidth]{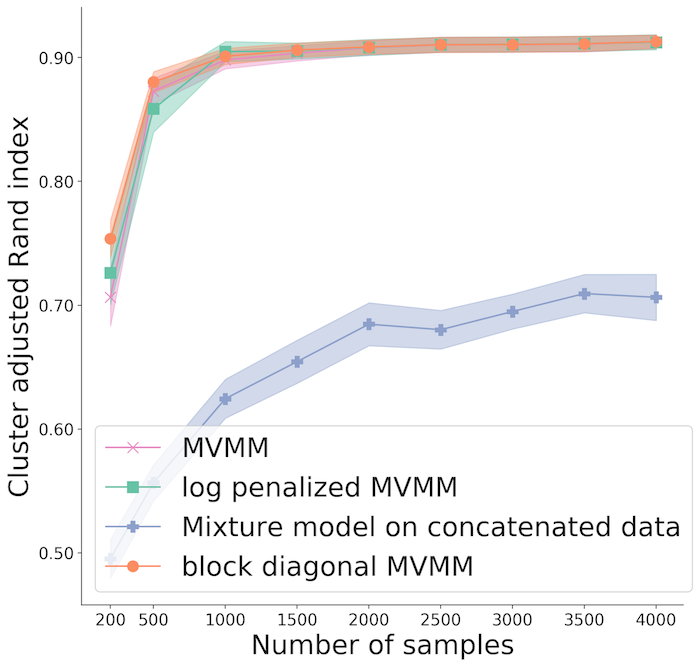}
\label{fig:lollipop_5_5__1_1__n_samples_vs_test_overall_ars_at_truth}
\end{subfigure}
\hfill
\begin{subfigure}[t]{0.2\textwidth}
\includegraphics[width=\linewidth,  height=\linewidth]{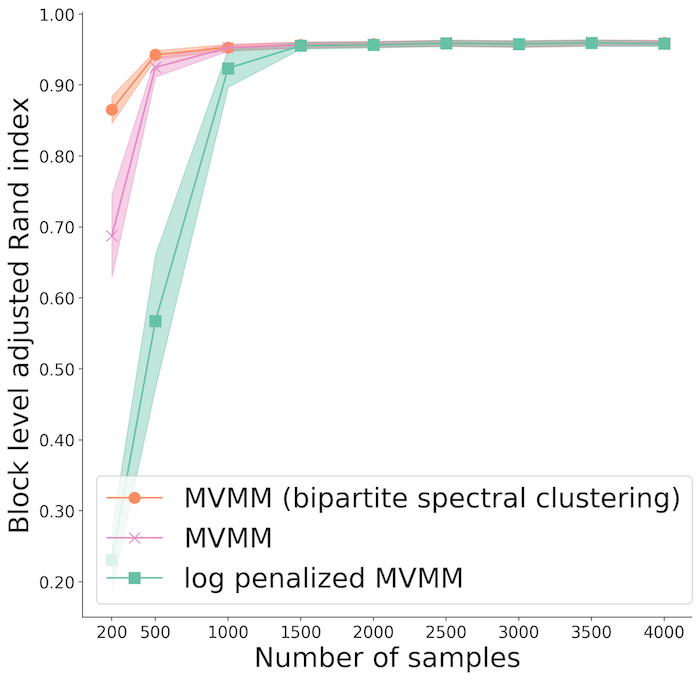}
\label{fig:lollipop_5_5__1_1__n_samples_vs_test_community_restr_ars_at_truth}
\end{subfigure}
\hfill
 \begin{subfigure}[t]{0.2\textwidth}
\includegraphics[width=\linewidth,  height=\linewidth]{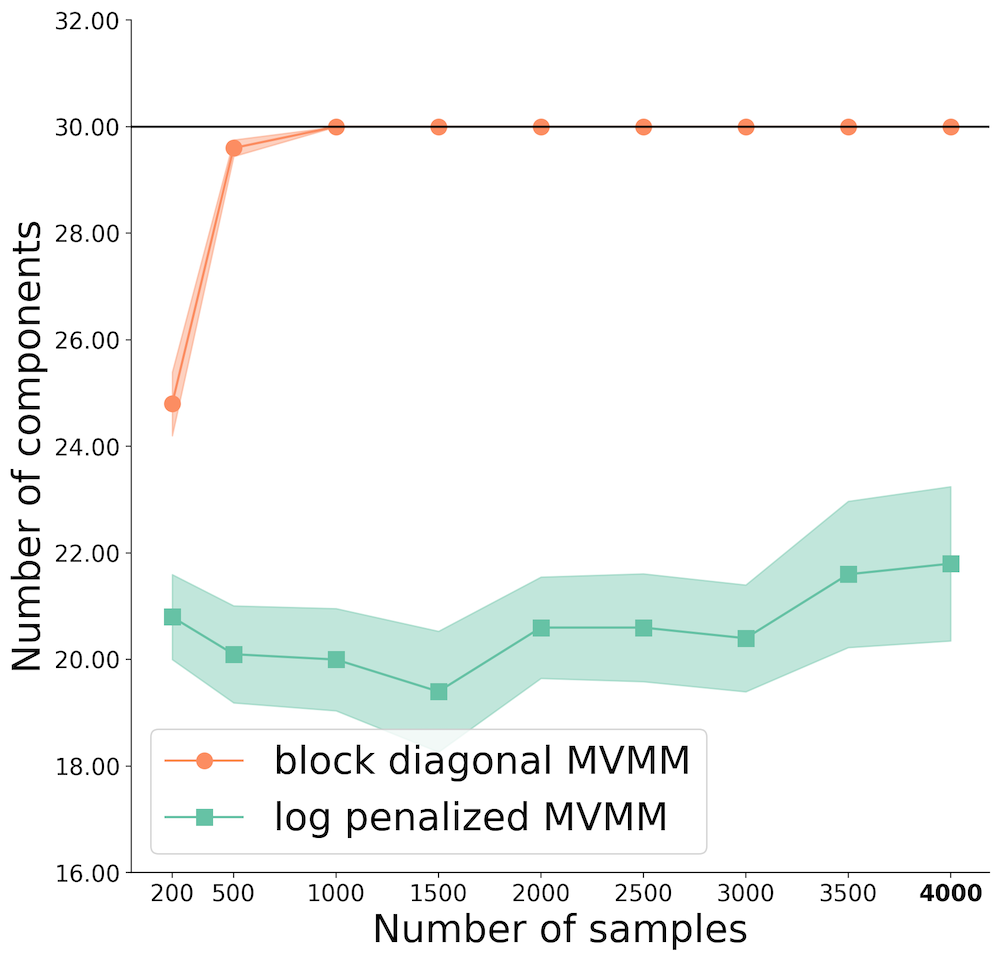}
\label{fig:lollipop_5_5__1_1__n_samples_vs_n_comp_est_at_bic_selected}
\end{subfigure}
\caption{
Results for the 6 block $\Pi$ matrix shown in Figure \ref{fig:lollipop_5_5__1_.5_pi_true}.
The top row shows the results for the uneven signal to noise ratio and the bottom row shows the results for the even signal to noise ratio.
}
\label{fig:lollipop_results}
\end{figure}

\begin{figure}[H]
 \centering
 \begin{subfigure}[t]{0.25\textwidth}
\includegraphics[width=\linewidth,  height=\linewidth]{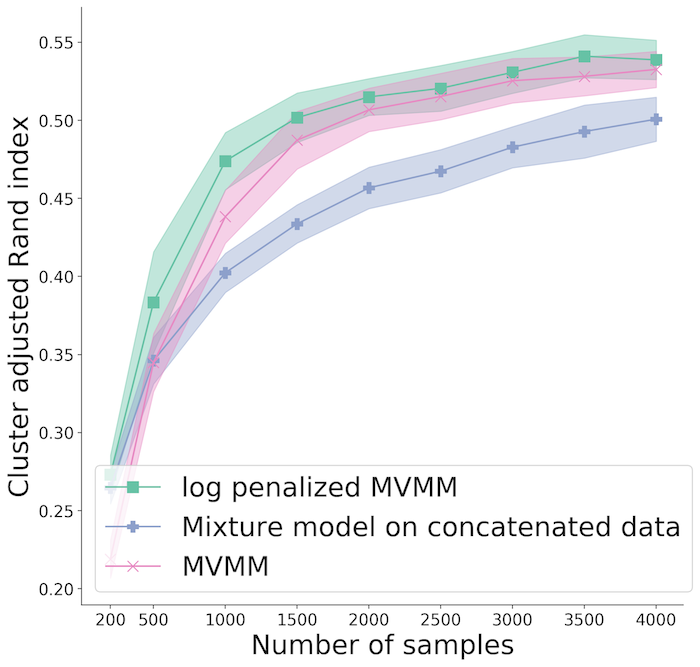}
\label{fig:sparse_pi_3__1_.5__n_samples_vs_test_overall_ars_at_truth}
\end{subfigure}
 \begin{subfigure}[t]{0.25\textwidth}
\includegraphics[width=\linewidth,  height=\linewidth]{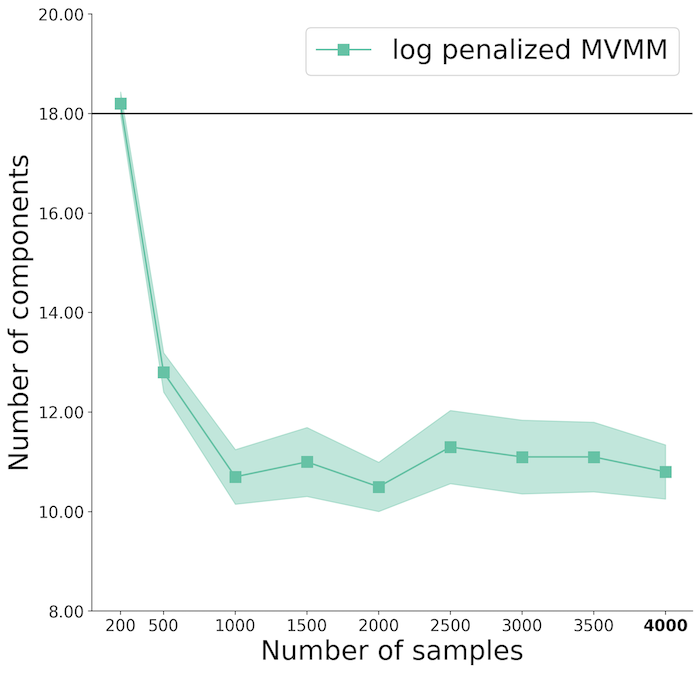}
\label{fig:sparse_pi_3__1_.5__n_samples_vs_n_comp_est_at_bic_selected}
\end{subfigure}
\vskip\baselineskip
\begin{subfigure}[t]{0.25\textwidth}
\includegraphics[width=\linewidth,  height=\linewidth]{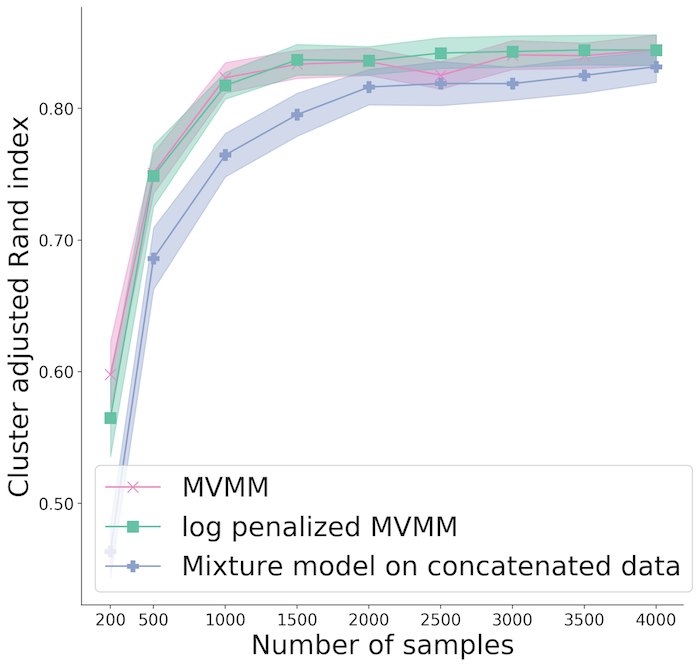}
\label{fig:sparse_pi_3__1_1__n_samples_vs_test_overall_ars_at_truth}
\end{subfigure}
 \begin{subfigure}[t]{0.25\textwidth}
\includegraphics[width=\linewidth,  height=\linewidth]{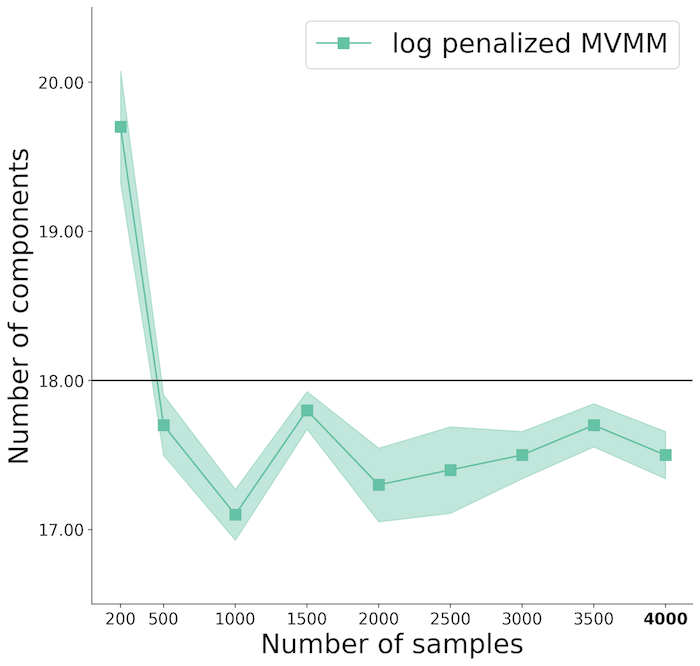}
\label{fig:sparse_pi_3__1_1__n_samples_vs_n_comp_est_at_bic_selected}
\end{subfigure}
\caption{
Results for the sparse $\Pi$ shown in Figure \ref{fig:sparse_pi_3__1_.5__pi_true}.
The top row shows the results for the uneven signal to noise ratio and the bottom row shows the results for the even signal to noise ratio.
Only the results for log-MVMM are shown.
}
\label{fig:sparse_pi_results}
\end{figure}


\section{Proofs} \label{s:proofs}

\subsection{Soft-thresholding with log penalty} \label{ss:proof_log_pen}

Let $f(\pi) = a\log(\pi) - \lambda \log(\delta + \pi)$ (see Figure \ref{fig:pen_fun}).
The intuition that Problem \eqref{eq:log_sparsity_problem_with_epsilon} sets terms to (approximately) 0 comes from the following.
\begin{itemize}
\item The solution to the unpenalized Problem \eqref{eq:log_sparsity_problem_with_epsilon} when $\lambda  = 0$ is $\pi^* = a$.

\item
If $a < \lambda$, $f(\pi)$ has a global maximum as $\pi^*= \frac{a \delta}{\lambda - a} \propto \delta$.
\end{itemize}
As $\delta \to 0$, the terms where $a_k < \lambda$ go to zero.
These terms become negligible in the probability constraint so the terms $\{k | a_k > \lambda \}$ solve the unpenalized problem (i.e. the problem if $\lambda=0$) with coefficients $a_k - \lambda$.

\begin{figure}[h]
\centering
\begin{subfigure}{0.3\textwidth}
\includegraphics[width=\linewidth]{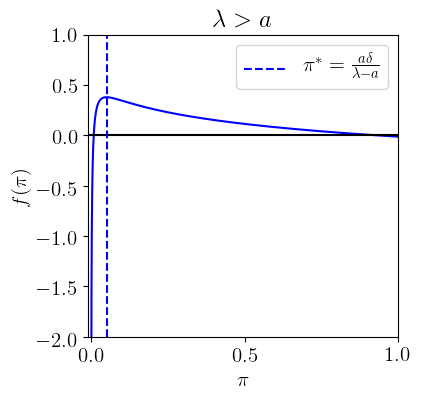}
\caption{
If $\lambda > a$ there is a global maximizer at $\frac{\delta a}{\lambda - a}$.
}
\label{fig:pen_fun_lambda_greater_than_a}
\end{subfigure}
\begin{subfigure}{0.3\textwidth}
\includegraphics[width=\linewidth]{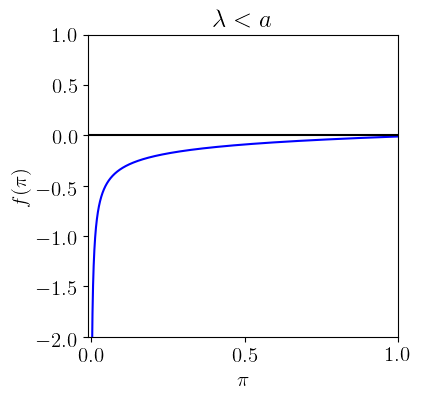}
\caption{
If $\lambda < a$ the function is strictly increasing and concave.
}
\label{fig:pen_fun_lambda_less_than_a}
\end{subfigure}

\caption{
Graph of $f(\pi)$ for two cases if $\lambda < a$ or $\lambda > a$. 
When $\lambda > a$, there is a global maximizer, which is proportional to $\delta$.
}
\label{fig:pen_fun}
\end{figure}

\begin{proof} of Theorem \ref{thm:approx_soft}

First we check there exists at least one global minimizer.
There exists a $\xi > 0$ such that $\pi \in [\xi, 1 - K \xi]^{K^C} \implies f(\pi) < f(\frac{1}{K} \mathbf{1}_K)$. 
Therefore, an optimal solution of \eqref{eq:log_sparsity_problem_with_epsilon} is the same as an optimal solution for the restricted problem where $\pi_k \in [\xi, 1 - K \xi]$ for $k=1, \dots, K$..
This restricted problem is a continuous function over a compact set and must attain a minimum thus  \eqref{eq:log_sparsity_problem_with_epsilon}  has at least one global minimizer.

Linear constraint qualification holds so the KKT conditions are first order necessary.
The Lagrangian of \eqref{eq:log_sparsity_problem_with_epsilon} is given by
$$
L(\pi, \eta) =   \sum_{k=1}^Ka_k  \log(\pi_k)  - \lambda \sum_{k=1}^K  \log(\delta + \pi_k) - \eta \pi^T \mathbf{1}_K
$$
for $\eta \in \mathbb{R}$.
We ignore the positivity constraint because the $-\log(z_k)$ terms ensure any stationary point is strictly positive.
The gradient of the Lagrangian is given by
\begin{equation*}
\frac{d \mathcal{L}}{d\pi_k} = \frac{a_k}{\pi_k} - \frac{\lambda}{\pi_k + \delta} - \eta \; \text{ for } k=1, \dots, K.
\end{equation*}

Suppose $(\pi, \eta)$ is a stationary point of the Lagrangian.
Setting $\frac{d \mathcal{L}}{d\pi_k} = 0 $ leaves us with
\begin{equation}\label{eq:pi_quad}
\eta \pi_k^2 + (\lambda + \eta \delta - a_k)\pi_k - a_k \delta = 0
\end{equation}

Because $\sum_k a_k=1$ and  $\lambda < \frac{1}{K}$,  without loss of generality $a_1 \ge \frac{1}{k} > \lambda$ so $a_1 > \lambda$.
If $\eta = 0$ then by \eqref{eq:pi_quad}, $\pi_k = \frac{a_k \delta}{\lambda - a_k}$ for each $k$. In this case, $\pi_1 < 0$ violates the positivity constraint so we conclude $\eta \neq 0$.
Thus
\begin{equation}\label{eq:pi_soln_set}
\pi_k \in  \left\{ \frac{(a_k - \lambda - \eta \delta) \pm \sqrt{(a_k - \lambda - \eta \delta)^2 + 4 a_k \eta \delta}}{2 \eta} \right\}
\end{equation}

Next we check that $\eta > 0$ at any stationary point of the Lagrangian.
Assume for the sake of contradiction that $\eta < 0$.
Recall $\lambda < a_1$ so $a_1 - \lambda - \eta \delta > 0$. Thus
$$
\frac{(a_k - \lambda - \eta \delta) + \sqrt{(a_k - \lambda - \eta \delta)^2 + 4 a_k \eta \delta}}{2 \eta} < 0
$$
which violates the constraint $\pi_1 \ge 0$.
Furthermore, $(a_k - \lambda - \eta \delta) - \sqrt{(a_k - \lambda - \eta \delta)^2 + 4 a_k \eta \delta} > 0$ so
$$
\frac{(a_k - \lambda - \eta \delta) - \sqrt{(a_k - \lambda - \eta \delta)^2 + 4 a_k \eta \delta}}{2 \eta} < 0
$$
which again violates the constraint $\pi_1 \ge 0$.
Therefore we conclude $\eta > 0$ for $\pi$ to be a stationary point.

It can be checked that if a $(-)$ is chosen in \eqref{eq:pi_soln_set} then $\pi_1 < 0$.
Thus at a stationary point
\begin{equation} \label{eq:pi_quad_eq}
\pi_k =  \frac{(a_k - \lambda - \eta \delta) + \sqrt{(a_k - \lambda - \eta \delta)^2 + 4 a_k \eta \delta}}{2 \eta}
\end{equation}

Next we show $\eta \delta \to 0$ when $\delta \to 0$.
Using the constraint $1 = \sum_{k=1}^K \pi_k$ we get
\begin{equation} \label{eq:dan}
2 \eta = \sum_{k=1}^K \left( a_k - \lambda - \eta \delta + \sqrt{(a_k - \lambda - \eta \delta)^2 + 4 a_k \eta \delta} \right)
\end{equation}

It can be checked that
$$
(a_k - \lambda - \eta \delta)^2 + 4 a_k \eta \delta \le (a_k + \lambda + \eta \delta)^2,
$$
therefore
\begin{align*}
2 \eta &  \le \sum_{k=1}^K |a_k - \lambda - \eta \delta|  + \sqrt{(a_k - \lambda - \eta \delta)^2 + 4 a_k \eta \delta}  \le \sum_{k=1}^K |a_k - \lambda - \eta \delta|  + |a_k + \lambda + \eta \delta| \\
& \le 2 K \eta \delta +   \sum_{k=1}^K |a_k - \lambda|  + a_k + \lambda.
\end{align*}
Thus
\begin{equation}\label{eq:eta_upper_bound}
\eta \le \frac{1}{2} \frac{\sum_{k=1}^K |a_k - \lambda|  + |a_k + \lambda | }{1 - K \delta} 
\end{equation}
which is upper bounded by a positive constant independent of $\delta$ because $\delta < \frac{1}{K}$.
We now conclude
\begin{equation} \label{eq:eta_eps_to_zero}
\lim_{\delta \to 0} \eta \delta = 0
\end{equation}

Finally, using \eqref{eq:dan} and \eqref{eq:eta_eps_to_zero} we see
\begin{align*}
\lim_{\delta \to 0} \eta & = \frac{1}{2} \lim_{\delta \to 0}   \sum_{k=1}^K  \left( a_k - \lambda - \eta \delta + \sqrt{(a_k - \lambda - \eta \delta)^2 + 4 a_k \eta \delta}  \right)\\
& =  \frac{1}{2}   \sum_{k=1}^K  a_k - \lambda + |a_k - \lambda|  = \sum_{k=1}^K (a_k - \lambda)_+, 
\end{align*}
and the result follows from \eqref{eq:pi_quad_eq}.

\end{proof}

\subsection{Extremal characterization of generalized eigenvalues } \label{ss:proof_eig}

\begin{proof} of Proposition \ref{prop:gevals_well_def_sing_B}

Let $B = V D V^T$ be the eigen-decomposition of $B$ where $V \in \mathbb{R}^{m \times m}$ and $D \in \mathbb{R}^{m \times m}$ is the diagonal matrix of non-zero eigenvalues of $B$.
Let $Z \in \mathbb{R}^{(n - m) \times n}$ be such that $Q = [V; Z]$ is an orthonormal matrix (i.e. $Z$ is a basis for the kernel).

Then $B^{-1/2} := QD^{-1/2} Q^T$ is the Moore-Penrose inverse of the square root of $B$ and 
$$
B^{-1/2} A B^{-1/2} = \begin{bmatrix} \widetilde{A}  & 0 \\ 0 & 0 \end{bmatrix},
$$
where $ \widetilde{A}  \in \mathbb{R}^{m \times m}$ since $\text{ker}(B) \subseteq \text{ker}(A)$.
Let $\lambda_1, \dots, \lambda_m$ be the eigenvalues of $\widetilde{A}$ with corresponding eigenvectors $\widetilde{v}_1, \dots \widetilde{v}_m \in \mathbb{R}^{m \times m}.$
Let $v_j \in \mathbb{R}^{n}$ be the concatenation of $\widetilde{v}_j$ along with $n - m$ zeros.
Then the $v_j$ are eigenvectors of $B^{-1/2} A B^{-1/2} $ with eigenvalues $\lambda_j$. 
Setting $w_j = B^{1/2} v_j$ we see $w_j$ are generalized eigenvectors of $(A, B)$ with generalized eigenvalues $\lambda_j$.

\end{proof}

\begin{proof} of Proposition \ref{prop:weighted_geval_extremal_representation}

First assume that $B$ is positive definite. 
Proposition 20.A.2.a from \citet{marshall1979inequalities} states
\begin{equation}\label{eq:prop_20a2a}
\begin{aligned}
& \sum_{j=1}^K \lambda_i(A)  \lambda_i(H) = &    & \underset{U \in \mathbb{R}^{n \times K}}{\text{maximum}}  & &  \text{tr}\left( U^T A U H \right)\\
& & & \text{subject to} & &  U^T U = I_K.
\end{aligned}
\end{equation}
It is straightforward to check that this maximum value is attained by $U = U_A(:, \; 1:K) U_H$ where $U_A(:, \; 1:K) \in \mathbb{R}^{n \times K}$ is an orthonormal matrix of eigenvectors corresponding to the largest $K$ eigenvalues of $A$ and $U_H \in \mathbb{R}^{K \times K}$ is an orthonormal matrix of eigenvectors of $H$.

If $B=I_K$ then \eqref{eq:weighted_sum_largest_eval_extremal} is a special case of \eqref{eq:prop_20a2a} with $H = \text{diag}(w)$.
Recall that $(\lambda, u)$ is a generalized eigenvector of $(A, B)$ if and only if $(\lambda, B^{1/2} u)$ is an eigenvector of $B^{-1/2}AB^{-1/2}$ and \eqref{eq:weighted_sum_largest_eval_extremal} follows.
Using this fact it is straightforward to extend the results for general, positive definite $B$.
Repeating this argument with $H = -\text{diag}(w)$ we obtain \eqref{eq:weighted_sum_smallest_eval_extremal}.

Next we relax the positive definite assumptions; assume $\text{ker}(B) \subseteq \text{ker}(A)$.
Let $Q \in \mathbb{R}^{n \times n}$ be an orthonormal matrix whose columns are $q_1, \dots,  q_n \in \mathbb{R}^n$.
Suppose $q_1, \dots, q_m \in \mathbb{R}^n$ is an orthonormal basis of $\text{ker}(B)^{\perp}$ and  $q_{m+1}, \dots, q_n \in \mathbb{R}^n$ is an orthonormal basis of $\text{ker}(B) \subseteq \text{ker}(A) $ where $m =  n - \text{dim}(\text{ker}(B)) \le K$  by assumption.
Then
\begin{equation} \label{eq:henry}
Q^T A Q = \begin{bmatrix} \widetilde{A} & 0 \\ 0 &0  \end{bmatrix} \text{ and } Q^T B Q = \begin{bmatrix} \widetilde{B} & 0 \\ 0 &0  \end{bmatrix}
\end{equation}
where $ \widetilde{A},  \widetilde{B} \in \mathbb{R}^{m \times m}$ and $\widetilde{B}$ is  positive definite.
Let $\widetilde{U} \in \mathbb{R}^{m \times K}$ be generalized eigenvectors corresponding to the largest $K$ generalized eigenvalues of $(\widetilde{A},  \widetilde{B})$ and let 
$$U = Q^T \begin{bmatrix} \widetilde{U}  \\ 0 \end{bmatrix} \in \mathbb{R}^{n \times K}$$
It is straightforward to see that the columns of $U$ are generalized eigenvectors for $(A, B)$.
We next check that $U$ is a global maximizer of \eqref{eq:weighted_sum_largest_eval_extremal}.
First note
$$ U^T B U= \begin{bmatrix} \widetilde{U}  \\ 0  \end{bmatrix}Q B Q^T \begin{bmatrix} \widetilde{U}  \\ 0  \end{bmatrix} =  \begin{bmatrix} \widetilde{U} \\ 0 \end{bmatrix}^T  \begin{bmatrix} \widetilde{B} & 0 \\ 0 &0  \end{bmatrix} \begin{bmatrix} \widetilde{U} \\ 0 \end{bmatrix} = \widetilde{U}^T \widetilde{B}\widetilde{U} = I_K$$
so $U$ is feasible and its objective value is given by
\begin{align*}
\text{Tr}\left(U^T A U \text{diag}(w)\right)  & =   \text{Tr}\left(\begin{bmatrix} \widetilde{U} \\ 0 \end{bmatrix}^T Q^T A Q \begin{bmatrix} \widetilde{U} \\ 0 \end{bmatrix} \text{diag}(w)\right) 
 = \text{Tr}\left(\widetilde{U}^T \widetilde{A} \widetilde{U} \text{diag}(w)\right)
\end{align*}

Suppose that $W \in \mathbb{R}^{n \times K}$ is a global maximizer of \eqref{eq:weighted_sum_largest_eval_extremal}.
Let $\widetilde{W} \in \mathbb{R}^{m \times K}, \overline{W}  \in \mathbb{R}^{n- m \times K}$ such that 
$$Q^T W =  \begin{bmatrix} \widetilde{W}  \\ \overline{W}  \end{bmatrix} \in \mathbb{R}^{n \times K}$$
Noting that $Q$ is orthonormal and using \eqref{eq:henry} we see
$$
I_K = W^T B W = W^T Q Q^T B Q^T Q W =  \widetilde{W}^T \widetilde{B} \widetilde{W},
$$
and 
\begin{align*}
\text{Tr}\left(W^T A W \text{diag}(w)\right) & = \text{Tr} \left(W^T Q Q^T A Q^T Q W \text{diag}(w)\right)   =  \text{Tr}\left(W^T Q  \begin{bmatrix} \widetilde{A} & 0 \\ 0 &0  \end{bmatrix} Q W \text{diag}(w) \right)  \\
& = \text{Tr}\left(\widetilde{W}^T  \widetilde{A}  \widetilde{W} \text{diag}(w)\right).
\end{align*}

Assume for the sake of contradiction that $ \text{Tr}\left(W^T A W \text{diag}(w)\right)  > \text{Tr}\left(U^T A U \text{diag}(w)\right)$.
Consider Problem \eqref{eq:weighted_sum_largest_eval_extremal} with $(\widetilde{A},  \widetilde{B})$. 
By the first part of the proof, $\widetilde{U}$ is a global maximizer (since $ \widetilde{B}$ is strictly positive definite).
From the above discussion we have that $\widetilde{W}$ is feasible for this problem with objective value $\text{Tr}\left(W^T A W \text{diag}(w)\right)$, however, this contradicts the fact that $\widetilde{U}$ is a global maximizer.

\end{proof}

Corollary \ref{cor:sym_evals_and_gevals} is proved below in Section \ref{prop:sym_lap_block_diag}.

\subsection{Spectrum of the symmetric Laplacian and block diagonal structure} \label{ss:proof_sym_lap_spect_bd}

We first give two propositions detailing the spectral properties of the symmetric, normalized Laplacian.
Recall the convention for isolated vertices discussed in Section \ref{ss:sym_lap} that ensures the diagonal of $L_{\text{sym}}(\cdot)$ is always equal to 1.

\begin{proposition} \label{prop:sym_lap_spectrum_ccs}
Let $A \in \mathbb{R}^{n \times n}_+$ be the adjacency matrix of a graph with undirected, positively weighted edges and no self loops (i.e. $A$ is symmetric and has 0s on the diagonal).

There is a one-to-one correspondence between the 0 eigenvalues of $L_{\text{sym}}(A)$ and the connected components of $G$ with at least two vertices.
Let $A_1, A_2, \dots, [n]$ correspond to the connected components with at least two vertices and let $v_i = \text{diag}(\text{deg}(A)) \mathbf{1}_{A_i}$ where $ \mathbf{1}_{A_i}$ is the vector with 1s in the entries corresponding to indices in $A_i$ and 0s elsewhere.
The eigenspace of 0 is spanned by $v_1, v_2, \dots$.

The number of eigenvalues of $L_{\text{sym}}(A)$ equal to 1 is at least the number of isolated vertices.
The basis vector with a 1 in the entry corresponding to an isolated vertex is an eigenvector with eigenvalue 1.

\end{proposition}

In general there is \textit{not} a one-to-one correspondence between isolated vertices and 1 eigenvalues e.g. consider
$$
L_{\text{sym}}(A_{\text{bp}}(\mathbf{1}_m \mathbf{1}_m^T))
$$
which has no isolated vertices but $2m$ eigenvalues equal to 1.

Note the difference between Proposition \ref{prop:sym_lap_spectrum_ccs} and Proposition 4 of \citet{von2007tutorial}.
Some papers choose the convention that the symmetric Laplacian has a 0 on the diagonal for isolated vertices.
When this alternative convention is selected, the symmetric, normalized Laplacian and the normalized Laplacian would treat isolated vertices the same. 

 \begin{proposition} \label{prop:sym_lap_tsym}
Let $X \in \mathbb{R}^{R \times C}_+$, then the eigenvalues of $L_{\text{sym}}(A_{\text{bp}}(X))$ are 
\begin{itemize}
\item located in $[0, 2]$,
\item symmetric around 1 meaning for every eigenvalue $\lambda  = 1 - \eta$, $\eta \ge 0$ there is a corresponding eigenvalue at $1 + \eta$,
\item given by $\{1 \pm \sigma_i(T_{\text{sym}}(X)) \}_{i=1}^{\min(R, C)}$ and $R + C - 2 \min(R, C)$ 1s.
\end{itemize}
The singular values of $T_{\text{sym}}(X)$ are located in  $[0, 1]$. 

Let $U \in \mathbb{R}^{R \times K}$, $V \in \mathbb{R}^{C \times K}$ be matrices of the largest $K$ left and right singular vectors of $T_{\text{sym}}(X))$ respectively.
Then $\begin{bmatrix} U \\ V\end{bmatrix}$ is the matrix of the smallest $K$ eigenvectors of $L_{\text{sym}}(A_{\text{bp}}(X))$ and  $\begin{bmatrix} U \\ -V\end{bmatrix}$ is the matrix of the matrix of the largest $K$ eigenvectors respectively.
\end{proposition}

\begin{proof} of Proposition \ref{prop:sym_lap_spectrum_ccs}

Consider the subgraph corresponding to the $m$ vertices of this graph which are contained in connected components with at least two vertices.
Without loss of generality assume that these are the first $m$ nodes of the graph and let $\widetilde{A} \in \mathbb{R}^{m \times m}$ be the corresponding adjacency matrix.
Then
\begin{equation*}
\text{L}_{\text{sym}}(A) = 
\begin{bmatrix} \text{L}_{\text{sym}}(\widetilde{A}) & 0_{m \times (n-m)} \\ 0_{(n-m) \times m} & I_{n - m} \end{bmatrix}
\end{equation*}
From this we see the unit vectors $e_i \in \mathbb{R}^m$ for $i=n+1, \dots, m$ are eigenvectors of $\text{L}_{\text{sym}}(A)$ with eigenvalue 1 and the claim about isolated vertices follows.

If $(\lambda, \widetilde{v})$ with  $\widetilde{v} \in \mathbb{R}^{m}$ is an eigenvalue/eigenvector pair for $\text{L}_{\text{sym}}(\widetilde{A})$ then $(\lambda, v)$ is an eigenvalue/eigenvector pair for $\text{L}_{\text{sym}}(A)$ with $v := (\widetilde{v}, 0, \dots, 0) \in \mathbb{R}^n$.
Proposition 4 of \citet{von2007tutorial} holds for $\text{L}_{\text{sym}}(\widetilde{A})$ which has no isolated vertices.
Therefore, an orthonormal set of $m$ eigenvectors of $\text{L}_{\text{sym}}(\widetilde{A})$ give $m$ orthonormal eigenvectors of $\text{L}_{\text{sym}}(A)$ with the same eigenvalues.
From this we see the claim about 0 eigenvalues of $\text{L}_{\text{sym}}(A)$ and the corresponding eigenspace follows.

\end{proof}

\begin{proof}  of Proposition \ref{prop:sym_lap_tsym}

By checking $\text{L}_{\text{sym}}(A_{\text{bp}}(X))$ is diagonally dominant we see it is positive-semi definite.
The diagonal elements of $\text{L}_{\text{sym}}(A_{\text{bp}}(X))$ are equal to 1. 
Consider the first row; if $\text{deg}(A_{\text{bp}}(X))_1=0$ then the first row of $\text{L}_{\text{sym}}(A)$ is equal to the first standard basis vector. 
If $\text{deg}(A_{\text{bp}}(X))_1 > 0$ then $\text{L}_{\text{sym}}(A_{\text{bp}}(X))_{1j}  = \frac{A_{\text{bp}}(X)_{1j}}{\sum_{i=1}^n A_{\text{bp}}(X)_{ij} } \le 1$.

We see that
\begin{equation*}
\text{L}_{\text{sym}}(A_{\text{bp}}(X)) = I - \begin{bmatrix} 0 & T_{\text{sym}}(X) \\ T_{\text{sym}}(X)^T & 0  \end{bmatrix}.
\end{equation*}
Note that the spectrum of the second matrix on the right hand side is symmetric around 0.
It is straightforward to check the remaining claims of the proposition.

\end{proof}

\begin{proof}  of Proposition \ref{prop:sym_lap_block_diag}

By inspecting the adjacency matrix $A_{\text{bp}}(X)$ it is clear there is a one-to-one correspondence between the zero rows/columns of $X$ and the isolated vertices in $G(X)$.
Without loss of generality, assume $G(X)$ has no isolated vertices and suppose $G(X)$ has $B$ connected components with at least two vertices.

Let $\sigma_r$ be a permutation of the rows of $X$ such that the first rows of $X$ belong to the first connected component, the next rows of $X$ belong to the second connected component, etc.
Let $\sigma_c$ be the analogous permutation of the columns of $X$.
Let $\widetilde{X}$ be the result of applying these two permutations to $X$, then $\widetilde{X}$ is block diagonal with $B$ blocks.
We thus conclude that the number of connected components of $G$ is a lower bound for $NB(X)$.

Now suppose there exists a permutation, $\sigma_r$ of the rows of $X$ and a permutation $\sigma_c$ of the columns of $X$ such that the resulting matrix has $C \ge B + 1$ blocks.
Let $\widetilde{X}$ be the result of applying these two permutations to $X$, then $A_{\text{bp}}(\widetilde{X})$ is the adjacency matrix of a bipartite graph with $C$ connected components.
But this is a contradiction, because shuffling the node labels of $G(X)$ induces a graph isomorphism so the number of connected components must be conserved.
Thus claims 1 and 2 are equivalent.

Claims 2 and 3 are equivalent by Proposition \ref{prop:sym_lap_spectrum_ccs}.
Claims 1 and 4 are equivalent because $G(X)$ has $B + Z_{\text{row}} + Z_{\text{col}}$ connected components and Proposition 2 of \citet{von2007tutorial}.

\end{proof}

\begin{proof} of Corollary \ref{cor:sym_evals_and_gevals}.
Note that $\text{dim}(\text{ker}(\text{diag}( \text{deg}(A_{\text{bp}} (X))))) = R + C - (\widetilde{R} + \widetilde{C})$.
Eigenvectors of $L_{\text{sym}}(A_{\text{bp}} (X))$ that live in the kernel of $\text{diag}( \text{deg}(A_{\text{bp}} (X)))$ correspond to isolated nodes thus give an eigenvalue of $1$.
By Proposition \ref{prop:sym_lap_spectrum_ccs}, $\lambda_{(k)}(L_{\text{sym}}(A_{\text{bp}} (X))) \le 1$ for $k \le \min(R, C)$.
Therefore, for $k \le \widetilde{R} + \widetilde{C}$, none of the $\lambda_{(k)}(L_{\text{sym}}(A_{\text{bp}} (X)))$ correspond to eigenvectors in the kernel of $\text{diag}( \text{deg}(A_{\text{bp}} (X)))$.
The result follows from Proposition \ref{prop:gevals_well_def_sing_B}.
 \end{proof}

\begin{proof}  of Proposition \ref{prop:sym_lap_block_diag_multiarray}
The proof of Proposition \ref{prop:sym_lap_block_diag} can be generalized to multi-arrays.
\end{proof}

\begin{proof}  of Corollary \ref{cor:eigen_subproblem}
The result follows from Propositions \ref{prop:weighted_geval_extremal_representation} and \ref{prop:sym_lap_tsym}.
$K \le \widetilde{R} + \widetilde{C}$ ensures the assumption of Proposition \ref{prop:weighted_geval_extremal_representation} are satisfied for case 3. 
\end{proof}

\subsection{Block diagonal optimization problem solution sets}

\begin{proof} of Proposition \ref{prop:soln_relations_global}

Problems \eqref{prob:bd_constr} and \eqref{prob:bd_eval_constr} are equivalent by Proposition \ref{prop:sym_lap_block_diag}.

Suppose $(X, U)$ are such a global minimizer of \eqref{prob:bd_extremal_rep}. 
By Proposition \ref{prop:sym_lap_block_diag}, $NB(X) \ge B$ so $X$ satisfies the constraints of  \eqref{prob:bd_constr}.
Assume for the sake of contradiction that $Y$ is a better minimizer of \eqref{prob:bd_constr} i.e.  $f(Y) < f(X)$. 
Let $U_Y$ be the smallest $B$ generalized eigenvectors of $(L_{\text{un}}(A_{\text{bp}}(Y), \text{diag} (\text{deg} (A_{\text{bp}}(Y))))$.
Then
\begin{equation}
\begin{aligned}
&f(Y) + \alpha \text{Tr} \left( U_Y^T L_{\text{un}}(A_{\text{bp}}(Y)) U_Y\right)  =  f(Y) + \alpha \sum_{j=1}^B  \lambda_{(j)} \left( L_{\text{sym}}(A_{\text{bp}}(Y))\right) = f(Y) \\
& <  f(X) = f(X) + \alpha \sum_{j=1}^B  \lambda_{(j)} \left( L_{\text{sym}}(A_{\text{bp}}(X))\right) 
 = f(X) + \alpha \text{Tr} \left( U^T L_{\text{un}}(A_{\text{bp}}(X)) U\right).
\end{aligned}
\end{equation}
Thus $(Y, U_Y)$ is a better minimizer of  \eqref{prob:bd_extremal_rep} contradicting the fact that $(X, U)$ is a global minimizer.
\end{proof}

\begin{proof} of Proposition \ref{prop:soln_relations_local} 

Suppose $X$ satisfies the conditions of claim 1.
Then $NB(X) \ge B$ by Proposition \ref{prop:sym_lap_block_diag} so $X$ satisfies the constraints of \eqref{prob:bd_constr}.
Suppose $Y$ is a better solution for \eqref{prob:bd_constr}, then $ \sum_{j=1}^B  \lambda_{(j)} \left( L_{\text{sym}}(A_{\text{bp}}(Y))\right) = 0$ and $f(Y) < f(X)$.
Thus $Y$ is a better solution for \eqref{prob:bd_eval_pen}.

Suppose $X$ conditions of claim 2.
For fixed $X$, $U_X$ is a coordinate-wise minimizer of \eqref{prob:bd_extremal_rep}  if and only if the columns of $U_X$ are the smallest $B$ generalized eigenvectors of $(L_{\text{un}}(A_{\text{bp}}(X), \text{diag} (\text{deg} (A_{\text{bp}}(X))))$ and 
$$
 \text{Tr} \left( U_X^T L_{\text{un}}(A_{\text{bp}}(X)) U_X\right) = \sum_{j=1}^B  \lambda_{(j)} \left( L_{\text{sym}}(A_{\text{bp}}(X))\right).
$$
Note the row sum condition on $X$ ensures the kernel constraints of Proposition \ref{prop:weighted_geval_extremal_representation} hold.
Suppose $(Y, \widetilde{U})$ are a better solution to \eqref{prob:bd_extremal_rep}  than $(X, U_X)$. 
Then $(Y, U_Y)$ are also a better solution to \eqref{prob:bd_extremal_rep} than $(X, U_X)$.
But then $Y$ is a better solution to \eqref{prob:bd_eval_pen}.

\end{proof}

\subsection{Algorithm convergence}

This section applies Zangwill's global convergence theorem \citep{zangwill1969nonlinear} to prove Proposition \ref{prop:conv}.
Following \citet{sriperumbudur2009convergence}, a \textit{point-to-set} mapping $\mathcal{A}: \mathcal{X} \to 2^\mathcal{Y}$ assigns a subset $\mathcal{A}(x) \subseteq \mathcal{Y}$ to a point $x \in \mathcal{X}$.
A point-to-set mapping is \textit{closed} if $x_k \to x^*, y_k \to y^*, y_k \in \mathcal{A}(x_k)$ together imply $y^* \in \mathcal{A}(x^*)$; this is a generalization of continuity for functions.
A \textit{generalized fixed point} of $\mathcal{A}: \mathcal{X}  \to 2^\mathcal{X}$ is a point $x$ such that $x \in \mathcal{A}(x)$.

\begin{lemma} \label{lem:closed_composition}
Let $f: A \to B$ and $g: B \to C$ be closed, non-empty point-to-set maps.
If $B$ is sequentially compact then $g \circ f: A \to C$ is closed.
\end{lemma}
\begin{proof} of Lemma \ref{lem:closed_composition}
Let $a_k \to a$, $c_k \to c$ and $c_k \in g(f(a_k))$.  
Let $b_k \in f(a_k) \cap g^{-1}(c_k)$ where the inverse denotes the set of pre-images.
By assumption on $B$ there exists a convergent subsequence $\{b_{k_i}\}_{i=1}^{\infty}$ such that $b_{k_i} \to b$ for some $b \in B$.
Since $f$ is closed, $b \in f(a)$; since $g$ is closed $c \in g(b)$, therefore $c \in g(f(a))$.
\end{proof}

Let $\mathcal{A}(X, U)$ be the point-to-set map corresponding to Algorithm \ref{algo:wnn_alt_algo}.
$\mathcal{A} := \textsc{update-X} \circ  \textsc{update-U}$ where $\textsc{update}-U(X)$ solves Problem \eqref{prob:sym_lap_pen_extremal_rep} for fixed $X$ and update $\textsc{update}-X$ is either the full update  (Assumption \ref{assumpt:x_update}.1) or surrogate update (Assumption \ref{assumpt:x_update}.2.)

Let $\psi(X)$ be the objective function in \eqref{prob:sym_lap_eval_pen}  and let $\phi(X, U)$ be the objective function in \eqref{prob:sym_lap_pen_extremal_rep}.
Each step of Algorithm  \ref{algo:wnn_alt_algo} decreases these objective functions.
\begin{proposition} \label{prop:alt_min_decrease}
Let $(X^*, U^*) \in \mathcal{A}(X, U)$ then $\phi(X^*, U^*) \le \phi(X, U)$ and $\psi(X^*) \le \psi(X)$.
\end{proposition}

\begin{lemma} \label{lem:gen_fp_coord_stat}

Under Assumption \ref{assumpt:x_update}.1, if $(X, U)$ is a generalized fixed point of $\mathcal{A}$ then $(X, U) \in \mathcal{LG}$.

Under Assumption \ref{assumpt:x_update}.2, if $(X, U)$ is a generalized fixed point of $\mathcal{A}$ then $(X, U) \in \mathcal{SG}$.
\end{lemma}
\begin{proof} of Lemma \ref{lem:gen_fp_coord_stat}
By construction $U$ is a global minimizer of  \eqref{prob:sym_lap_pen_extremal_rep} for fixed $X$.

Under Assumption \ref{assumpt:x_update}.1, $X$ is a global minimizer of \eqref{prob:sym_lap_pen_extremal_rep} for fixed $V$ thus the first claim follows.

The constraints of \eqref{prob:wnn_x_update}/\eqref{prob:wnn_x_update_maj} are affine in $X$ so the KKT conditions are first order necessary.
Under Assumption \ref{assumpt:x_update}.2 if $X^*$ is a generalized fixed point of $\mathcal{A}$, $X^*$ is a minimizer of \eqref{prob:wnn_x_update_maj} and thus satisfies the KKT conditions for \eqref{prob:wnn_x_update_maj}.
Because $Q$ and $f$ have the same first order behavior by Assumption  \ref{assumpt:x_update}.2, a KKT point of  \eqref{prob:wnn_x_update_maj} is also a KKT point of \eqref{prob:wnn_x_update} and the second claim follows.
\end{proof}

\begin{proof} of Proposition \ref{prop:conv}
We apply Zangwill's global convergence theorem \citep{zangwill1969nonlinear} by checking the three conditions for Theorem 2 of \cite{sriperumbudur2009convergence}. 
Let $\phi(X, U)$ be the objective function of \eqref{prob:sym_lap_pen_extremal_rep} and let $\Gamma$ be the set of generalized fixed points of $\mathcal{A}$.

Let $S_{X^0}$ denote the compact set in Assumption \ref{assumpt:basic}.
By Propositions \ref{prop:weighted_geval_extremal_representation} and \ref{prop:sym_lap_tsym} the second term in the objective function of \eqref{prob:sym_lap_pen_extremal_rep} is upper bounded by $\alpha w^T \mathbf{1}_K$.
Thus, by Proposition \ref{prop:alt_min_decrease} the iterates $\{X^s\}_{s=0}^{\infty} \subseteq S_{X^0}$.

By assumption \ref{assumpt:degree_lower_bound}, without loss of generality we can add the constraint $\text{deg}(A_{\text{bp}}(X)) \ge  \eta \mathbf{1}_{R + C}$ to \eqref{prob:sym_lap_pen_extremal_rep}.
Since $\eta > 0$, the constraint given by \eqref{eq:c_diag} implies that at any solution $||U_k||_2^2 \le \frac{1}{\eta}$ where the inequality is applied element wise.
Thus the iterates $\{U^s\}_{s=0}^{\infty} \subseteq S_{\eta} := \{U | ||U_k||_2^2 \le \frac{1}{\eta}, k \in [K]\}$ which is a compact set.
We conclude $\{(X^s, U^s)\}_{s=0}^{\infty} \subseteq S := S_{X^0} \times S_{\eta}$ a compact set and condition (1) holds.

By construction of $\mathcal{A}$ and Proposition \ref{prop:alt_min_decrease}, condition (2) holds.

If $(X', U') = \mathcal{A}(X, U)$ then the constraint sets for the $U$ and $X$ update problem starting from $(X', U')$ are non-empty.
Therefore, the $X$ update and $U$ update steps are non-empty that compose to make  $\mathcal{A}$ are non-empty.
By the above discussion, the $U$ update always lives in the compact set $S_{\eta}$. 
Therefore, Lemma \ref{lem:closed_composition} shows $\mathcal{A}$ is closed and condition (3) holds.

Therefore, by Theorem 2 of \cite{sriperumbudur2009convergence},
all the limit points of $\{(X^s, U^s)\}_{s=0}^{\infty}$ are the generalized fixed points of $\mathcal{A}$ and $\lim_{s \to \infty} \phi(X^s, U^s) = \lim_{s \to \infty} \phi(X^*, U^*)$ where $(X^*, U^*)$ is some generalized fixed point of $\mathcal{A}$.

The result now follows from Lemma \ref{lem:gen_fp_coord_stat}.
\end{proof}

\begin{proof} of Proposition \ref{prop:soln_set_containment}
First implication follows from proof by contradiction.
Consider Problem \eqref{prob:sym_lap_pen_extremal_rep} where $U$ is fixed and the constraint set is non-empty.
For all such fixed $U$, the constraints on $X$ are affine thus Linear Constraint Qualification holds thus the KKT conditions are first order necessary.
\end{proof}

\subsection{Block diagonal MVMM}

\begin{proof} of Proposition \ref{prop:log_lasso_heuristic}
Let $f(x) = - a \log(x + \epsilon) + b x$. 
Then $f(x)$ is strictly convex on $(-\epsilon, \infty)$. 
Setting $ 0 = f'(x) = \frac{-1}{x + \epsilon} + b$ leaves us with $x = \frac{a}{b} - \epsilon$ for the unique stationary point of $f(x)$.
If $\frac{a}{b} - \epsilon > 0$, this must be the minimizer of \eqref{prob:meowmoewmoew}.
Otherwise, $x=0$ must be the minimizer.
\end{proof}

\bibliographystyle{apalike}
\bibliography{refs}

\end{document}